%% file: main.tex
\documentclass[10pt,twocolumn,letterpaper,table]{article}
\usepackage{usenix-2020-09}

\newif\ifusenixsty
\usenixstytrue
\newif\ifcameraready
\camerareadyfalse
\newif\ifmarkchanges
\markchangesfalse
\newif\ifremovecomments
\removecommentstrue
\newif\ifshowadditions
\showadditionsfalse
\newif\ifshowrr
\showrrfalse

\input{0_packages}
\input{0_config}
\input{0_macros}

\begin{document}
	
	\pagestyle{empty}

	\let\textcircled=\pgftextcircled
	
	\ifshowrr
		\input{texfiles/nsdi23_rr.tex}
		\clearpage
		\setcounter{section}{0}
		\renewcommand*{\theHsection}{chY.\the\value{section}}
		\setcounter{figure}{0}
	\fi
	
	\title{Solving Max-Min Fair Resource Allocations Quickly on Large Graphs }
	\author{
{\rm Pooria Namyar$^{\dag\ddag}$, Behnaz Arzani$^{\dag}$, Srikanth Kandula$^{\dag}$, Santiago Segarra$^{\dag\star}$,} \\  \rowcolor{whiteRowColor}
{\rm Daniel Crankshaw$^{\dag}$\thanks{The author contributed to this work while at Microsoft.}~, Umesh Krishnaswamy$^{\dag}$, Ramesh Govindan$^{\ddag}$, Himanshu Raj$^{\dag}$} \vspace{1mm}\\ \rowcolor{whiteRowColor}
$^{\dag}$Microsoft, $^{\ddag}$University of Southern California, $^\star$Rice University
}
	\date{}
	\maketitle
	
	\input{cr/abstract}
	\input{cr/intro_v4}

	\input{cr/motivation_v3}
	\input{cr/design_v4}

	\input{cr/evaluation}
	\input{cr/discussion}
	\input{cr/relatedworks}

	\input{cr/conclusions}

\vspace{0.05in}
\noindent
{\bf Acknowledgements:} We thank Ymir Vigfusson and the anonymous reviewers for their feedback on this paper. We also thank Luis Irun-Briz for his support of this work. This material is based upon work supported by Microsoft and the U.S. National Science Foundation under grant No. CNS-1901523.

	\clearpage
	\newpage
	\bibliography{cr/reference}
	\bibliographystyle{plain}
	
	\clearpage
	\input{cr/appendix}
	
\end{document}

%% file: 0_packages.tex
\usepackage{lmodern}
\usepackage{amssymb,amsmath,amsfonts}
\usepackage[T1]{fontenc}
\usepackage[utf8]{inputenc}

\usepackage{upquote}

\usepackage{microtype}
\UseMicrotypeSet[protrusion]{basicmath}
\usepackage{hyperref}
\usepackage{graphicx,grffile}

\usepackage{times} 
\usepackage{tikz}
\usepackage{endnotes,epsfig}
\usepackage{multirow}
\usepackage{xspace}
\usepackage{tabularx}
\usepackage{ragged2e}
\usepackage{booktabs}
\usepackage{paralist}
\usepackage[american]{babel}
\usepackage[shortlabels]{enumitem}
\usepackage{courier,color,wrapfig}
\usepackage{xspace}
\usepackage{balance}
\usepackage{enumitem}
\usepackage{epstopdf}
\usepackage{multirow}
\usepackage{booktabs}
\usepackage{url}
\usepackage{amssymb}
\usepackage{pifont}
\usepackage{subfigure}
\usepackage{tikz}
\usepackage{comment}
\usepackage{mathtools}
\usepackage[normalem]{ulem}
\usepackage{xkeyval}
\usepackage{cleveref}
\usepackage{xcolor}
\usepackage[short]{optidef}

\ifremovecomments
\usepackage[disable, textsize=small]{todonotes}
\else
\usepackage[textsize=small]{todonotes}
\fi

\usepackage{ifthen}
\usepackage{etoolbox}

\usepackage[margin=5pt,font=small,labelfont={bf,rm}]{caption}
\DeclareCaptionFormat{myformat}{#1#2#3\hrulefill}
\usepackage{sepfootnotes}
\usepackage[compact]{titlesec}
\usepackage{bbding}
\usepackage{enumitem}

\colorlet{whiteRowColor}{white}
\colorlet{rowColorOne}{gray!5}
\colorlet{rowColorTwo}{gray!20}
\rowcolors{2}{rowColorOne}{rowColorTwo}
\definecolor{maroon}{RGB}{128,0,0}

\usepackage{array}
\newcolumntype{L}[1]{>{\raggedright\let\newline\\\arraybackslash\hspace{0pt}}m{#1}}
\newcolumntype{C}[1]{>{\centering\let\newline\\\arraybackslash\hspace{0pt}}m{#1}}
\newcolumntype{R}[1]{>{\raggedleft\let\newline\\\arraybackslash\hspace{0pt}}m{#1}}
\usepackage{amsmath,amsthm}
\makeatletter
\def\thm@space@setup{%
	\thm@preskip=6pt plus 1pt minus 3pt
	\thm@postskip=\thm@preskip
}
\makeatother
\usepackage{bbm}
\usepackage[linesnumbered,boxed]{algorithm2e}

\makeatletter
\def\BState{\State\hskip-\ALG@thistlm}
\makeatother


%% file: 0_config.tex
\makeatletter
\def\maxwidth{\ifdim\Gin@nat@width>\linewidth\linewidth\else\Gin@nat@width\fi}
\def\maxheight{\ifdim\Gin@nat@height>\textheight\textheight\else\Gin@nat@height\fi}
\makeatother
\setkeys{Gin}{width=\maxwidth,height=\maxheight,keepaspectratio}
\newcommand\paraspace{\vspace*{1ex}}

\providecommand\parab[1]{\paraspace\noindent\textbf{#1}}
\providecommand\parae[1]{\paraspace\textbf{\textit{#1}}}

\apptocmd\normalsize{%
	\abovedisplayskip=5pt
	\abovedisplayshortskip=5pt
	\belowdisplayskip=5pt
	\belowdisplayshortskip=5pt
}{}{}


%% file: 0_macros.tex
\newcommand{\sysname}{{\small\sf Soroush}\xspace}

\newcommand{\systitle}{{\sf Soroush}\xspace}

\newcommand{\weightedroutingmatrix}{\mymatrix{\Gamma}}
\newcommand{\routingmatrix}{\mymatrix{A}}
\newcommand{\capacity}{\vect{c}}
\newcommand{\rate}{\vect{f}}
\newcommand{\fairsharevector}{\boldsymbol{\zeta}}
\newcommand{\fairshare}{\zeta}
\newcommand{\setflows}{\set{F}}
\newcommand{\setactiveflows}{\set{D}_a}
\newcommand{\setdemands}{\set{D}}

\newcommand{\ie}{\emph{i.e.,}\xspace}
\newcommand{\eg}{\emph{e.g.,}\xspace}
\newcommand{\etal}{\emph{et al.}\xspace}

\newcommand{\secref}[1]{\S\ref{#1}}
\newcommand{\figref}[1]{Fig.~\ref{#1}}

\newcommand{\tabref}[1]{Table~\ref{#1}}
\newcommand{\algoref}[1]{Alg.~\ref{#1}}
\newcommand{\eqnref}[1]{Eqn.~\ref{#1}}

\newcommand{\theoremref}[1]{Theorem~\ref{#1}}

\newenvironment{Itemize}%
{\begin{itemize}
\setlength{\leftmargin}{1em}%
\setlength{\itemsep}{0in}%
\setlength{\topsep}{-.1in}%
\setlength{\partopsep}{-.1in}%
\setlength{\parsep}{-.1in}%
\setlength{\parskip}{-0in}}%
{\end{itemize}}

{\begin{enumerate}
\setlength{\leftmargin}{1em}%
\setlength{\itemsep}{0in}%
\setlength{\topsep}{-.1in}%
\setlength{\partopsep}{-.1in}%
\setlength{\parsep}{-.1in}%
\setlength{\parskip}{-0in}}%
{\end{enumerate}}

\newcommand{\cut}[1]{}

\newcommand{\squishlist}{
	\begin{list}{$\bullet$}
		{ \setlength{\itemsep}{0pt}
			\setlength{\parsep}{3pt}
			\setlength{\topsep}{3pt}
			\setlength{\partopsep}{0pt}
			\setlength{\leftmargin}{1.5em}
			\setlength{\labelwidth}{1em}
			\setlength{\labelsep}{0.5em} } }
	
	\newcommand{\pgftextcircled}[1]{
		\setbox0=\hbox{#1}%
		\dimen0\wd0%
		\divide\dimen0 by 2%
		\begin{tikzpicture}[baseline=(a.base)]%
			\useasboundingbox (-\the\dimen0,0pt) rectangle (\the\dimen0,1pt);
			\node[circle,draw,outer sep=0pt,inner sep=0.1ex] (a) {#1};
		\end{tikzpicture}
	}

	\newcommand{\squishend}{
\end{list}  }

\ifmarkchanges
	
	\newcommand{\crdel}[1]{\sout{\textcolor{red}{#1}}}
\else
	\ifshowadditions
		
	\else
		
	\fi
	\newcommand{\crdel}[1]{}
\fi

\newcommand{\set}[1]{\mathcal{#1}}

\newcommand{\vect}[1]{\mathbf{#1}}
\newcommand{\indicator}[1]{\mathbbm{1}\left[{#1}\right]}
\newcommand{\norm}[1]{\lvert{#1}\rvert}
\newcommand{\mymatrix}[1]{\boldsymbol{\mathit{#1}}}

\newtheorem{theorem}{Theorem}

\newtheorem{lemma}{Lemma}

\SetCommentSty{algcmt}
\SetKwComment{tcp}{}{}
\newcommand{\algcomment}[1]{\tcp*{#1}}

%% file: texfiles/nsdi23_rr.tex
<<<<<<< HEAD

=======
>>>>>>> 5ddc071ac21189a5b3673287b9b8fe83dffd3947
\section*{Review Response}

\newcommand{\quotereview}[1]{\medskip\begin{center}\fbox{\parbox{0.40\textwidth}{\it {#1}}}\end{center}\medskip}

We thank the reviewers for their time and insightful feedback. We substantially overhauled our submission and unfortunately marking up the edited portions will mark up almost the entire text. Thus, instead, we summarize our key changes below and point to specific changes in the submission when responding to the meta review and individual reviews~(\secref{rr:meta} and~\secref{rr:individual} respectively).


\subsection{Summary of Key Changes}
\label{rr:change_list}
Here are the list of major changes:
\begin{Itemize}
	\item Since the previous submission, one of our allocators -- {\sf GeoBinner}-- has been deployed as the default in a large cloud provider's traffic engineering pipeline. We report results from production. See~\secref{rr:production}.
	\item We conduct new experiments to compare with NCFlow and POP. Note that neither consider max-min fair allocations explicitly. NCFlow, in fact, only considers maximizing the total flow and has poor performance when used to obtain fair allocations. POP does better but has no worst-case approximation guarantee unlike our allocators. More details are in~\secref{sss:mr_c2}.
	\item We have simplified and improved the description of our graph-based model for multi-resource max-min fair allocation problems. We also explain which problems fit (and do not fit) in this model. See~\secref{sss:mr_c4}. 
	\item We have simplified and rewritten the design section to better connect the different allocators. See~\secref{sss:mr_c1}.
	\item We have rewritten portions of the paper to prove the available guarantees (or lack thereof) of our allocators. See~\secref{sss:mr_c5}.
	\item We intuitively explain why our single-shot LPs are faster than a sequence of LPs. See~\secref{sss:mr_c6}.
	\item We added new examples and pseudocode to more carefully explain our generalization of waterfilling to the multipath case (see~\autoref{f:example:water} and the algorithms in~\secref{s:genwaterfilling}).
\end{Itemize}

\subsection{Production Deployment}
\label{rr:production}
We have successfully integrated \sysname into a production traffic engineering controller that manages a public cloud's wide area network with over $1000$ nodes and edges. 
Our measurements over the past few months, on thousands of unique traffic demand matrices from production, demonstrate up to $5.4\times$ speed up ($2.4\times$ on average) compared to the previous production iterative max-min fair solver. These results are in~\figref{fig:traffic-engineering:integration} in \secref{subsec:TE}.

\subsection{Meta Review}
\label{rr:meta}

\subsubsection{Comment 1}
\label{sss:mr_c1}
\quotereview{Unnecessary formalism in \secref{s:unified-formulation} \& \secref{sec::opt}.}

We have refactored the formalism so that only essential details remain in the main body of the paper.

\parab{$\blacktriangleright$ Formalism in \secref{s:unified-formulation}.}
To clarify, all our optimization-based allocators (EB and GB) share the ${\sf FeasibleAlloc}$ constraints that ensure their allocations meet capacity and rate constraints as well as other affine functions such as heterogeneous utilities and differing resource requirements. The full formulas are now in appendix (A.1) and our model in~\secref{s:unified-formulation} only explains the concerns and intuition in text.

\parab{$\blacktriangleright$ Formalism in \secref{sec::opt}.} We rewrote this portion to clarify the connection between the single-shot exact formulation and the approximate variants. 

To clarify, both our optimization-based allocators (EB and GB) use the $\epsilon$ rewarding from the exact single-shot formulation. Furthermore, GB and EB address the limitations of the exact single-shot method in different ways: GB uses fixed bin boundaries that are geometrically increasing (\ie $[0, U), [U, \alpha U), \ldots [U\alpha^{b-1}, U\alpha^b), \ldots$) and hence does not require sorting network constraints and achieves the same $\alpha$-approx fairness guarantee as SWAN.  EB allows bin boundaries to dynamically evolve within the optimization but, again  explicitly orders them, (\ie $\ell_{b-1} \leq \ell_{b}$ in~\secref{s:comb_ext}) and therefore also does not need sorting network constraints. EB is empirically fairer but analyzing the worst-cast guarantee is more involved.

Also, to clarify, our generalized waterfiller allocators owe nothing to the single-shot exact formulation. They are parallelizable combinatorial algorithms as shown in Algo.~\ref{alg:one_path_approx_wt_waterfilling} and~\secref{s:genwaterfilling}.

\subsubsection{Comment 2}
\label{sss:mr_c2}
\quotereview{Comparison with NCFlow~\cite{Abuzaid-ncflow} \& POP~\cite{Narayanan-POP}}
POP and NCFlow use decomposition techniques to scale TE and CS. But neither explicitly support max-min fair resource allocation. NCFlow only maximizes the total flow and the authors mention in the paper that it is hard to extend it to max-min fairness objective. Similarly the POP paper maximizes total flow and maximum concurrent flow (\ie the smallest fractional allocation) but does not provide any results on max-min fairness objective. 

We conduct new experiments that apply POP for the max-min fairness objective. Specifically, we created new POP-inspired variants for both the exact max-min fair (Eqn.~\ref{eq:seq-max-min-flow}) and the SWAN-like approximate max-min fair~(Eqn.~\ref{eq:seq-approx-max-min-flow}) optimizations. In these adaptations, we use the same partitioning as recommended by the POP paper. These POP-ped variants do not have any worst-case guarantees on fairness (unlike GB) and empirically perform poorly relative to our new allocators. Results for POP-ped variants on two different topologies, two load factors and two traffic distributions are in \secref{s:exp:other} (\figref{fig:decomp:pop}) and \secref{sec:POP-eval-extended} (\figref{fig:decomp:pop:app}). Our results suggest the quality of POP's solution depends on the traffic distribution, whereas \sysname consistently performs better in these experiments. Fundamentally, POP only achieves per partition max-min fairness which is different from and can be arbitrarily worse than the global max-min fairness. We discuss this further in the paper.

\subsubsection{Comment 3}
\label{sss:mr_c3}
\quotereview{Multi-resource Terminology Clarification}
We have updated the model in~\secref{s:unified-formulation}, introduction and many other places to clarify this issue.

In short, the multiple resources are different links in traffic engineering problems and different kinds of resources -- such as GPUs, CPUs, memory-- on different servers in the cluster scheduling case.  Notice that a demand must receive rate {\em jointly} on a group of resources which we will call a path. For example, a worker task must receive some GPU, CPU and memory to execute. Furthermore, a demand can receive allocations on multiple paths, \ie different {\em groups} of {\em resources}. In cluster scheduling, worker tasks can be assigned to different servers.

Our model subsumes both of the above cases in an intuitive graph-based model.

Waterfilling does not apply to the multipath case.


\subsubsection{Comment 4}
\label{sss:mr_c4}
\quotereview{Balancing TE and CS in the main body.}
We only discuss TE and CS separately in~\secref{s:unified-formulation} and in~\secref{s:comb_ext}. The core of the allocator description only references the common model and will apply to any problem (including TE and CS) that can be specified in our graph-based model. 

In~\secref{s:comb_ext}, we also added how our core allocators can extend to handle weighted max-min fairness and specific variants that arise in CS such as resource requirements not being the same across all resources.


\subsubsection{Comment 5}
\label{sss:mr_c5}
\quotereview{Clarifying theoretical guarantees.}
We modified all the figures and claims to clarify that \sysname {\em empirically} pareto-dominates the baselines (see~\figref{fig:intro:qualitative-comparison} for an example). In the introduction (\secref{sec:intro}), we have (a) added~\tabref{tab:alloc_prop} to summarize all the theoretical and empirical results, and (b) updated~\figref{fig:intro:qualitative-comparison} to explicitly show empirical Pareto-dominance and theoretical guarantees.

Theoretically, GB has a worst-case fairness guarantee. AdaptiveWaterfiller comes close but our guarantees are weaker. EB likely has stronger guarantees and our empirical findings echo this but mathematical analysis is more involved.

\subsubsection{Comment 6}
\label{sss:mr_c6}
\quotereview{Single vs Multiple LPs}
In~\secref{s:motivation}, we added a paragraph to explain why it is faster to solve a single optimization compared to multiple. We also include~\secref{ss:expected-run-time} in which we analyze the expected speed up of \sysname~---~the speed up depends on: (1) the number of optimizations in the iterative approach, and (2) the number of variables in the single optimization compared to those in each optimization of the iterative solution. Since \sysname only solves one optimization and adds only a small number of variables, it is faster. 

\subsection{Other Individual Reviews}
\label{rr:individual}

\subsubsection{Reviewer A}
\quotereview{Confusing notation in~\figref{fig:intro:qualitative-comparison}}
We updated~\figref{fig:intro:qualitative-comparison} and replaced the notations.

\quotereview{"rank order of rates"}
Fixed, thanks!

\quotereview{say earlier that the range of problems you can solve is graph-based multi-resource allocations}
Clarified in the first sentence of~\secref{s:unified-formulation} and the title of the paper.

\subsubsection{Reviewer B}
\quotereview{Theorem 1}
Fixed, thanks!

\quotereview{Running allocators in parallel}
We have added a new paragraph in \secref{s:sysname-overview} to acknowledge that we can run these allocators in parallel. However, we believe that in cases where we have strict run-time requirements, running them in parallel may result in violations since we need to wait for all the allocators to finish.

\quotereview{Fairness metric}
To clarify, this is an offline fairness metric. We have computed the optimal max-min fair allocation for each scenario either using Danna (optimal in TE) or Gavel with waterfilling (optimal in CS). Then, we have measured the distance between each allocations obtained from each method and the optimal max-min fair allocation. We have added a sentence in the evaluations to mention this.

\quotereview{Bars in~\figref{fig:traffic-engineering:all-samples:efficiency}}
In~\figref{fig:traffic-engineering:all-samples:efficiency:low}, the bars are small because all the demands are lightly loaded, and therefore, all the algorithms can satisfy almost all the demands (note that fairness is also close to 1 for these cases). We added a sentence in the caption of the figure to clarify this.

\subsubsection{Reviewer C}

\quotereview{Bar in~\figref{fig:traffic-engineering:all-samples:efficiency:low}}
In~\figref{fig:traffic-engineering:all-samples:efficiency:low}, the bars are small because all the demands are lightly loaded, and therefore, all the algorithms can satisfy almost all the demands (note that fairness is also close to 1 for these cases). We added a sentence in the caption of the figure to clarify this.

\quotereview{\figref{fig:traffic-engineering:example} \& \figref{fig:cluster-scheduling:example}}
These figures are intended to demonstrate the trade-off (including the wall-clock run-times) for one example traffic scenario. We have modified the captions to make this clear.

\quotereview{CS experiment details}
We added more details in~\secref{sec::CS_eval_extended} including the various type of jobs we used and how we generate them.

\subsubsection{Reviewer D}

\quotereview{Practical Value}
We would like to thank the reviewer for sharing their experience and their concern regarding practicality and potential simplifying assumptions that can make the work different from reality. However, we want to emphasize that we demonstrated the practical value of our work by integrating it into one of the largest cloud operator's TE controller. And it has improved runtime by $2.4\times$ on average without any impact on fairness or efficiency.

\quotereview{ML in Traffic Engineering}
In~\secref{s:related-work}, we added an explanation on why it is not easy to extend recent ML-based TE solutions~\cite{DOTE,LearningToRoute,Teal} to max-min fairness.

To sum, recent work in ML for TE focuses on objectives such as max flow that require just one LP. Training with stochastic gradient requires that the objectives be convex or quasi-convex~\cite{DOTE}. We are unaware of any work that considers end-to-end training on a sequence of LPs, which is the exact form for max-min fairness. In fact, it may be more tractable for ML+TE methods to learn our GeoBinner which is a single LP and likely convex. All that said, there is no direct counterpart in ML+TE to compare with {\sysname}. We defer to future work to explore how ML can benefit max-min fairness and \sysname. We believe combinatorial algorithms with known guarantees may still have some advantage over unexplainable ML-based solutions.

\subsubsection{Reviewer E}
We have addressed all the reviews as part of our answer to Meta review but we want to clarify the properties of GB.

\parab{Geometric Binner (GB)} is an approximation method that offers per flow fairness guarantees: this means (and we proved) it produces the same solution as SWAN (which Microsoft has been using for a long time) and it does so while being fundamentally faster. \sysname also includes other methods such as EB that are empirically fairer but that does not imply GB is not fair (it always meets its $\alpha$-approximate guarantee in practice).

%% file: cr/abstract.tex
\noindent
{\bf Abstract--} 
We consider the max-min fair resource allocation problem.
The best-known solutions use either a sequence of optimizations or waterfilling, which only applies to a narrow set of cases. 
These solutions have become a practical bottleneck in WAN traffic engineering and cluster scheduling, especially at larger problem sizes.
We improve both approaches: (1) we show how to convert the optimization sequence into a single fast optimization, and (2) we generalize waterfilling to the multi-path case. 
We empirically show our new algorithms Pareto-dominate prior techniques: they produce faster, fairer, and more efficient allocations.
Some of our allocators also have theoretical guarantees: they trade off a bounded amount of unfairness for faster allocation. 
We have deployed our allocators in Azure's WAN traffic engineering pipeline, where we preserve solution quality and achieve a roughly $3\times$ speedup.

%% file: cr/intro_v4.tex
\section{Introduction}
\label{sec:intro}

Multi-resource fair allocators have become essential for cloud operators as multi-tenancy, availability, and efficiency grow in importance.
These allocators divide the resources fairly among different requests (applications, users, or network flows). 
Operators use them to meet customer expectations, especially during congestion and for network neutrality.

Recent works present fair allocators in settings such as WAN traffic engineering~\cite{danna-practical-max-min,SWAN,B4,Umesh-blastshield} and GPU scheduling~\cite{Gandiva-fair-Chaudhary,Allox-Tan,Gavel-Deepak}. We show these allocators achieve fairness at the cost of {\em speed} (crucial for maintaining high utilization as loads change~\cite{Abuzaid-ncflow}) and {\em efficiency}\footnote{In this paper, we use efficiency and utilization interchangeably.} (essential for profit).

We aim to achieve a better balance between fairness, efficiency, and speed, and our novel algorithms offer operators greater flexibility to control the trade-off between them. We focus on max-min fairness~---~where we cannot allocate more to one request without reducing the allocation of another with {\em an equal or smaller value}~---~because it is simple, commonly used in practice~\cite{SWAN,B4,onewan,Gavel-Deepak}, and can promote efficiency.\footnote{We defer extending to other notions of fairness to future work.}

\begin{figure}
	\centering
	\includegraphics[width=0.95\linewidth]{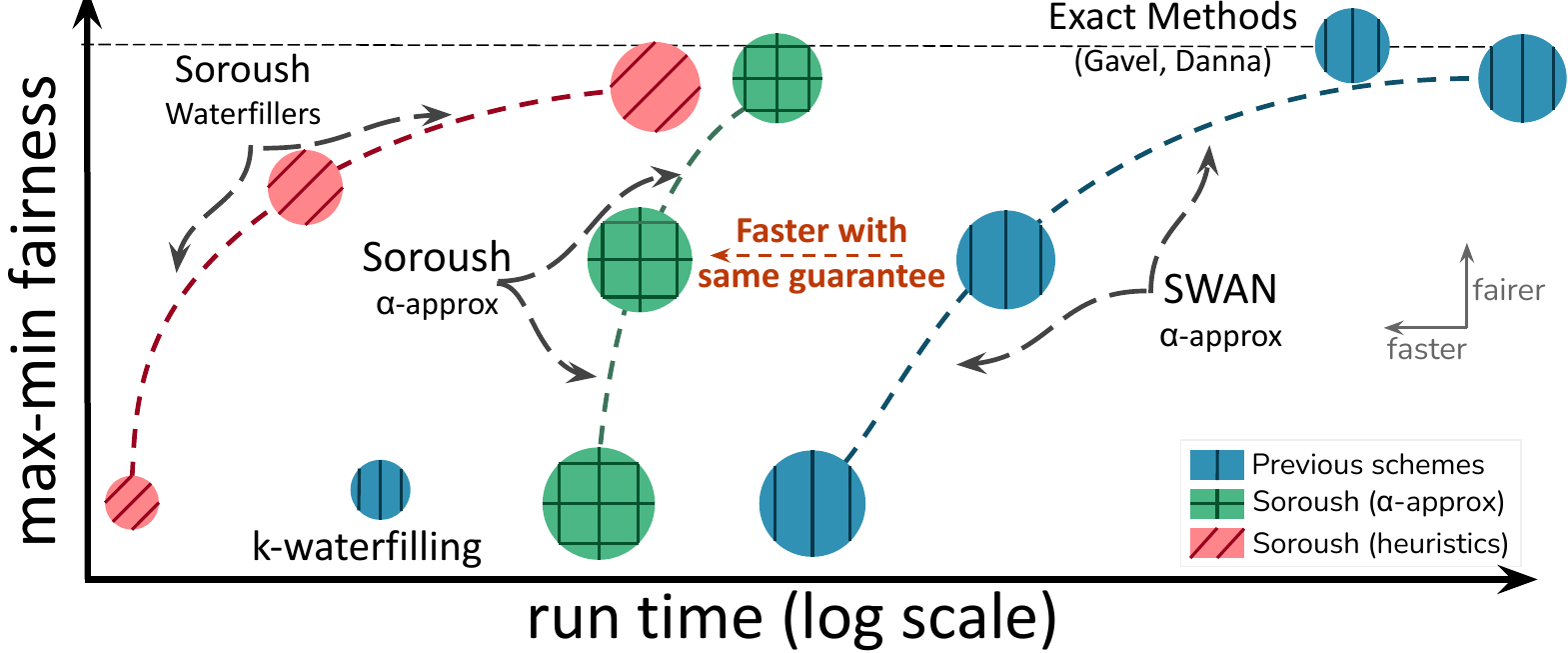}
	\caption{\textbf{Comparing the new allocators with state-of-the-art.} {\sysname} offers parameterizable max-min fair allocators. The axes are fairness and speed; the marker size corresponds to efficiency (larger is more efficient). Our new allocators empirically Pareto-dominate other schemes, and some of them have theoretical guarantees on fairness (\secref{sec:eval}).}
	\label{fig:intro:qualitative-comparison}
\end{figure}

The definition of max-min fairness naturally leads to iterative solutions that prioritize smaller requests over larger ones and assign rates in order from smallest to largest.  In multi-resource settings, these solutions solve either mixed-integer or linear optimizations~\cite{Dritan-Dual-blocking,Pioro-Efficient-Max-Min-Fair} at {\em each} step. Their scalability depends on the size of each individual optimization and the number of iterations~--~typically a function of the number of resources and requests.


Operators need to invoke resource allocators when workloads change or failures occur. However, today's exact~\cite{danna-practical-max-min,Gavel-Deepak} or approximate solutions (\eg that trade-off fairness for speed~\cite{SWAN}) are too slow in reacting to these events at the production scale. They take \textit{tens of minutes to hours} (\secref{sec:eval}) on WANs with 100s of routers that serve millions of flows or clusters with 1000s of jobs.


We ask: are {\em iterative} optimizations necessary for max-min fair resource allocation? Max-min fair algorithms must maximize smaller allocations before assigning more capacity to larger ones. Current solutions iterate because they do not know the sorted order of these rate allocations apriori. One of our ideas is to (1) use sorting networks~\cite{Sorting-Networks-Batcher} to discover the sorted order of max-min fair allocations {\em dynamically within the optimization} and (2) use a linear weighted objective that explicitly incentivizes the optimization to allocate more rates to requests with smaller indices in the sorted order. The result is a single-shot optimization for max-min fair allocation.

The above single-shot optimization is not always practical as modeling sorting networks within an optimization can significantly increase its size, and
the linear weighted objective can cause double-precision issues when there are many requests.
To develop a practical solution, we combine this idea with an approximate max-min fair allocator from SWAN~\cite{SWAN}. This combination results in a new allocator, {\sf GeometricBinner} (or {\sf GB}), which is fundamentally faster than SWAN, does not need sorting networks, has no double precision issues, and provides the same fairness guarantees as SWAN.


Waterfilling-based algorithms~\cite{Bertsekas-data-netwroks} can be faster than black-box optimizations, but they are specialized for cases where each request seeks rates on a single path.
In a broad class of problems~\cite{Gandiva-fair-Chaudhary,danna-practical-max-min,SWAN,B4,Umesh-blastshield,Allox-Tan,Gavel-Deepak}, requests ask for allocation from multiple paths. 
In these cases, the global fair share is not locally fair along each path, and waterfilling does not apply.
Our solution, {\sf AdaptiveWaterfiller} (or {\sf AW}), extends waterfilling to multi-path settings. 
{\sf AW} is faster than {\sf GB} but does not have a worst-case fairness guarantee.
We prove {\sf AW} can converge to a small set of allocations that contain the optimal.

Our third algorithm, which is empirically the fairest, combines the above approaches.
We apply {\sf GB} with one change: use the allocations from {\sf AW} to spread requests more uniformly among bins (instead of the fixed geometrically increasing bin sizes in {\sf GB}). This allocator, {\sf EquidepthBinner} (or {\sf EB}), is slower than both {\sf AW} and {\sf GB} (executes each once) but intuitively improves fairness for the same reason that equi-depth binning reduces histogram approximation error~\cite{jagadish1998optimal}.

\sysname\footnote{Our code is available at \href{https://github.com/microsoft/Soroush}{https://github.com/microsoft/Soroush}} is the collection of these algorithms, each providing a different trade-off among fairness, efficiency, and speed. Operators can use our
simple decision process to choose the allocator (and its hyper-parameters) that achieves their desired trade-off. \tabref{tab:alloc_prop} lists our allocators, their theoretical and empirical properties, and their parameters.

\begin{table}[t]
	\scriptsize
	\centering
	\begin{tabular}{p{0.5in}ll} \hline
		\textbf{\textit{Allocator}} & \textbf{\textit{Properties}} & \textbf{\textit{Parameters}} \\ \hline \hline
		\textbf{Geometric} & ({\bf T}) $\alpha$-approx fairness guarantee & $\alpha$\\ 
		\textbf{Binner} & ({\bf E}) Faster than other $\alpha$-approx methods & $\epsilon$ \\ \hline \hline
		\textbf{Adaptive} & ({\bf T}) Solution in a small set containing optimal & \#iterations\\
		\textbf{Waterfiller} &  ({\bf E}) Fastest & \\ \hline\hline
		\textbf{Equi-depth} & ({\bf T}) Better than Adaptive Waterfiller & \#iterations\\ 
		\textbf{Binner} & ({\bf E}) Fairest and fast & \#bins, $\epsilon$ \\ \hline\hline
	\end{tabular}
	\caption{The {\sysname} allocators, their properties ({\bf T}heoretical and {\bf E}mpirical), and their parameters. These allocators provide different trade-offs among fairness, efficiency, and speed.\label{tab:alloc_prop}}
	
\end{table}


To show \sysname is general, we introduce a graph model for multi-resource, max-min fair resource allocation problems where edges model resources and paths capture groups of resources the allocator must assign together. Requests (demands) can then ask for resources on any choice of multiple paths. This compact and general model subsumes problems from at least two domains: traffic engineering (TE) and cluster scheduling (CS). \sysname can solve any future max-min fair allocation problem if the user can specify it in this model.

Our extensive evaluation in both TE and CS (which we summarize in~\figref{fig:intro:qualitative-comparison}) show the new allocators Pareto-dominate the state-of-the-art in fairness, speed, and efficiency.

We deployed {\sf GeometricBinner} in the production TE pipeline at Azure where it provides a $2.4\times$ average speedup (up to $5.4\times$ in some cases) without any impact on fairness and efficiency compared to the previous allocator.


%% file: cr/motivation_v3.tex
\section{Motivation and Overall Approach}
\label{sec:design}
\label{s:motivation}

\begin{figure}[t]
	\centering
	\includegraphics[width=0.9\linewidth]{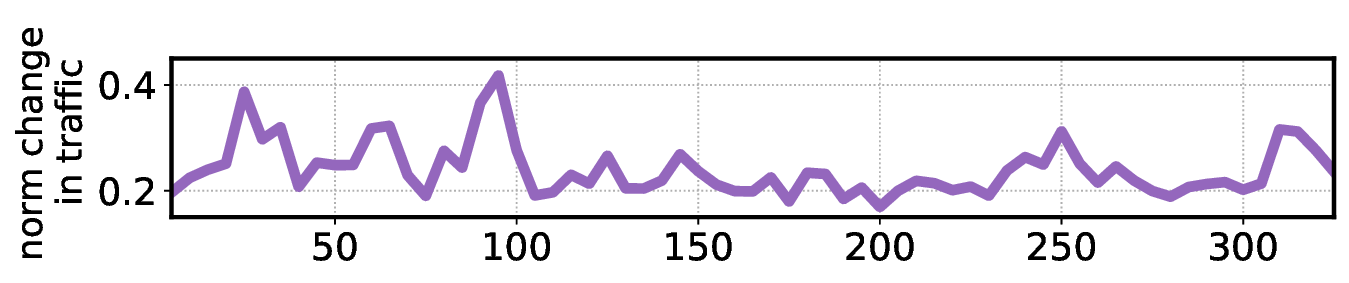}
	\includegraphics[width=0.9\linewidth]{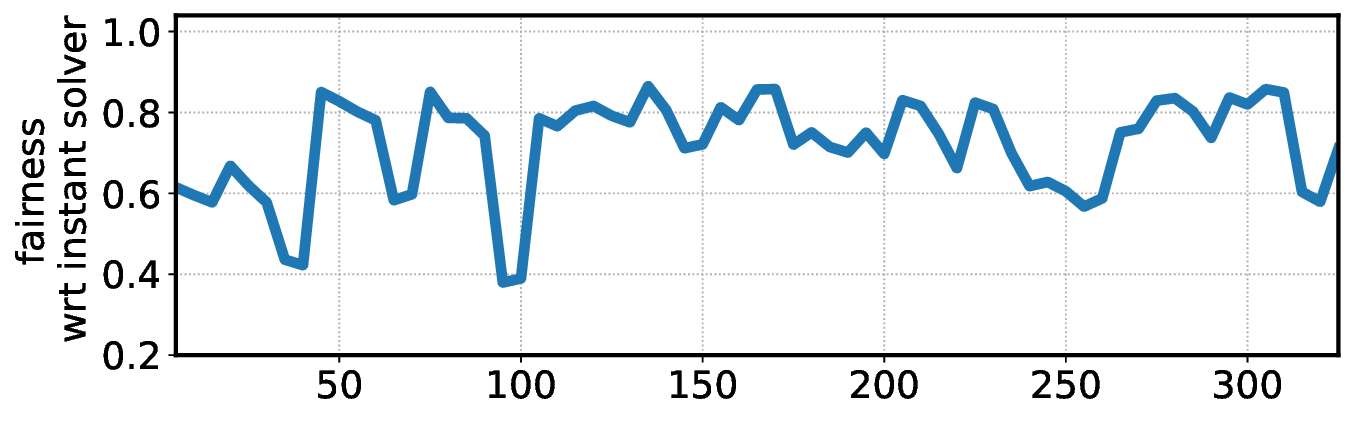}
	\includegraphics[width=0.9\linewidth]{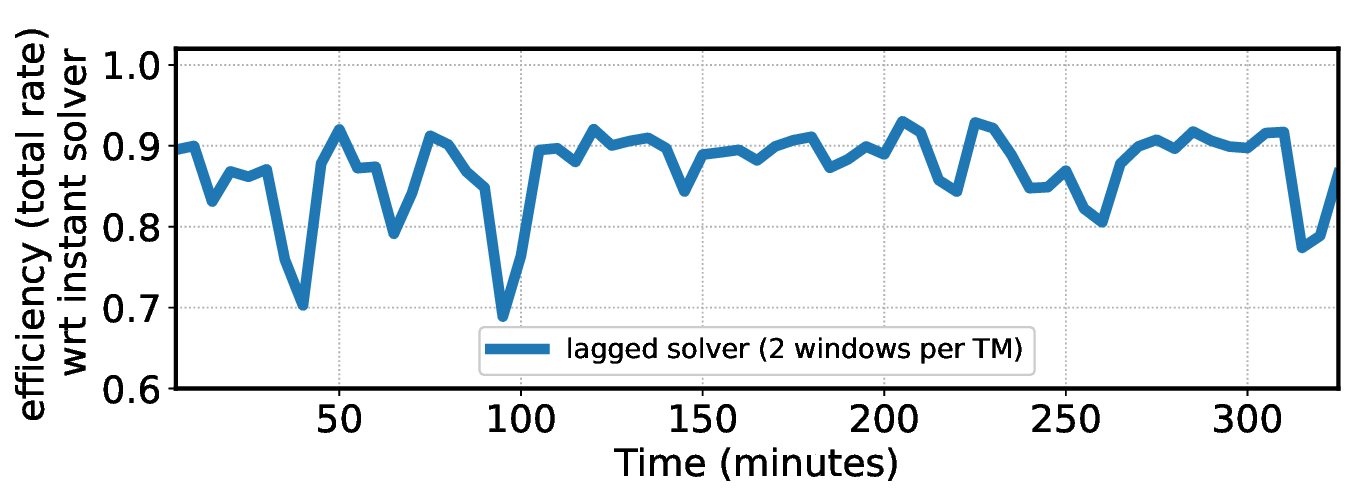}
	\caption{\textbf{Slow max-min fair resource allocators cause under-utilization and unfairness.} We compare two instances of the SWAN solver~--~one computes the allocations instantly while the other needs two windows~--~on a $5$-hour trace from Azure's WAN. The results indicate a large gap between the two solutions in fairness and efficiency.}
	\label{fig:example-slow-max-min-fair}
\end{figure}
Faster workload dynamics~\cite{Abuzaid-ncflow} and higher availability requirements~\cite{FFC} have made fast resource allocation a necessity. Operators of multi-tenant clouds further require solutions that ensure fairness and maintain high efficiency~\cite{SWAN, B4}. Prior work (\cite{Abuzaid-ncflow, FFC, Narayanan-POP, Singh-RADWAN} in TE or~\cite{Gavel-Deepak} in CS) fails to meet one or more of these requirements.

Efficient and fair solvers~\cite{danna-practical-max-min, SWAN, Gavel-Deepak} cannot adapt quickly to conditions that frequently change. Some production environments~\cite{Abuzaid-ncflow} use the most recent previous allocation when the solver cannot allocate resources within a fixed time window.
This is sub-optimal: nodes that increase their demands in the new window do not get enough resources, and others who request less may receive more than they need.

\begin{figure}[tb]
	\centering
	\subfigure{\includegraphics[width=0.47\linewidth]{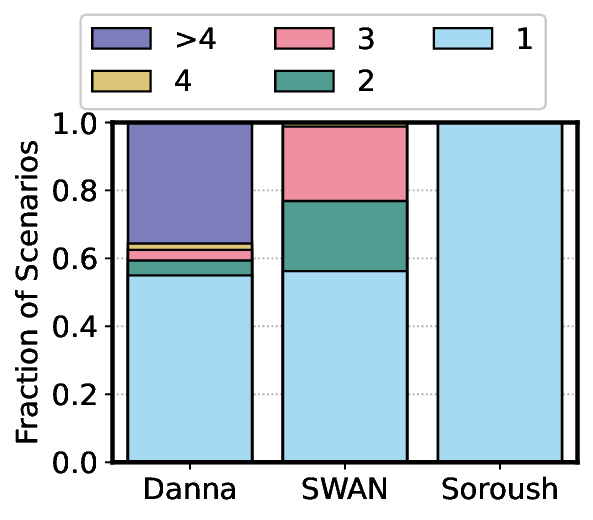}} \hfil
	\subfigure{\includegraphics[width=0.37\linewidth]{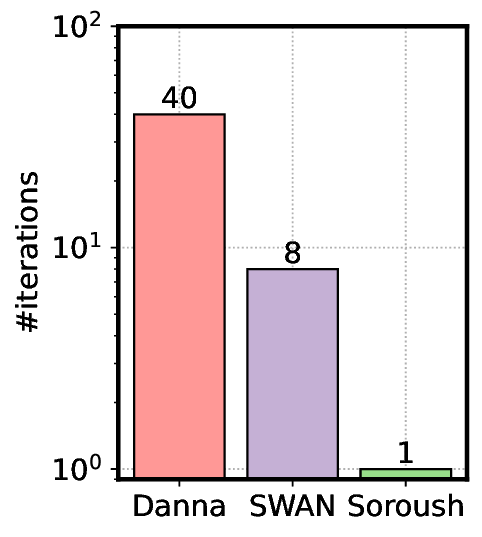}}
	\caption{\textbf{State-of-the-art methods can not keep up with frequently changing demands.} We capture the number of windows (left) and the number of iterations (right) each approach needs. To keep up with demands, they must finish within a single 5-minute window. However, SWAN~\cite{SWAN} and Danna \etal~\cite{danna-practical-max-min} often need more than one window and miss their deadline. The results are on a topology with $\sim$200 nodes and $\sim$500 edges. Left captures 160 different scenarios. Right is a highly loaded scenario~\cite{Abuzaid-ncflow} -- the results hold across all the algorithms in \sysname.}
	\label{fig:dem-change}
\end{figure}

We quantify the impact of this strategy in the TE setting using a $5$-hour trace from Azure's production WAN (\figref{fig:example-slow-max-min-fair}), which uses a $5$-minute window. We observe a solver that needs two windows ($10$~minutes) to allocate resources reduces fairness by $20$\%~--~$60$\% and efficiency by $10$\%~--~$30$\% relative to a solver that completes within one window. 

How often do solvers miss their deadline? We use traces from~\cite{Abuzaid-ncflow} to show the distribution of the number of windows an exact solver (Danna \etal~\cite{danna-practical-max-min}) and Microsoft's approximate solver (SWAN~\cite{SWAN}) need to compute max-min fair allocations. For nearly half of the traffic trace, these solvers exceed the $5$-minute window and often need $2$ to $3$ windows to finish (\figref{fig:dem-change}, left). This is because these approaches invoke expensive optimizations multiple times (\figref{fig:dem-change}, right). 

\sysname invokes at most one optimization and always completes within a single window. Whether a one-shot optimization is faster than an iterative approach that solves multiple optimizations depends on two factors: (a) the number of optimizations in the iterative approach, and (b) the size of the optimization in the one-shot approach compared to those in the iterative solution\footnote{LP solver latency is polynomial in the problem size~\cite{boyd2004convex,solvingLP}.}. Our one-shot optimizations are faster than previous solutions~\cite{danna-practical-max-min, SWAN, Gavel-Deepak} because we only add a small number of variables to convert the problem into one that can be solved in one shot, and we avoid the overhead in solving multiple optimizations. See~\secref{ss:expected-run-time} for a detailed analysis.

While examples here are from TE, CS resource allocators are similar~\cite{Gavel-Deepak}. We omit the details for brevity.

\subsection{Our model}
\label{s:unified-formulation}
We model the max-min fair resource allocation problem as a capacitated graph\footnote{We present the formulation of this model in~\secref{sec:general-form}}. Each {\em edge} represents a different resource, and {\em edge capacities} show the amount of available resources. {\em Paths} in the graph encode a collection of resources we must allocate together (\eg GPU and memory), and {\em demands} can request resources on a subset of these paths. Our model also supports other affine constraints over graph variables (\eg \textcolor{maroon}{text in maroon} below). \sysname solves any max-min fair resource allocation problem we can specify in this model.

Our model takes as input:
\vspace{-0.05in}
\begin{Itemize}
\item A set of resources $\mathcal{E}$, each with capacities $c_e, e \in \mathcal{E}$.
\item A set of paths $\mathcal{P}$ where each path is a group of dependent resources that we must allocate together.
\item A set of demands $\mathcal{D}$ where each demand $k \in \mathcal{D}$:
	\begin{Itemize}
	\item Requests some rate $d_k$.
	\item Has weight $w_k$ (for weighted max-min fairness).
	\item Can be routed over a set of multiple paths $P_k \in \mathcal{P}$.
	\item \textcolor{maroon}{Consumes $r_k^e$ of the capacity on edge $e$ for each unit rate we assign.}
	\item \textcolor{maroon}{Has utility $q_k^p$ on path $p$ for each unit rate.}
	\end{Itemize}
\end{Itemize}

A max-min fair allocator assigns rates to demands such that the weighted ratios $\{\frac{f_k}{w_k}\}$ are max-min fair: to increase the allocation of any demand, we have to reduce the allocation of another demand with a smaller ratio.

{\tiny
	\begin{table}[t!]
		\centering
		{\footnotesize
			\begin{tabular}{l p{2.6in}}
				\renewcommand{\arraystretch}{1.3}
				{\bf Term} & {\bf Interpretation}\\    \hline \hline
				$\mathcal{E}, \mathcal{D}, \mathcal{P}$ & sets of resources, demands and paths \\\hline
				$c_e$ & capacity of resource $e \in \mathcal{E}$\\\hline
				$\mathbf{f}, f_k^p$ & rate allocation vector and rate for demand $k$ on path $p$\\\hline
				$d_k, w_k$ & requested rate and weight for demand $k \in \mathcal{D}$\\\hline
				$r^e_k, q^p_k$ & scaling resource consumption and rate utility for demand $k$\\\hline
			\end{tabular}
		}
		\caption{Notation for our max-min fair resource allocation model. (more details in \tabref{t:notation-general})\label{t:notation}}
	\end{table}
}

\begin{figure}[t!]
	\centering
	\includegraphics[width=0.8\linewidth]{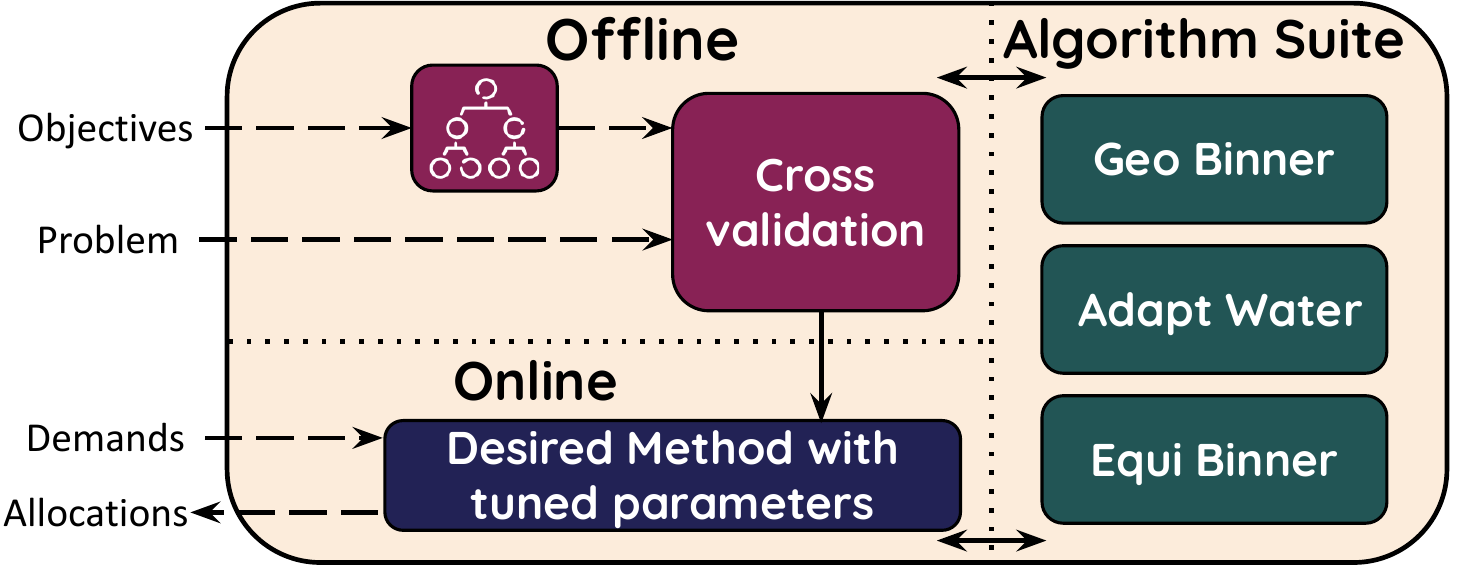}
	\caption{Choosing allocators (and their parameters).}
	\label{fig:design-overview}
\end{figure}


We can use this model to specify max-min fair allocation problems in traffic engineering (TE)~\cite{SWAN,B4} and cluster scheduling (CS)~\cite{Gavel-Deepak, drf, tetris}. 

\parab{TE.} The actual links in the network are the resources, and demands are services that require a specific rate between nodes in the network. The TE scheme picks the paths for each demand. Weights describe how the operator wants to divide the rates (\eg split rate between search and ads services).

\parab{CS.} Each path corresponds to a server and contains multiple edges. Each edge models a different type of resource on each server (\eg CPU, memory, or GPU). Demands are jobs that require a number of workers. We model heterogeneity (workers may progress at different rates on different servers) with the utility term  $q_k^p$ and scale how much of each resource the worker uses with $r_k^e$.\footnote{We can also use these terms to model similar aspects in TE.} We also support extensions, such as jobs with varying resource requirements~\cite{drf,tetris}.



\vspace{0.02in}
We are unaware of any model as general as ours for max-min fair allocation. However, solving this general model is hard: when we must allocate resources along multiple paths ({\em groups} of {\em resources}), local fairness does not imply global fairness, so single-path solutions~\cite{s-PERC-Lavanya} are ineffective.
In \sysname, we focus on this model and leave the extension to problems that we cannot model as graphs~\cite{gersht1990joint} for future work.


\subsection{\systitle Overview}
\label{s:sysname-overview}

\sysname offers a suite of \textit{allocators} that produce approximate solutions for graph-based max-min fair allocation problems. An allocator is either an algorithm or optimization (or a combination of both) that assigns max-min fair rates that meet the demand and capacity constraints. 

\tabref{tab:alloc_prop} lists our allocators, their key properties, and salient parameters that let the user trade-off between fairness, speed, and efficiency (see~\secref{s:sysn-solv-engine} for more details). We have many choices here and na\"{i}vely running multiple allocators in parallel can waste computational resources.
We suggest a simple decision process along with a hyperparameter search to help practitioners choose a suitable allocator (\figref{fig:design-overview} and \ref{fig:decision-tree}).
A sensitivity analysis in~\secref{sec:eval} indicates that this decision process is robust.
We do not claim any credit for it but use it as evidence that we have effective mechanisms to make the choice.

The {\sf GeometricBinner}~(\secref{s:oneshotopts}) is the only allocator in \sysname with worst-case fairness guarantees: for a given $\alpha > 1$, it provably assigns rates to each demand within $[\alpha^{-1}, \alpha]$ times its optimal max-min fair rate. {\sf EquidepthBinner}~(\secref{s:comb_ext}) is the best choice for users who prefer fairness and efficiency but do not need formal worst-case guarantees. {\sf AdaptiveWaterfiller}~(\secref{s:genwaterfilling}) is suitable for users who prefer speed over efficiency.

\begin{figure}[t!]
	\includegraphics{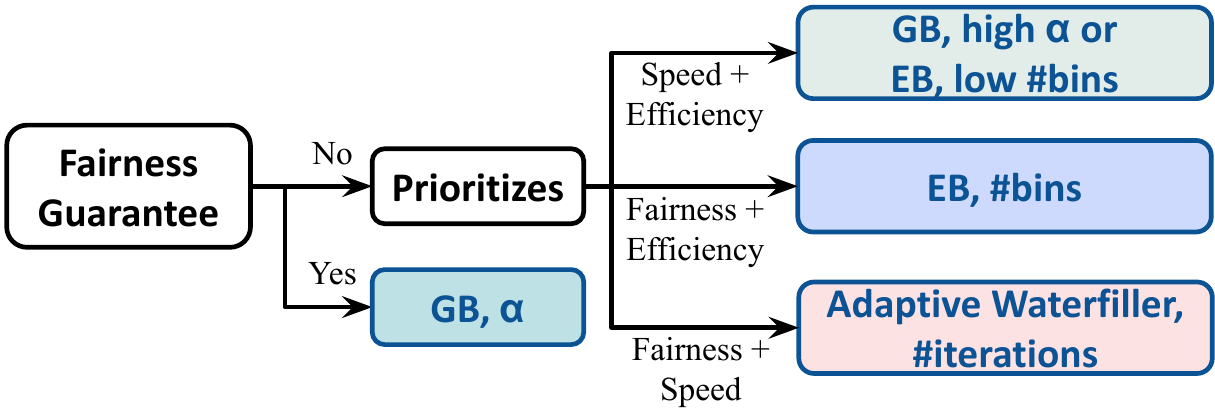}
	\caption{An example of how to pick the right allocator.}
	\label{fig:decision-tree}
\end{figure}

%% file: cr/design_v4.tex
\section{Max-Min Fair Resource Allocators}
\label{s:sysn-solv-engine}

We present two novel kinds of multi-resource max-min fair resource allocation algorithms.

\subsection{One-shot Optimizations}
\label{s:oneshotopts}
\label{sec::opt}

\noindent {\bf Overview.} We can think of max-min fair resource allocation as an optimization with a prioritized list of objectives: first, we maximize the smallest allocation, then the second smallest, and so on. This intuition naturally leads to a sequence of linear optimization problems (LPs)~\cite{danna-practical-max-min, SWAN}.

Prior exact methods are slow since they solve nearly as many LPs as the number of unique resources~\cite{danna-practical-max-min}~(number of edges in a network or machines in a cluster).

We will show how to linearize a prioritized list of objectives such that we can solve one optimization instead of a (long) sequence. The optimization we arrive at is analytically interesting but can encounter double-precision errors, requires a sorting network to sort allocations, and is consequently slower in practice~---~we instead linearize an approximate version.


SWAN~\cite{SWAN} uses an approximate solution that needs to solve fewer LPs. It gradually and geometrically increases the maximum possible rate for each demand and guarantees the final allocations are within $\alpha\times$ their optimal fair rates. Users pick an $\alpha$ based on their requirements for fairness and speed~(e.g., $\alpha=2$ in SWAN). A larger $\alpha$ requires fewer LPs but results in less fair allocations.\footnote{The number of LPs is $\log_\alpha Z$ where $Z$ is the ratio between the largest and the smallest request.}
Microsoft has been using SWAN in production for many years~\cite{onewan}. 

We develop an approximate one-shot optimization by linearizing SWAN's approximate geometric method. Our idea is to define ``bins'' that capture the geometric rate increase at each iteration and introduce new variables to model each flow's allocation from each bin instead of the cumulative total. This combination of techniques is novel and has the same worst-case fairness guarantees as SWAN. 
By linearizing at the granularity of bins, we no longer encounter double precision issues, do not need a sorting network, and achieve an empirically faster solution.

We flesh out the details next (Tables~\ref{t:notation} and~\ref{t:notation-soroush} show our notations).
We discuss why this one-shot optimization is fundamentally smaller and faster than SWAN's sequence of optimizations. We present results from our production deployment in~\secref{subsec:TE}.

\parab{Max-min fair allocation as a sequence of LPs.} If we have $n$ demands, we can use $n$ LPs to compute max-min fair allocations~---~the $i^{\mbox{\tiny {th}}}$ LP in the sequence maximizes the $i^{\mbox{\tiny {th}}}$ smallest rate. Let $t_i$ be the $i^{\mbox{\tiny {th}}}$ smallest rate, then:

\vspace{-0.15in}
{\small
	\begin{align}
	\label{eq:seq-max-min-flow}
	& \mbox{\sf MaxMin}_{i}(\mathcal{E}, \mathcal{D}, \mathcal{P})~\triangleq && \hspace{-3mm} \arg\max~~ t_{i} \\
	& \hspace{4mm}\mbox{s.t.} \quad && \hspace{-17mm} (t_1, \dots, t_{i-1}) \in \mbox{\sf MaxMin}_{i-1}(\mathcal{E}, \mathcal{D}, \mathcal{P}), \nonumber \\
	&&& \hspace{-17mm} f_k \geq t_{i},\quad\quad \forall k ~~\mbox{whose rate is not yet frozen} \nonumber \\ 	 
	&&& \hspace{-17mm} \mathbf{f} \in {\sf FeasibleAlloc}(\mathcal{E}, \mathcal{D}, \mathcal{P}).\nonumber
	\end{align}
}

Note that the algorithm freezes demands that can not receive more than $t_i$ in each iteration. These demands will not receive any allocations in later iterations.
{\tiny
	\begin{table}[t!]
		\centering
		{\footnotesize
			\renewcommand{\arraystretch}{1.3}
			\begin{tabular}{l p{2.3in}}
				{\bf Term} & {\bf Meaning}\\
				\hline
				\hline
				$\mathbf{t}$, $t_i$ & sorted rate vector and the $i^{\mbox{\tiny {th}}}$ smallest rate\\
				\hline
				$N_\beta, \mathcal{D}_b$ & number of bins and set of demands in bin $b$\\	\hline
				${\ell}_b, s_b$ & boundary and slackness of bin $b$\\
				\hline
				$\mathbf{f}_b$, $f_{kb}$ & bin allocation vector and rate of demand $k$ in bin $b$ \\
				\hline
			\end{tabular}
		}
		\caption{Additional notation for {\sysname}. \label{t:notation-soroush}}
	\end{table}
}

\parab{Our one-shot optimal max-min fair solution.} We change~\eqnref{eq:seq-max-min-flow} (changes are in color) to a single optimization by (1) using a sorting network~\cite{Sorting-Networks-Batcher, pretium,FFC} to sort the rates as part of the optimization (\figref{fig:sorting-network-example}), and (2) using a linear weighted objective where the weight of a demand decreases based on its rank in the sorted order~---~these weights incentivize the optimization to increase the smaller rates.

\eqnref{eq:seq-max-min-flow} does not need sorting because each LP maximizes the next smallest rate. The one-shot optimization, however, must explicitly sort the allocations in order to weight them appropriately in the objective. Let $\epsilon < 1$, then:

\vspace{-0.1in}
{\small
	\begin{align}
	\label{eq:opt_max_min_flow}
	& \mbox{\sf OneShotOpt}(\mathcal{E}, \mathcal{D}, \mathcal{P}) ~ \triangleq && \hspace{-10mm} \arg\max_{\mathbf{f}} \textcolor{maroon}{\sum_{i=1}^{n} \epsilon^{i-1}} t_i \\
	& \hspace{4mm} \mbox{s.t.} \quad && \hspace{-28mm}\textcolor{maroon}{({t_1},\dots,{t_n}) = \text{\sf sorted rates} (\mathbf{f}),} \nonumber \\
	&&& \hspace{-28mm} \mathbf{f} \in {\sf FeasibleAlloc}(\mathcal{E}, \mathcal{D}, \mathcal{P}).\nonumber
	\end{align}
}

We prove {\sf OneShotOpt} leads to max-min fair rates:
\begin{theorem}
	\label{th:epsilonOneShot}
	There exist values of $\epsilon$ for which the optimization in~\eqnref{eq:opt_max_min_flow} yields the same max-min fair rate allocations as the sequence of optimizations shown in~\eqnref{eq:seq-max-min-flow}.
\end{theorem} 

\begin{proof}[Proof Sketch]
Let ${\mathbf{t^\dag}}$ be the rate vector solution from~\eqnref{eq:opt_max_min_flow}.
Notice that the optimal max-min fair rate vector, say ${\mathbf{t^*}}$, is a feasible solution to~\eqnref{eq:opt_max_min_flow}. 
Thus, $\sum_i \epsilon^{i-1} t^\dag_i \geq \sum_i \epsilon^{i-1}t^*_i$ (otherwise $\mathbf{t^\dag}$ is not the optimal solution to \eqnref{eq:opt_max_min_flow}).
We can re-arrange and get $t^*_1 - t^\dag_1 \leq {\epsilon}\left(\sum_{i >1} \epsilon^{i-1} t^\dag_i - \sum_{i > 1} \epsilon^{i-1}t^*_i\right)$.

By definition of max-min fairness, we have $t^*_1 \geq t^\dag_1$ since $t^*_1$ is the smallest allocation in the optimal max-min fair solution. If we can find a feasible assignment where the smallest rate $t^\dag_1$ is higher, then $t^*$ cannot be max-min fair~---~we can increase the smallest rate without hurting any other demand with an smaller allocation (because no such demand exists).
These statements together imply the smallest rates must match as $\epsilon \rightarrow 0$.
The rest follows by induction.
\end{proof}

{\sf OneShotOpt} is not practical. We need a small $\epsilon$ to find an optimal solution (see proof), but we will encounter double precision errors if the smallest weight ($\epsilon^{n-1}$) is too small. In this case, the formulation may have to sacrifice optimality and use a large $\epsilon$ to solve the one-shot optimization (especially when there are many demands). Even with a large $\epsilon$, we find that solving LPs with a full sorting network is slow~\cite{pretium,FFC} since sorting networks add $O(nlog^2(n))$ additional constraints.

\parab{Our one-shot {\sf \bf GeometricBinner}}~(or {\sf GB}) linearizes the following approximate max-min fair technique. Compared to~\eqnref{eq:seq-max-min-flow},~\eqnref{eq:seq-approx-max-min-flow} shows a shorter sequence of LPs inspired by SWAN~\cite{SWAN} (changes in color) but also differs from SWAN~\footnote{\eqnref{eq:seq-swan} in~\secref{sec:swan_form} shows SWAN's formulation for comparison.} in one crucial way (it introduces new variables to track the {\em increase} in the allocation of each demand in each iteration):

\vspace{-0.15in}
{\small
	\begin{align}
	\label{eq:seq-approx-max-min-flow}
	&\mbox{\sf ApproxMaxMin}_{\textcolor{maroon}{b}}(\mathcal{E}, \mathcal{D}, \mathcal{P})~\triangleq && \hspace{-1mm}\arg\max \textcolor{maroon}{\sum_{k \in \mathcal{D}} {f_{k}}} \\
	& \hspace{4mm} \mbox{s.t.} \quad && \hspace{-25mm} \textcolor{maroon}{f_{k} = \sum_{\mbox{bins}~j \leq b} f_{kj},} \quad & \hspace{-1mm} \textcolor{maroon}{\forall k \in \mathcal{D}} \nonumber\\
	&&& \hspace{-25mm} \textcolor{maroon}{f_{k1} \le U,} & \hspace{-1mm} \textcolor{maroon}{\forall k \in \mathcal{D}} \nonumber\\
	&&& \hspace{-25mm} \textcolor{maroon}{f_{kb} \leq U(\alpha^{b-1} - \alpha^{b-2}),} & \hspace{-1mm} \textcolor{maroon}{\forall b > 1, \forall k \in \mathcal{D}} \nonumber \\ 
	&&& \hspace{-25mm} \textcolor{maroon}{f_{kb}  = 0 \quad\quad \mbox{if}~\sum_{j < b} f_{kj} < U\alpha^{b-2},} & \hspace{-1mm} \textcolor{maroon}{\forall b>1, \forall k \in \mathcal{D}} \nonumber \\
	&&& \hspace{-25mm} (\mathbf{f}_{1}, \dots, \mathbf{f}_{b-1}) \in \mbox{\sf ApproxMaxMin}_{b-1}, \nonumber \\
	&&& \hspace{-25mm} \mathbf{f} \in {\sf FeasibleAlloc}(\mathcal{E}, \mathcal{D}, \mathcal{P}).\nonumber
	\end{align}
}

The changes cause the $b^{\mbox{\tiny{th}}}$ LP, where index $b$ begins at $1$, to allocate only {\em up to} $U(\alpha^{b-1} - \alpha^{b - 2})$ for $b > 1$ and up to $U$ for $b=1$. The algorithm also freezes any demand that does not receive the full rate from the previous iteration. $\alpha$ and $U$ are input parameters that control the fairness guarantee and the minimum rate, respectively. Observe that each LP in the sequence allocates rates unfairly, but the unfairness is bounded as each LP allocates rates only within a small range. 

\begin{figure}[t!]
	\centering
	\includegraphics[width=0.85\linewidth]{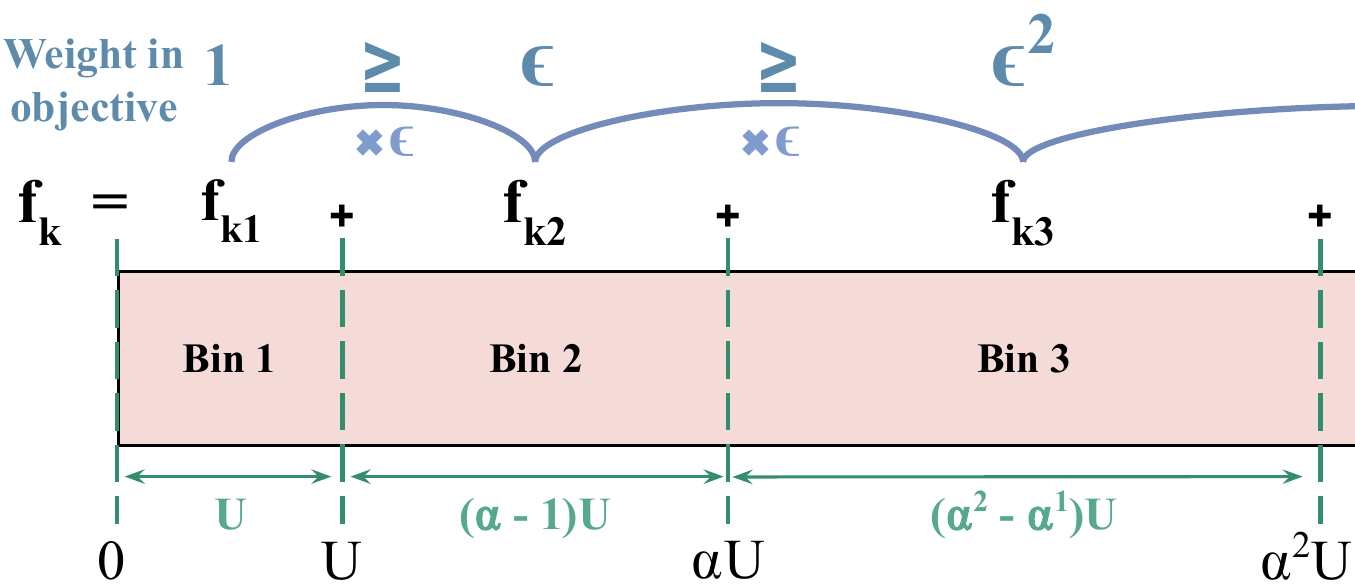}
	\caption{\textbf{Geometric Binning: } Approximate one-shot max-min fair allocations. We can model the problem in one shot because we add a new variable to track the allocation to each demand from each bin. With this idea, we can then change the objective to incentivize the optimization to make sure its allocation saturates smaller bins before allocating from subsequent bins.}
	\label{f:geo-binning-example}
\end{figure}

\vspace{0.03in}
\figref{f:geo-binning-example} shows the key idea behind our one-shot geometric binner.
If we consider each allocation as the sum of contributions from different, geometrically-sized bins, we can use $\epsilon$-weighting per bin to incentivize the optimization to allocate more from the smaller bins. 
The resulting formulation is:

\vspace{-0.1in}
{\small
	\begin{align}
	\label{eq:new_geobinning}
	& \mbox{\sf GeoBinning}(\mathcal{E}, \mathcal{D}, \mathcal{P})~\triangleq && \hspace{-1mm} \arg\max_{\mathbf{f}} \sum_{k \in \mathcal{D}}~\textcolor{maroon}{\sum_{\mbox{bins} ~ b} \epsilon^{b-1} f_{kb}}\\
	& \hspace{4mm} \mbox{s.t.} \quad && \hspace{-17mm} f_{k} = \sum_{\mbox{bins} ~ b} f_{kb}, & \hspace{-9mm}\forall k \in \mathcal{D} \nonumber\\
	&&& \hspace{-17mm} f_{k1} \le U, & \hspace{-9mm} \forall k \in \mathcal{D} \nonumber\\
	&&& \hspace{-17mm} f_{kb} \le U(\alpha^{b-1} - \alpha^{b-2}), & \hspace{-9mm} \forall b > 1,\forall k \in \mathcal{D} \nonumber\\
	&&& \hspace{-17mm} \mathbf{f} \in {\sf FeasibleAlloc}(\mathcal{E}, \mathcal{D}, \mathcal{P}).\nonumber 
	\end{align}
}

The geometric binner (\eqnref{eq:new_geobinning}):
\begin{Itemize}
	\item Applies to various bin choices beyond the geometric ones we used here. We use a similar intuition from equi-depth binning in databases~\cite{jagadish1998optimal} to show in \secref{s:comb_ext} that we can choose bin boundaries to improve fairness.
	
	\item Offers the same fairness guarantee as SWAN~\cite{SWAN} when using the same (geometric) bins. {\sf GeoBinning} allocates rates within an $\alpha$ ratio of the optimal max-min fair rates for any demand.
	\begin{theorem}
		\label{th:binFilling}
		\eqnref{eq:new_geobinning} assigns resources to a demand $k$ in bin $b$ only if it has assigned demand $k$ the full rate from all of the larger-weighted bins. \end{theorem} 
	\begin{proof}[Proof Sketch]
		Assume otherwise: \eqnref{eq:new_geobinning} has assigned a non-zero rate to some demand $k$ in some bin $b$ without assigning the full rate from some other bin $j < b$. Then, we can move some $\delta$ rate from bin $b$ to $j$ and not violate any constraints yet improve the objective value.\footnote{Smaller indexed bins have larger weights because $\epsilon <1$.}
	\end{proof}
	We can combine Theorem~\ref{th:binFilling} with the proof technique of Theorem~\ref{th:epsilonOneShot} to prove~\eqnref{eq:new_geobinning} will allocate the same rates as~\eqnref{eq:seq-approx-max-min-flow}, so the approximation ratio proof from~\cite{SWAN} applies directly.
	\item Lets users adjust $\alpha$ to balance the trade-off between fairness approximation guarantee and the solver time.
	\item Is less likely to run into precision issues compared to~\eqnref{eq:opt_max_min_flow} since there are fewer bins than demands.
	\item Requires no sorting constraints unlike~\eqnref{eq:opt_max_min_flow}.
	\item Is only slightly larger in size compared to {\em each of the optimizations} in~\eqnref{eq:seq-approx-max-min-flow} (see~\secref{ss:expected-run-time} for more details). The key difference is that we can now run one LP instead of many. The one-shot optimization is empirically faster, likely due to redundant computation between the LPs in the sequence of optimizations in~\eqnref{eq:seq-approx-max-min-flow}.
\end{Itemize}

\subsection{Multi-path Waterfilling}
\label{s:genwaterfilling}

We also generalize the classical waterfilling algorithm for max-min fair allocation over multiple resources.
We present parallelizable combinatorial algorithms (and not optimizations) with better empirical performance compared to~\secref{s:oneshotopts} but weaker fairness guarantees.


%

Waterfilling is a well-known method that applies to scenarios where all the demands are {\em unconstrained} and require resources on {\em exactly one path}~\cite{s-PERC-Lavanya}. Under these conditions, we can achieve max-min fairness by visiting resources in the ascending order of their fair share and splitting their capacities {\em fairly} among the demands~\cite{csfqorig} (see~\algoref{alg:one_path_wt_waterfilling}\footnote{Notice weighted max-min fair allocation is roughly the same with one minor change: we relatively weigh the rate we allocate to each flow.})

We extend this approach to constrained demands by adding a virtual edge with a capacity equal to the requested rate for each demand. This augmented topology ensures that demands receive at most what they asked for. For small requests, the virtual edge becomes the bottleneck and limits the allocation.

%
%

{
	\RestyleAlgo{ruled}
	\begin{algorithm}[t!]
		\begin{small}
			\DontPrintSemicolon
			\LinesNumbered
			\caption{Waterfilling algorithm to compute single-path weighted max-min fair rates.}
			\label{alg:one_path_wt_waterfilling}
			\KwIn{$\weightedroutingmatrix$ where $\weightedroutingmatrix[e, k]$ is the weight of single-path demand $k$ on link $e$.}
			\KwIn{$\capacity$: link capacity vector.}
			\KwOut{$\rate$: max-min rate allocation vector.}
			
			$\setactiveflows \leftarrow [0, \dots, K-1]$ \algcomment{initial list of active demands}
			$\rate \gets \mathbf{0}$ \algcomment{initial rate vector}
			
			\While{$\norm{\setactiveflows} > 0$}{
				$\vect{n} \gets \weightedroutingmatrix \mathbf{1}$ \algcomment{total weight per link}
				$\vect{\fairsharevector} \gets \frac{\capacity}{\vect{n}}$ \algcomment{vector division, fairshare per link}
				$e \gets \mathop{\mathrm{arg\,min}}~~\vect{\fairsharevector}$ \algcomment{link with minimum fair share}
				$\setdemands_e \gets \{ k: \weightedroutingmatrix[e, k] > 0 \}$ \algcomment{active demands on link $e$}
				\ForEach{$k \in \setdemands_e$}{
					$\rate[\setactiveflows[k]] \gets \fairsharevector[e] \weightedroutingmatrix[e, k]$ \algcomment{fix to weighted fair share}
					\ForEach{$l: \weightedroutingmatrix[l, k] > 0$}{
						$\capacity[l] \gets \capacity[l] - \rate[\setactiveflows[k]]$ \algcomment{deduct allocations}
					}
				}
				$\capacity \gets \capacity[\setminus \{e\}]$ \algcomment{remove link $e$}
				$\weightedroutingmatrix \gets \weightedroutingmatrix[\setminus \{e\}; \setminus \setdemands_e]$ \algcomment{remove link $e$ and its demands}
				$\setactiveflows \gets \setactiveflows \setminus \setdemands_e$	\algcomment{update set of active demands}
			}
			\Return{$\rate$}
		\end{small}
	\end{algorithm}
}

It is harder to generalize waterfilling to multi-path settings because the local max-min fair allocation at individual resources is not globally fair. \figref{f:example:water} shows a simple example where the blue demand~---~which has access to more paths~---~must receive a locally \emph{unfair} share on the common link in order to produce a globally max-min fair solution. We next modify waterfilling to produce approximate, globally max-min fair rates in the general multi-resource setting.

\renewcommand{\arraystretch}{1.1}
\begin{figure}[t]
	\centering
	\subfigure[Global and Local (per-link) max-min fairness are different.\label{f:global_mmf_example}]{\includegraphics[width=0.85\linewidth]{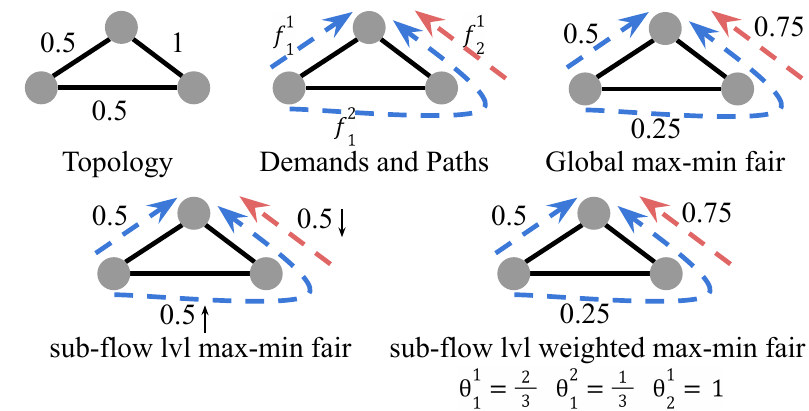}}
	\subfigure[{\sf AW}'s weight multipliers and allocations converge to global max-min fair allocation.\label{f:iter_example}]{
		\begin{small}
		\begin{tabular}{l|cccccccc}
			
			& \multicolumn{8}{l}{\bf \# iteration $t\longrightarrow$}\\\hline
			
			$\theta^1_1$ 	& $\frac{1}{2}$ & $\frac{3}{5}$ 	& $\frac{7}{11}$ 	& $\frac{15}{23}$ 	& $\frac{31}{47}$ 	& $\frac{63}{95}$ & $\ldots$ & $\frac{2}{3}$ \\
			
			$\theta^2_1$ 	& $\frac{1}{2}$ & $\frac{2}{5}$ 	& $\frac{4}{11}$ 	& $\frac{8}{23}$	& $\frac{16}{47}$	& $\frac{32}{95}$ & $\ldots$ & $\frac{1}{3}$\\
			
			$\theta^1_2$ 	& $1$ 			& $1$				& $1$ 			 	& $1$				& $1$				& $1$ & $\ldots$ & $1$\\

			$f^1_1$ 		& $\frac{1}{2}$ & $\frac{1}{2}$		& $\frac{1}{2}$ 	& $\frac{1}{2}$		& $\frac{1}{2}$		& $\frac{1}{2}$ & $\ldots$ & $\frac{1}{2}$\\

			$f^2_1$ 		& $\frac{1}{3}$ & $\frac{2}{7}$		& $\frac{4}{15}$ 	& $\frac{8}{31}$	& $\frac{16}{63}$	& $\frac{32}{127}$ & $\ldots$ & $\frac{1}{4}$\\

			$f^1_2$ 		& $\frac{2}{3}$ & $\frac{5}{7}$		& $\frac{11}{15}$ 	& $\frac{23}{31}$	& $\frac{47}{63}$	& $\frac{95}{127}$ & $\ldots$ & $\frac{3}{4}$\\
			
		\end{tabular}
		\end{small}
	}
	\caption{An example to illustrate the difficulty in extending waterfilling to multi-path settings and how our {\sf AdaptiveWaterfiller} effectively tackles the issue. Waterfilling (single-path) computes local fair shares and is ineffective in multi-path settings as it ignore the dependencies between different paths of a single demand.\label{f:example:water}}
\end{figure}
\renewcommand{\arraystretch}{1}

\parab{ApproxWaterfiller}~(or {\sf aW}): For each demand, {\sf aW} creates several ``subdemands'', each going through one of the demand's paths. Subdemands of each demand pass through a shared virtual edge to ensure the algorithm does not allocate more than the requested rate. We then use waterfilling to assign rates to these subdemands. This algorithm simply ignores the coupling between multiple paths and does not reach global max-min fair rates, but we use it as the starting point to generalize waterfilling. As a solution, it is fast, and we also use a variant of~\algoref{alg:one_path_wt_waterfilling} to speed it up further (\algoref{alg:one_path_approx_wt_waterfilling}). 

The new algorithm simplifies~\algoref{alg:one_path_wt_waterfilling} by retaining the initial order of the links in subsequent iterations. In each iteration, it only recomputes the fair share for the link under consideration and fixes the rates for the demands bottlenecked by that link. It is approximate (even in the one-path case~\cite{Janus-Alipourfard,swarm}) but is faster and more parallelizable.

\vspace{0.02in}
Global max-min fairness assigns lower rates to subdemands that are on congested paths but their corresponding demand can get enough allocation from its other paths. For example, the blue demand in~\figref{f:global_mmf_example} should receive a lower allocation on the path through the congested link. Intuitively, the allocator can get closer to global max-min fair assignments by moving each demand's allocation from more congested paths to less congested ones. We can achieve this by iteratively seeking more rates from subdemands that have received higher rates (\ie on less-contended paths) in previous iterations.

{
	\RestyleAlgo{ruled}
	\begin{algorithm}[t]
		\DontPrintSemicolon
		\LinesNumbered
		\begin{small}
			\caption{Our Approx. Waterfilling algorithm to compute single-path weighted max-min fair rates.}
			\label{alg:one_path_approx_wt_waterfilling}
			\KwIn{$\weightedroutingmatrix$ where $\weightedroutingmatrix[e, k]$ is the weight of single-path demand $k$ on link $e$.}
			\KwIn{$\capacity$: link capacity vector.}
			\KwOut{$\rate$: max-min rate allocation vector.}
			
			${\rate} \gets \boldsymbol{\infty}$ \algcomment{initial max-min rate set to $\infty$}
			$\vect{n} \gets \weightedroutingmatrix \mathbf{1}$ \algcomment{total weight per link}
			$\set{L} \gets \mathop{\mathrm{arg\,sort}}~\frac{\capacity}{\vect{n}}$ \algcomment{vector division, sort links in ascending order}
			\ForEach{$e \in \set{L}$}{
				$\setdemands_e \gets \{ k: \weightedroutingmatrix[e, k] > 0 \}$ \algcomment{demands on link $e$}
				\While{$\setdemands_e \neq \emptyset$}{
					$\fairshare \gets \frac{\capacity[e]}{\sum_{k \in \setdemands_e} \weightedroutingmatrix[e, k]}$ \algcomment{fair share of link $e$}
					$\set{B} \gets \{k \in \setdemands_e: \rate[k] < \fairshare \weightedroutingmatrix[e,k]\}$ \;
					\uIf(\tcp*[f]{if no flows bottlenecked elsewhere,}){$\set{B} = \emptyset$}{
						$\rate[\setdemands_e] \gets \fairshare \weightedroutingmatrix[e, \setdemands_e]$ \algcomment{fix rate to weighted share.}
						\text{break} \; 
					}\Else(\tcp*[f]{otherwise, remove those bottlenecked elsewhere.}){
						$\capacity[e] \gets \capacity[e] - \sum_{k \in \set{B}} \rate[k]$ \;
						$\setdemands_e \gets \setdemands_e \setminus \set{B}$
					}
				}
			}
			\Return{$\rate$}
		\end{small}
	\end{algorithm}
}
\parab{AdaptiveWaterfiller}~(or {\sf AW}):  Motivated by this intuition, {\sf AW} uses a weighted version of {\sf aW} (using \algoref{alg:one_path_wt_waterfilling} or~\algoref{alg:one_path_approx_wt_waterfilling}) and adjusts the input weights ($\weightedroutingmatrix$) to seek more rate from subdemands on less congested paths. 

Let $\theta^p_{k}$ be the weight multiplier for the subdemand of demand $k$ on path $p$. {\sf AW} initializes these multipliers as $\theta^p_{k} = \frac{1}{\|\{p \in \mathcal{P}_k\}\|}$. In each iteration, {\sf AW} first computes $\weightedroutingmatrix$ directly from $\theta$. $\weightedroutingmatrix[e, k_p]$ is the weight of the subdemand $k_p$ (demand $k$ on path $p$) on link $e$\footnote{Note that waterfilling requires each demand to be on a single path. We use the notation $k_p$ to show the single-path subdemands.}. Following the definition, $\weightedroutingmatrix[e, k_p] = \theta^p_k~\indicator{e \in p}$. {\sf AW} then invokes one of the waterfilling algorithms\footnote{We use~\algoref{alg:one_path_approx_wt_waterfilling} for our experiments since it is an order of magnitude faster with only a slight decrease in fairness (\figref{fig:traffic-engineering:all-samples:speedup-fairness}).} with these weights $\weightedroutingmatrix$. For iteration $t+1$, {\sf AW} sets $\theta^p_{k} (t + 1) = \frac{f^{p}_{k} (t)}{\sum_p f^p_{k}(t)}$ where $f^p_k (t)$ is the rate demand $k$ obtains from path $p$ in iteration $t$. We show how multipliers and rates evolve in our example in~\figref{f:iter_example}.

{\sf AW} converges when $\theta^{p}_{k} (t+1) = \theta^{p}_{k}(t)$. We can prove that adapting weight multipliers gets close to global max-min fairness: we say a rate assignment in the multi-path setting is \emph{bandwidth-bottlenecked} if for all demands $k$, (i) each of its subdemands $f^{p}_{k}$ is bottlenecked on some link $l$, and (ii) $f_k \ge f_j$, for all demands $j$ that have any subdemand on any such link $l$. We prove in~\S\ref{sec:proof_bottleneck} that:
  
\begin{theorem}\label{T:bandwith_bottleneck}
	{\em If} the adaptive waterfiller converges, it converges to a bandwidth bottlenecked assignment. 
\end{theorem}

We prove the global max-min fair allocation is bandwidth-bottlenecked (see~\secref{sec:additional_results_waterfilling}). The converse is not true~---~not all bandwidth-bottlenecked allocations are max-min fair. However, the set of bandwidth-bottlenecked allocations is {\em significantly} smaller than the set of all feasible allocations. We also prove that {\sf AW} converges when its assignment is bandwidth-bottlenecked (\ie it stops iterating). Empirically, {\sf AW}'s allocations stabilize within $5$~--~$10$ iterations on average~(\secref{s:conv-sens-analys}).

Adaptive waterfiller produces allocations that belong to a constrained set containing the optimal max-min fair rates. It is slower than approximate waterfiller because it iterates and updates weight multipliers. It is faster than the Geometric Binner as it does not solve an LP. Users can tune the maximum number of iterations to trade-off between fairness and speed.

\subsection{Combinations and Extensions}
\label{s:comb_ext}
Empirically, we find that the geometric binner (\S\ref{s:oneshotopts}) is fairer than what its worst-case guarantee suggests
(recall, we prove the rates will be within $[\alpha^{-1}, \alpha]$ times the optimal max-min fair rate).
We can attribute most of the unfairness to bins that happen to contain many demands (\figref{fig:imbalanced-bins:app}): can we set the bin boundaries differently to improve fairness?

We use the generalized waterfillers (\S\ref{s:genwaterfilling})~---~which are fast but lack worst-case guarantees~---~to set bin boundaries in a way that spreads demands more uniformly across bins:

\parab{Equi-depth Binner}~(or {\sf EB}) applies {\sf GeoBinning} (\eqnref{eq:new_geobinning}) with a few changes: it uses the rate allocation from {\sf AdaptiveWaterfiller} to approximate the order across demands; distributes demands more uniformly over bins; and finds the bin boundaries as part of the optimization.
Specifically, {\sf EB} divides demands $\mathcal{D}$ into $N_\beta$ equi-sized sets ($\mathcal{D}_1 \ldots \mathcal{D}_n$) based on their increasing order of rates from {\sf AW}. In {\sf EB}, the demands in a set $D_b$ only receive rates from one bin with index $b$.  {\sf EB} dynamically chooses bin boundaries: $\forall k \in D_b, {\ell}_{b-1} \leq f_{k} < {\ell}_b + s(b)$ where $s(b)$ is a small constant that helps reducing the impact of inaccuracies from {\sf AW}. We provide a more formal definition of {\sf EB} in~\S\ref{ss:eb_formulation}.

{\sf EB} is slower than {\sf GB} and {\sf AW} because it executes both but we expect it to be fairer than {\sf GB}~---~it spreads demands more uniformly across bins. We empirically confirm this hypothesis. It is hard to formally analyze {\sf EB} but we suspect it also offers tighter worst-case guarantees. This is subject for future work.


\parab{Extensions:} We did not explicitly account for weighted max-min fairness~(e.g., $w_k$ in~\secref{s:unified-formulation}) when describing our one-shot optimization. 
This extension is straightforward. For example, we can replace the first constraint in the geometric binner~({\sf GB}) in~\eqnref{eq:new_geobinning} with $f_{k}/w_k = \sum_{\text{bins} ~ b}f_{kb}$. Algorithm~\ref{alg:one_path_approx_wt_waterfilling}, which we use in our generalized waterfillers in~\secref{s:genwaterfilling}, already supports weights. We compute the per-edge per-subdemand weighted routing matrix as $\weightedroutingmatrix[e, k_p] = w_k \theta_{k}^p \indicator{e \in p}$.

We have also omitted the heterogeneous utilities, different resource consumption scales, and other affine functions in our model~(\S\ref{s:unified-formulation}) when describing the solutions. These extensions are also straightforward and we show them in~\S\ref{sec:general-form}. The key is to appropriately manipulate the constraints that determine when an allocation is feasible~({\sf FeasibleAlloc},~\eqnref{eq:feasible-alloc}).

%% file: cr/evaluation.tex
\section{Evaluation}
\label{sec:eval}
\noindent {\bf Implementation.} 
We implemented \sysname in Python and C\# using Gurobi 9.1.1~\cite{gurobi} as the solver.

\parab{Summary of results.}
We apply \sysname to traffic engineering (TE) and cluster scheduling (CS). We show \sysname captures the trade-off between speed, fairness, and efficiency. We also show the results from integrating \sysname with Azure's production TE system where it reduces the run-times by up to $5.4\times$ without any impact on efficiency and fairness.

In TE, all the allocators in {\sysname} are faster than both the optimal algorithm by Danna \etal~\cite{danna-practical-max-min} (referred to as Danna) and the more practical $\alpha$-approximate SWAN~\cite{SWAN}. 
{\sysname} contains algorithms that match or exceed the efficiency or fairness of these methods while running orders of magnitude faster.
{\sysname} can also trade-off (a little) fairness and efficiency for up to $3$ orders of magnitude speed up.

Our solution scales to one of the largest WAN topologies (over $1000$ nodes and $1000$ edges), which is significantly larger than those in~\cite{danna-practical-max-min, B4, SWAN, zhong2021arrow, bogle2019teavar} and matches the size of topologies in~\cite{Abuzaid-ncflow}. We also analyze the sensitivity of {\sysname} to demand variations and other relevant inputs.

In CS, we show \sysname outperforms two variants of Gavel~\cite{Gavel-Deepak}.
Our Equi-depth binner ({\sf EB}) has the same fairness and efficiency as the optimal variant of Gavel (the one with waterfilling), but is $2$ orders of magnitude faster. 

\subsection{Benchmarks and Metrics}
\label{sec::eval_setup}

\noindent{\bf Benchmarks.} We use state-of-the-art solutions in both WAN-TE and CS as benchmarks to evaluate \sysname:

\parae{WAN-TEs.} We use Danna~\cite{danna-practical-max-min}, SWAN~\cite{SWAN}, and a modified version of the k-waterfilling algorithm~\cite{s-PERC-Lavanya} as benchmarks. We also provide limited comparisons with B4~\cite{B4} for completeness (see~\S\ref{subsec:TE}). The k-waterfilling algorithm only applies to single-path, infinite-demand scenarios --- we extend it to multi-path, demand-constrained cases. 
We tune each benchmark for maximum speed (see~\S\ref{sec:benchmark-tuning}). Following~\cite{SWAN}, we set $\alpha=2$ for SWAN and {\sf GB} unless mentioned otherwise.


\parae{CS.} We compare with two variants of Gavel~\cite{Gavel-Deepak}, the state-of-the-art max-min fair allocator in CS (with and without waterfilling).
We use Gavel's public implementation.


\begin{table}
	\footnotesize
	\centering
	\begin{tabular}{lcc}
		Topology & \# Nodes & \# Edges \\
		\hline
		\hline
		WANLarge & $\sim$1000s  & $\sim$1000s\\ 
		WANSmall & $\sim$100s & $\sim$1000s\\
		\href{http://www.topology-zoo.org/maps/Cogentco.jpg}{Cogentco} & 197 & 486 \\
		\href{http://www.topology-zoo.org/maps/UsCarrier.jpg}{UsCarrier} & 158 & 378 \\
		\href{http://www.topology-zoo.org/maps/GtsCe.jpg}{GtsCe}  & 149  & 386 \\
		\href{http://www.topology-zoo.org/maps/TataNld.jpg}{TataNld}  & 145 & 372 \\ 
	\end{tabular}
	\caption{Topologies used for the evaluation of \sysname.}
	\label{tab:topologies}
\end{table}

\parab{Metrics.} We use the following metrics for comparisons:


\parae{Fairness.} We report fairness of a particular allocation ($\mathbf{f}$) as its distance from the optimal max-min fair allocation ($\mathbf{f}^{*}$)\footnote{Danna and Gavel (w waterfilling) compute the optimal max-min fair allocations in TE and CS respectively. They are too slow for practice but we can run them to completion outside of a production environment.}. 
For fairness distance, we use the $q_\vartheta$ metric~\cite{lu2021pre,marcus2019neo}. 
This metric is resilient to numerical instability and is computed as $\min\big(\frac{\max(f_k, \vartheta)}{\max(f^{*}_k, \vartheta)},\frac{\max(f^{*}_k,\vartheta)}{\max(f_k,\vartheta)}\big)$ for a given demand $k$. 
We report the geometric mean of $q_\vartheta$ across all the demands as the overall fairness measure (the geometric mean is less sensitive than the arithmetic mean to outliers).
For our evaluations, we use $\vartheta=0.01\%$  of the resource (link or GPU) capacities. 

\parae{Efficiency.} We measure efficiency in TE as the total rate allocated to flows relative to Danna (\ie $\frac{e}{e_{danna}}$).
For CS, we measure the effective throughput which is the progress rate of a job given an allocation. We report CS efficiency relative to Gavel (\ie $\frac{e}{e_{gavel}}$).

\parae{Runtime.} 
In most cases, we report speed up (\ie \emph{relative} runtime compared to a baseline  $\frac{s_{baseline}}{s}$).
Our runtimes consist of the time each algorithm needs to compute the allocations.
We measure runtimes on an AMD Operaton 2.4GHz CPU (6234) with 24 cores and 64GB of memory.


\begin{figure*}[t]
	\centering
	\subfigure{\centering\includegraphics[width=0.8\linewidth]{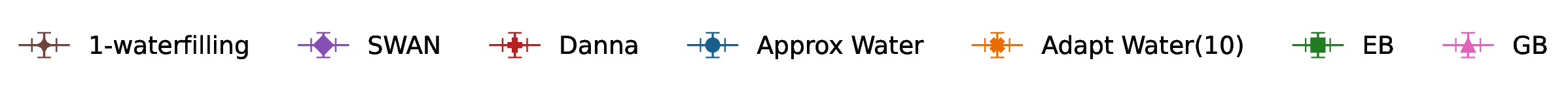}}
	\addtocounter{subfigure}{-1} \vspace{-4mm} \\
	\subfigure[High Load (scale factor $\in$ \{64, 128\})]{\centering
		\includegraphics[width=0.31\linewidth]{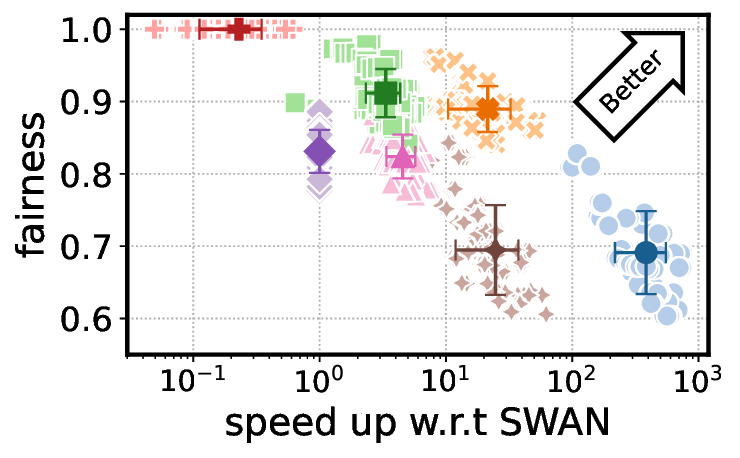}}
	\subfigure[Medium Load (scale factor $\in$ \{16, 32\})]{\centering
		\includegraphics[width=0.31\linewidth]{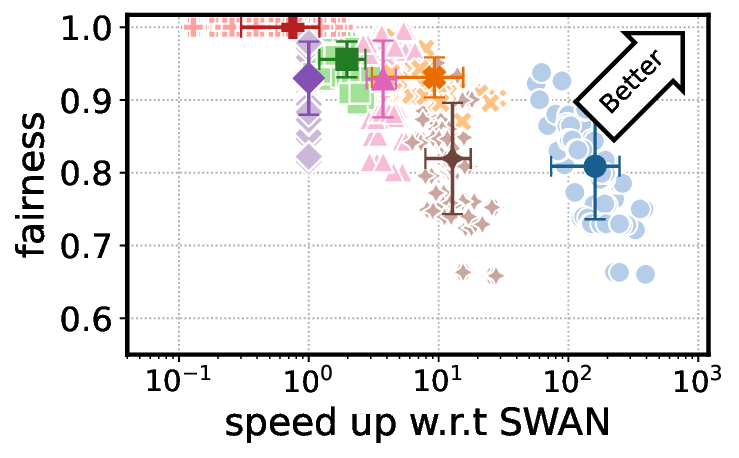}}
	\subfigure[Light Load (scale factor $\in$ \{1, 2, 4, 8\})]{\centering
		\includegraphics[width=0.32\linewidth]{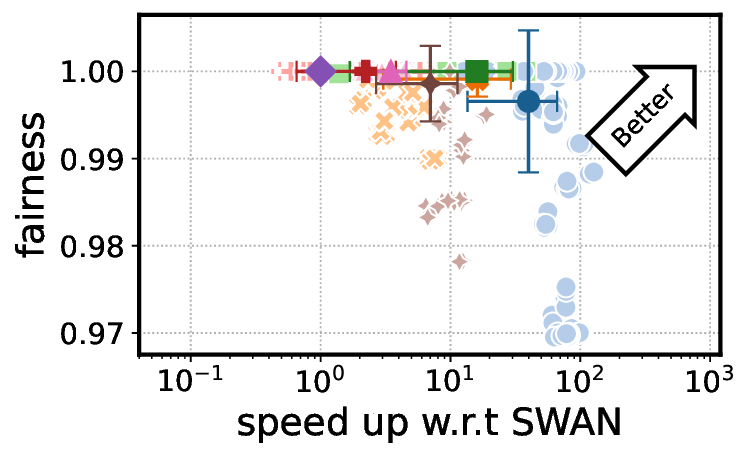}} 
	\caption{\textbf{The fairness vs speed trade-off across different approaches.} As in~\cite{Abuzaid-ncflow}, we use the scale-factor to denote the level of load. We observe even the slowest algorithm in \sysname is faster than SWAN and Danna. While 1-waterfilling is faster than most of the algorithms in \sysname, it has to sacrifice much more in terms of fairness (it is $30\%$ less fair than Danna in the high load case).} 
	\label{fig:traffic-engineering:all-samples:speedup-fairness}
	
	\subfigure[High Load (scale factor $\in$ \{64, 128\})]{\centering
		\includegraphics[width=0.32\linewidth]{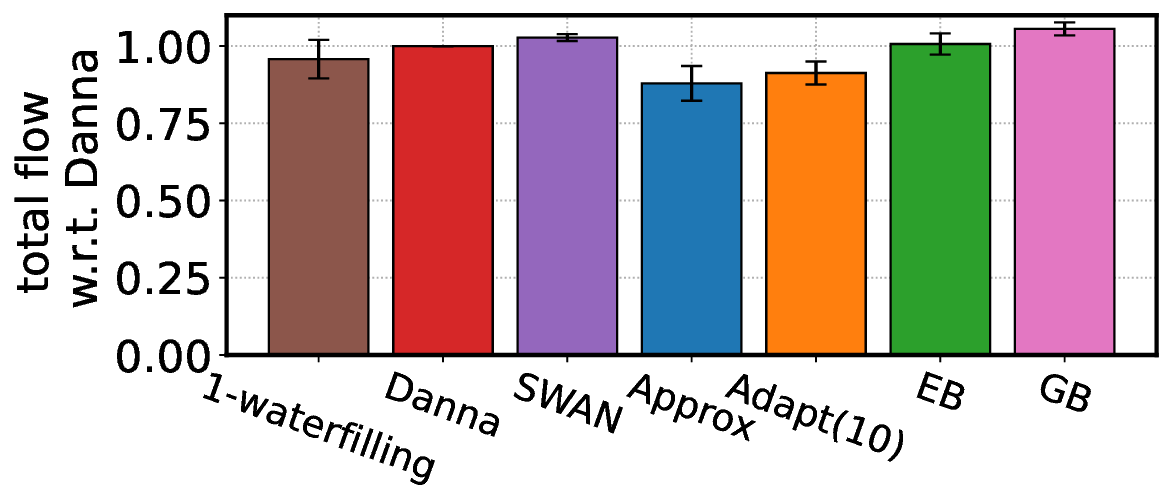}\label{fig:traffic-engineering:all-samples:efficiency:high}}
	\subfigure[Medium Load (scale factor $\in$ \{16, 32\})]{\centering
		\includegraphics[width=0.32\linewidth]{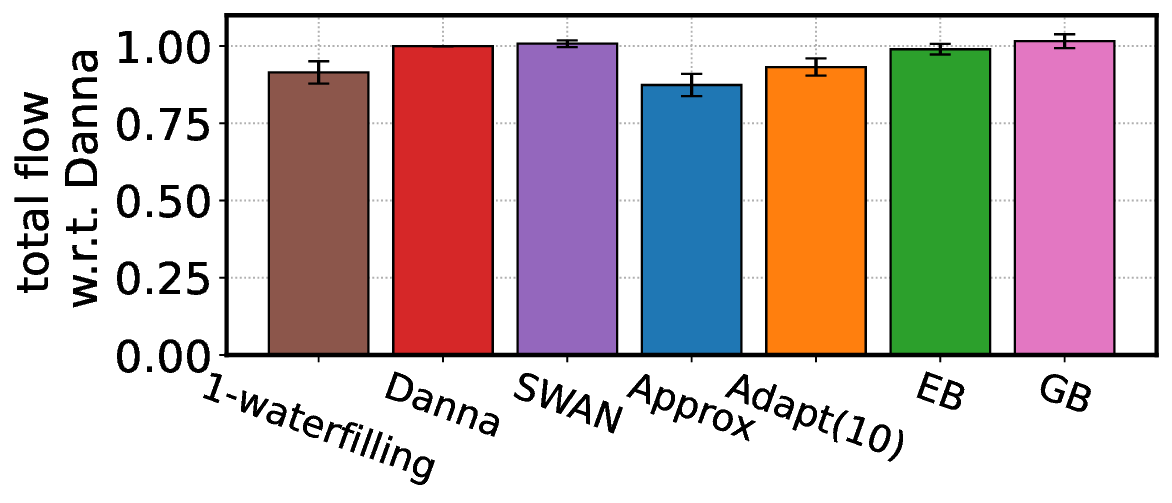}}
	\subfigure[Light Load (scale factor $\in$ \{1, 2, 4, 8\})]{\centering
		\includegraphics[width=0.32\linewidth]{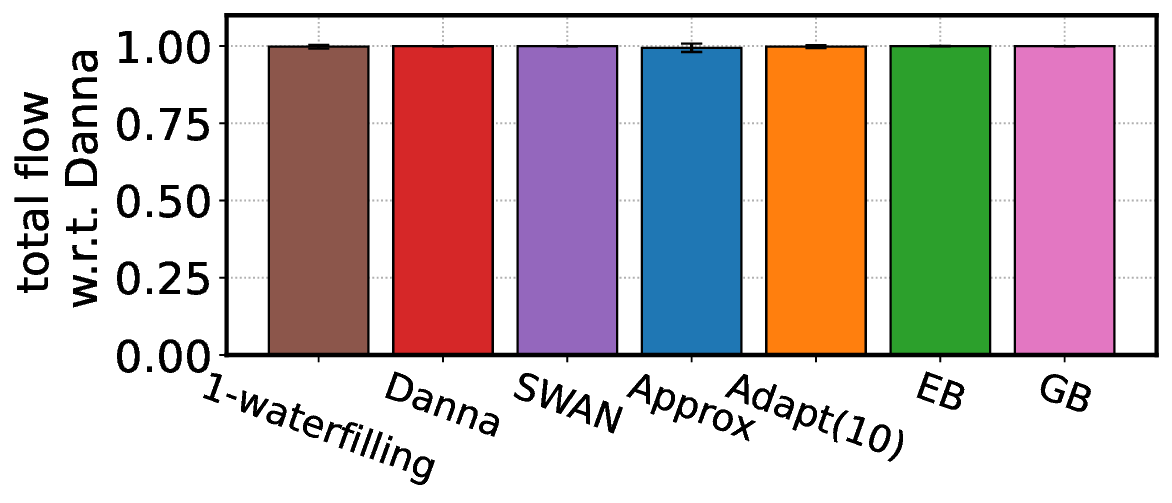}\label{fig:traffic-engineering:all-samples:efficiency:low}}
	\caption{\textbf{The efficiency of {\sysname}'s algorithms and our benchmarks.} We report numbers relative to Danna. Empirically, \sysname Pareto-dominates SWAN, 1-waterfilling, and Danna on the efficiency, agility, and fairness. In (c), the error bar is small because of light load~---~most solutions can satisfy all the demands (fairness is also close to one for most algorithms in these cases.)}
	\label{fig:traffic-engineering:all-samples:efficiency}
\end{figure*}


\subsection{WAN Traffic Engineering}
\label{subsec:TE}

\noindent{\bf Experiment Setup.} \tabref{tab:topologies} summarizes the topologies in our evaluation. We show the results for both Azure's production WAN topology and the topologies from the Topology Zoo~\cite{topo-zoo}. We use K-shortest paths~\cite{YenKLoopLess} to find the paths between node pairs (K=16 unless mentioned otherwise). 

For topologies from Topology Zoo, we generate traffic using Poisson~\cite{traffic-model-ref1}, Uniform, Bimodal, and Gravity~\cite{traffic-model-ref1,gravity-model-ref2} distributions. We follow~\cite{Abuzaid-ncflow} and generate traffic at different \textit{scale factors}. Our traffic spans a range of loads: light (scale factors \{1, 2, 4, 8\}), medium (\{16, 32\}), and high (\{64, 128\}). At higher loads, more flows compete for traffic than at medium or light loads. We report results of over $640$ experiments, which capture different traffic and topology combinations.

%

\parab{Comparison to benchmarks (\figref{fig:traffic-engineering:all-samples:speedup-fairness} and \ref{fig:traffic-engineering:all-samples:efficiency}).} All algorithms in {\sysname} are faster than SWAN and Danna (\figref{fig:traffic-engineering:all-samples:speedup-fairness}).
Each approach is in a different color in this figure, and
each point corresponds to a single traffic demand on a single topology. The plot also shows the mean and standard deviations along the fairness and speedup axes.

We see the trade-off across these different max-min fair resource allocators: (a) Danna is optimal but also by far the slowest (on average taking $4.3\times$ longer than the second slowest algorithm, SWAN, under high-load); (b) 1-waterfilling is the fastest of the baselines but does not consider flow-level fairness ($30\%$ less fair than Danna on average but 4 orders of magnitude faster); (c) SWAN sits somewhere in between. It is faster than Danna (solves fewer optimizations), but slower than 1-waterfilling (1-waterfilling does not solve any optimization). It is fairer than 1-waterfilling but unlike Danna does not achieve optimal max-min fairness; (d) \sysname empirically Pareto-dominates these baselines as each of its algorithms provide a different point on the trade-off space.

Our algorithms are most effective under \emph{high loads} (arguably, speed \textit{and} fairness matter most). \sysname's Geometric Binner ({\sf GB}) is faster than SWAN by $4.5\times$ on average ($6\times$ in the $90^{th}$ percentile) because it only solves a single optimization. {\sf GB} also has worst-case fairness guarantees. The Equi-depth Binner ({\sf EB}) is faster than SWAN, slightly slower than {\sf GB}, and fairer than both. {\sysname}'s Approximate Waterfiller is even faster than 1-waterfilling (by an order of magnitude) with the same flow-level fairness. \sysname's Adaptive Waterfiller improves fairness ($19\%$ higher on average) at a slight speed reduction (still $21.4\times$ faster than SWAN on average). 

\figref{fig:traffic-engineering:all-samples:efficiency} compares the efficiency of different methods. Under low loads, all schemes are comparable. The differences become evident at higher loads, where {\sf EB} is approximately as efficient as Danna. {\sf GB} and SWAN are more efficient, likely because they sacrifice fairness.

\begin{figure}[t]
	\centering
	\subfigure[Fairness vs Run-time]{\centering
		\includegraphics[width=0.9\linewidth]{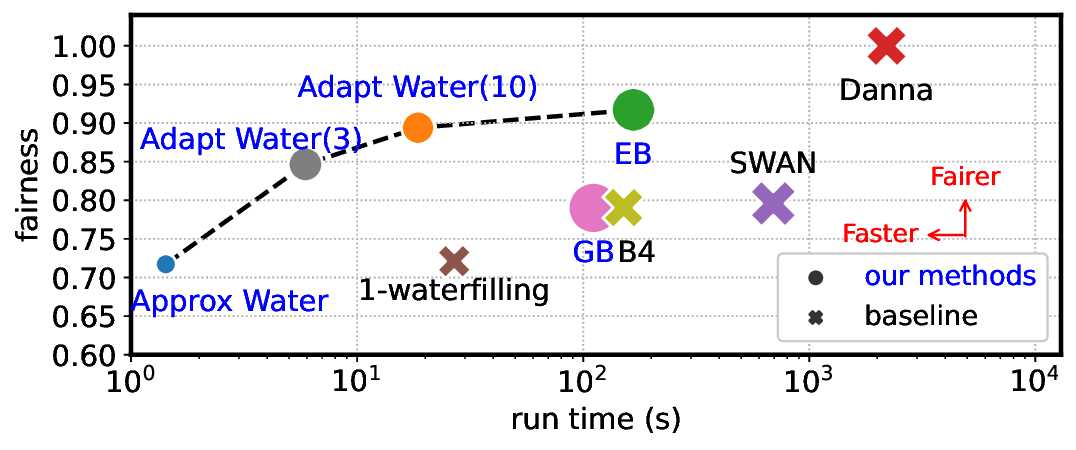}} \vspace{-3mm} \\
	\subfigure[Efficiency wrt Danna]{\centering
		\includegraphics[width=0.9\linewidth]{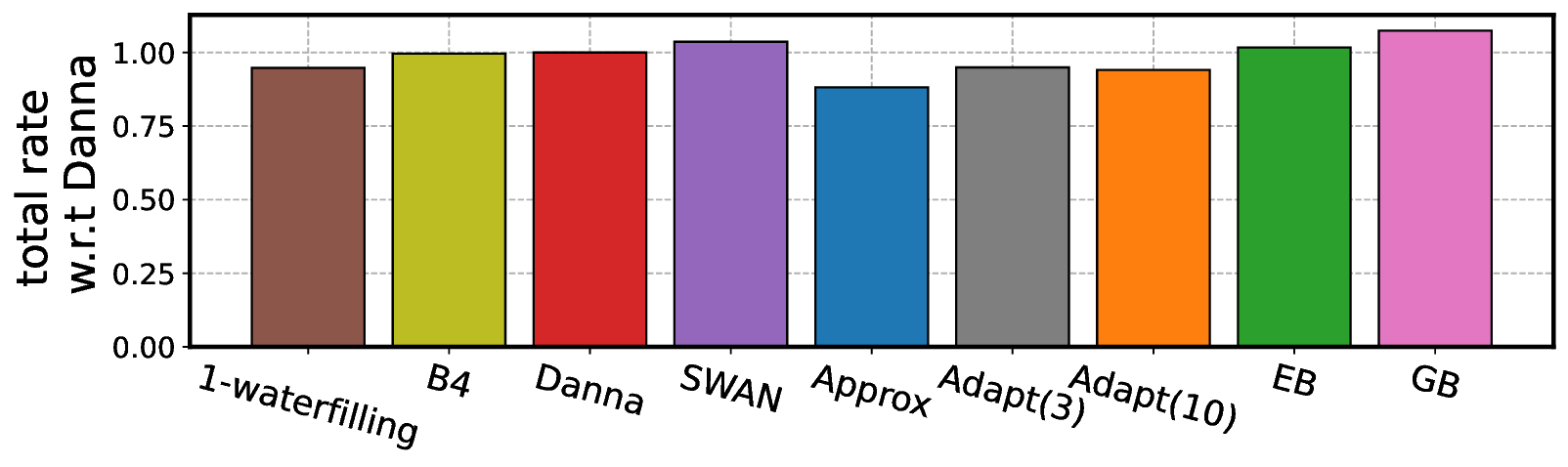}}
	\caption{\textbf{The empirical Pareto-dominance of \sysname over all of our baselines on an example topology (Cogentco) and an example workload with $64\times$ scale factor.} The size of the markers in (a) are in proportion to the efficiency of each algorithm --- we report exact comparisons in (b).}
	\label{fig:traffic-engineering:example}
      \end{figure}
      
We can see these differences more clearly when we focus on a single topology and workload in~\figref{fig:traffic-engineering:example}.
\sysname's allocators Pareto-dominate other approaches. The Approximate Waterfiller, Adaptive Waterfiller (number of iterations = 3 and 10), and {\sf EB} are faster than SWAN and Danna. Adaptive Waterfiller and {\sf EB} are also fairer than SWAN while having comparable efficiency. Operators can use {\sf GB} to get strong worst-case guarantees (at the cost of reduced fairness). 
B4~\cite{B4}'s TE algorithm is just as fast and fair as {\sf GB} (albeit slightly less efficient) but does not have fairness guarantees.
Note that we can control the fairness and runtime of {\sf GB} by tuning $\alpha$, whereas we can not control either in B4.


In summary, in settings where Danna is impractical, \sysname outperforms other TE algorithms (SWAN, 1-waterfilling, B4). Depending on the requirements, users can opt for Adaptive or Approximate Waterfillers, or {\sf EB} (or {\sf GB} if fairness guarantees are important). They can also customize the parameters in each allocator to further tune the balance.

\parab{Production deployment (\figref{fig:traffic-engineering:integration}).} We have successfully deployed \sysname in the production TE pipeline of Azure. Microsoft opted for {\sf GB} as it has the same fairness guarantees as their existing TE solver. \figref{fig:te:integration:month} shows cumulative density function (CDF) of the relative speed up of \sysname compared to the provider's previous allocator. These measurements are over a month-long deployment in a WAN with thousands of nodes. \sysname reduces the run-time on average by $2.4\times$ (up to $5.4\times$) without impacting fairness or efficiency.

We compare \sysname with the previous allocator on production demands at different loads (\figref{fig:te:integration:load}). \sysname's speedup increases with the load because the previous iterative solver invokes more optimizations at higher loads. \sysname's efficiency also increases because its $\epsilon$-trick can exploit minor fairness violations to improve efficiency. In all cases, \sysname is within 1\% of the previous solver's fairness.

\begin{figure}[t]
	\centering
	\subfigure[Speedup]{\includegraphics[width=0.55\linewidth]{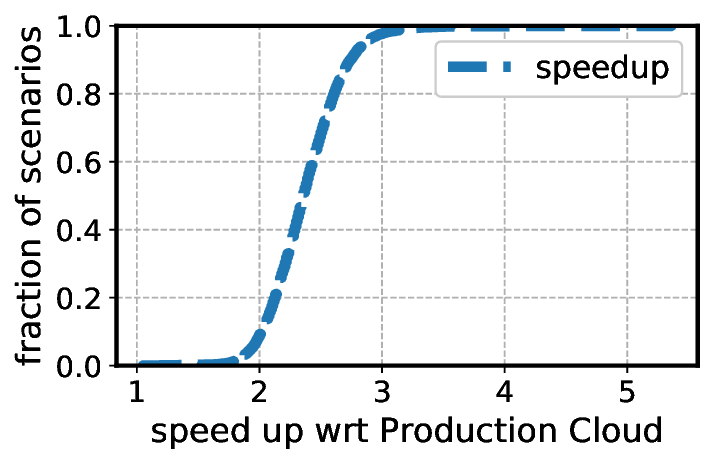} \label{fig:te:integration:month}} \vspace{-2mm} \\
	\subfigure[Impact of load]{\includegraphics[width=0.8\linewidth]{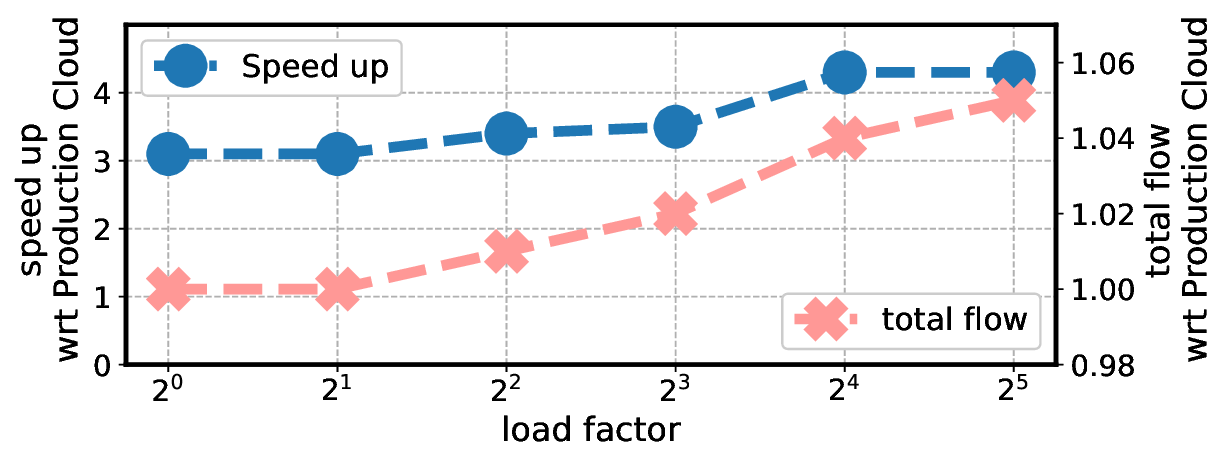} \label{fig:te:integration:load}}
	\caption{\textbf{Results from deploying \sysname in production.} (a) Month-long measurements show substantial speedup with no impact on efficiency or fairness compared to the provider's previous max-min fair allocator. (b) Using production traces, we show the benefit of \sysname improves as loads~\cite{Abuzaid-ncflow} increase.}
	\label{fig:traffic-engineering:integration}
%
	\centering
	\includegraphics[width=0.8\linewidth]{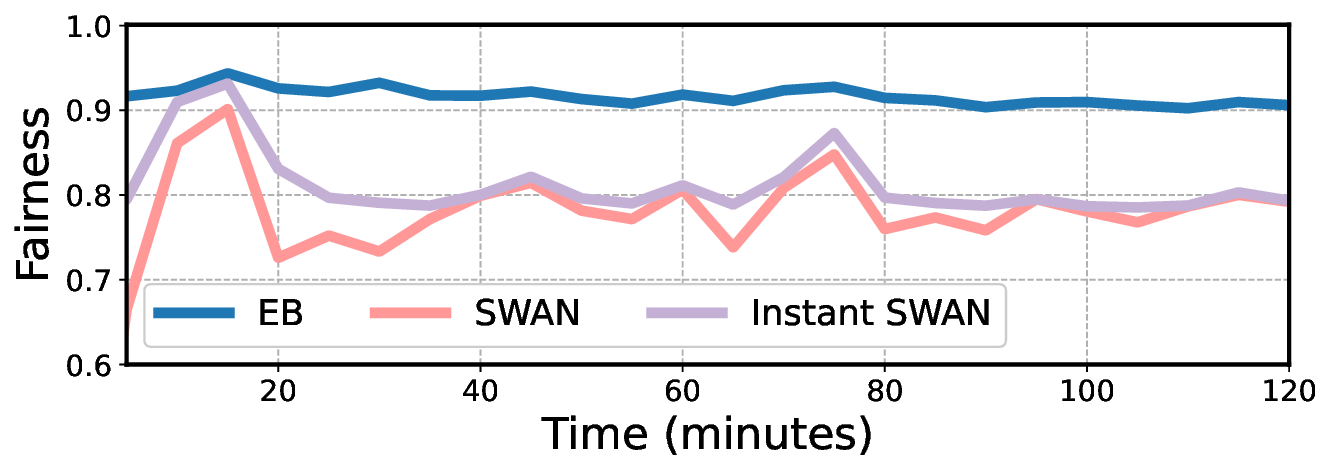}
	\caption{\textbf{Impact of solver runtimes on fairness when demands change.} SWAN can not react to the new demands quickly and faces another 10\% reduction in fairness whereas {\sf EB} can keep track of the changes. These results are on Cogentco following NCFlow's change distribution~\cite{Abuzaid-ncflow} on medium load traffics.}
	\label{fig:tracking-demands}
\end{figure}



\begin{figure}[t]
	\centering
	\subfigure[Fairness vs Run-time]{\centering
		\includegraphics[width=0.9\linewidth]{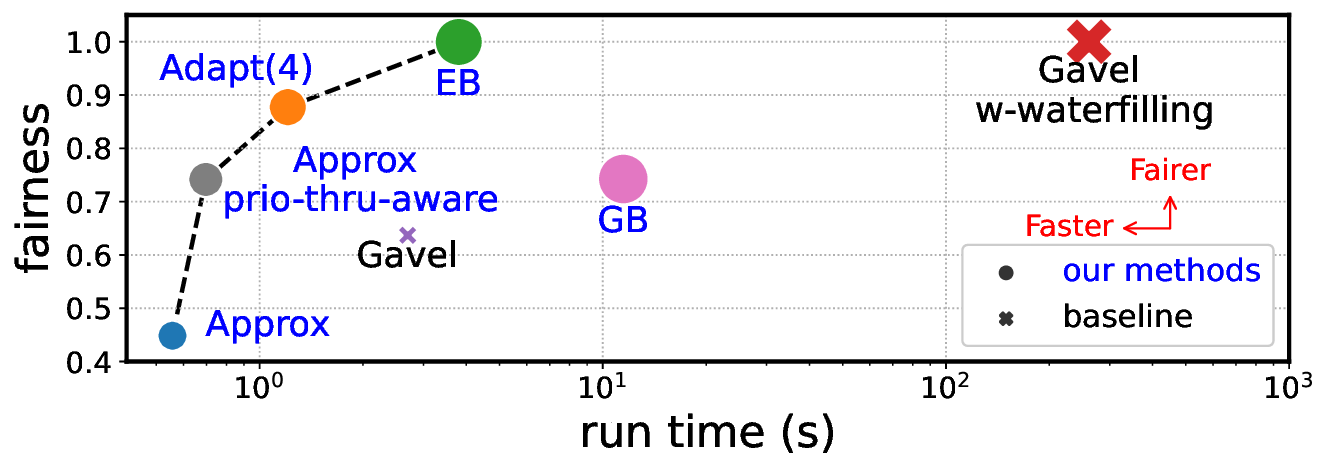}} \vspace{-3mm} \\
	\subfigure[Efficiency wrt Gavel w-waterfilling]{\centering
		\includegraphics[width=0.9\linewidth]{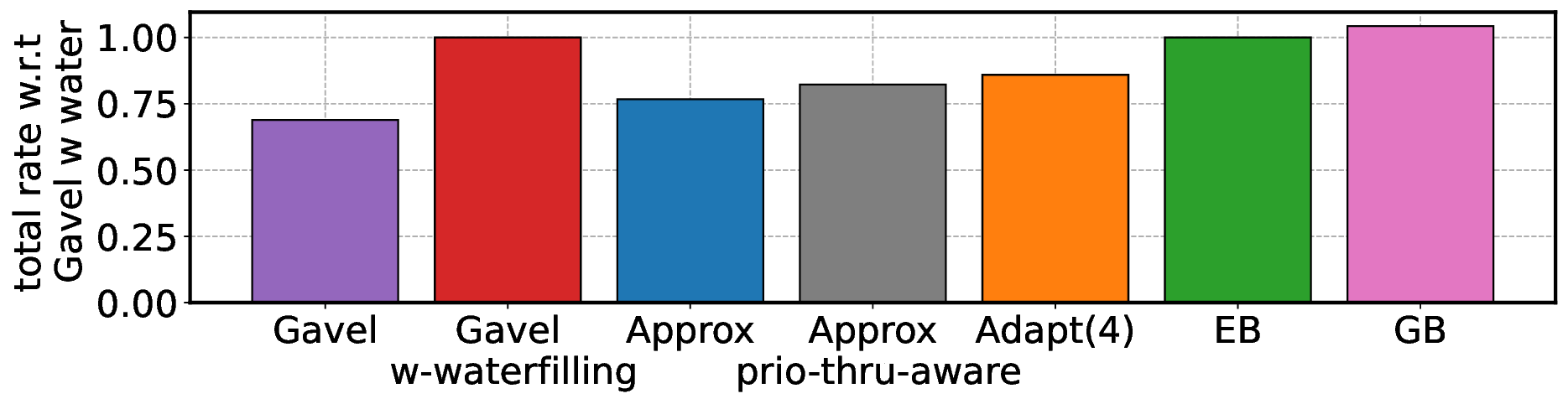}}
	\caption{\textbf{Trade-off between efficiency, fairness, and speed in CS on an example scenario (with 8192 jobs).} (a) shows the fairness vs run-time behavior of the different approaches; (b) shows the efficiency relative to the Gavel w waterfilling. Empirically, \sysname Pareto-dominates both variants of Gavel.}
	\label{fig:cluster-scheduling:example}
\end{figure}

\parab{Tracking Changing Demands (\figref{fig:tracking-demands}).} We evaluated each method on a sequence of traffic, arriving every five minutes (a window), starting from a medium load traffic demand. Our methodology is the same as NCFlow~\cite{Abuzaid-ncflow}. In this scenario, SWAN needs two windows to compute each allocation~--~it only computes allocations for half of the demands. This results in up to 10\% reduction in fairness compared to an instant SWAN (a hypothetical scheme that computes the allocation instantly). However, {\sf EB}\footnote{{\sf GB} is faster than {\sf EB}. If the latter can keep up, so can {\sf GB}. We have omitted an evaluation based on {\sf GB} for this reason.} reacts to changes quickly and meets all the deadlines. In general, SWAN's inability to keep track of demands leads to even higher unfairness than {\sf EB} (relative to what we reported in \figref{fig:traffic-engineering:all-samples:speedup-fairness}). Also, as we move from medium to high load, we expect the difference to be more as SWAN is even slower and needs to solve more optimizations.


\begin{figure}[t]
	\centering
	\subfigure[Convergence of the Adaptive Waterfiller]{\includegraphics[width=0.85\linewidth]{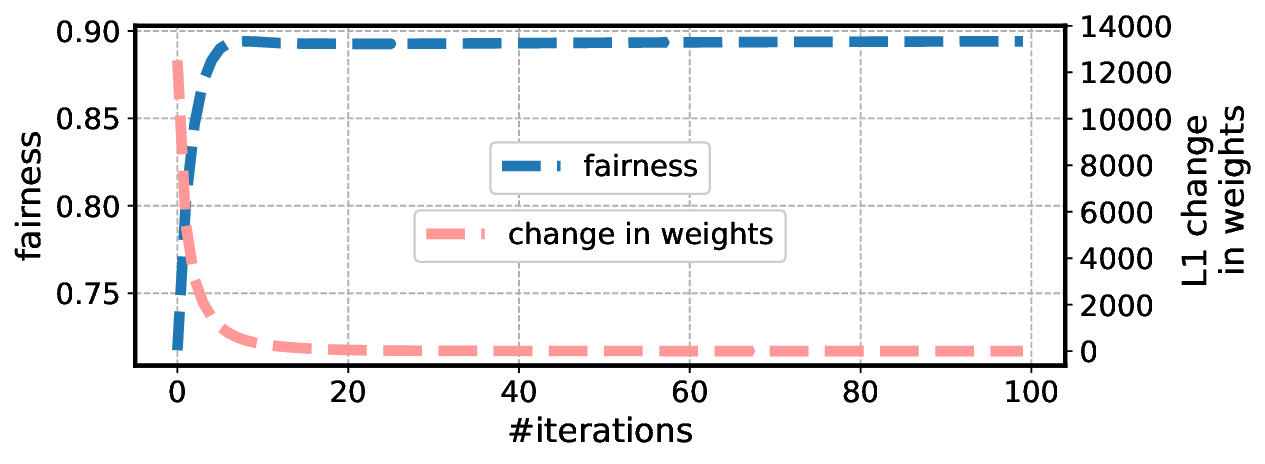} \label{fig:ablation:iter}} \vspace{-2mm} \\
	\subfigure[Impact of the number of bins on the fairness]{\includegraphics[width=0.8\linewidth]{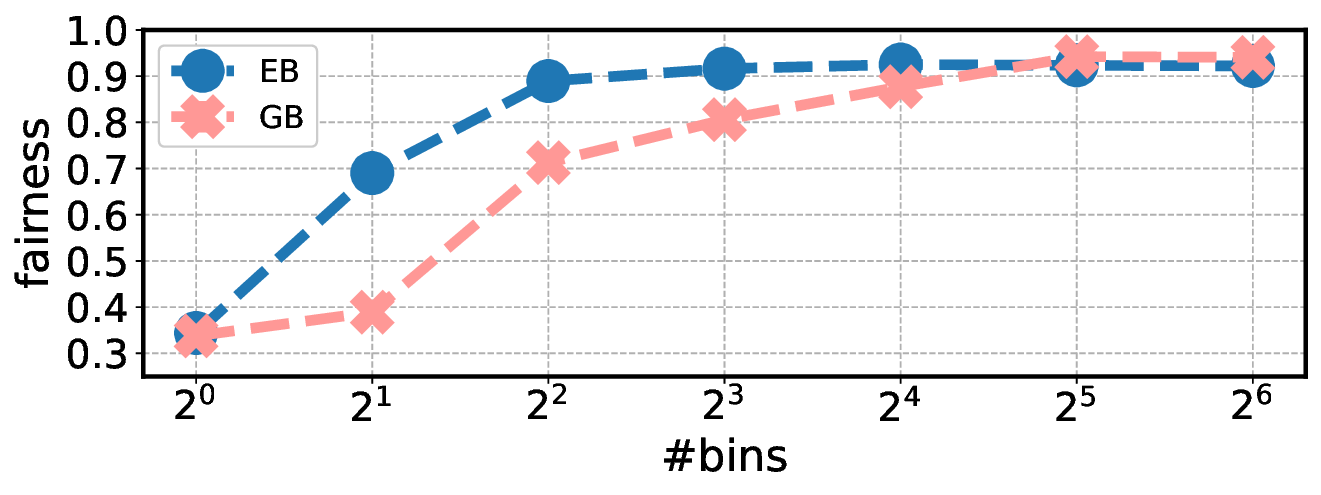} \label{fig:ablation:fairness}} \vspace{-2mm} \\
	\subfigure[Impact of the number of bins on the efficiency]{\includegraphics[width=0.8\linewidth]{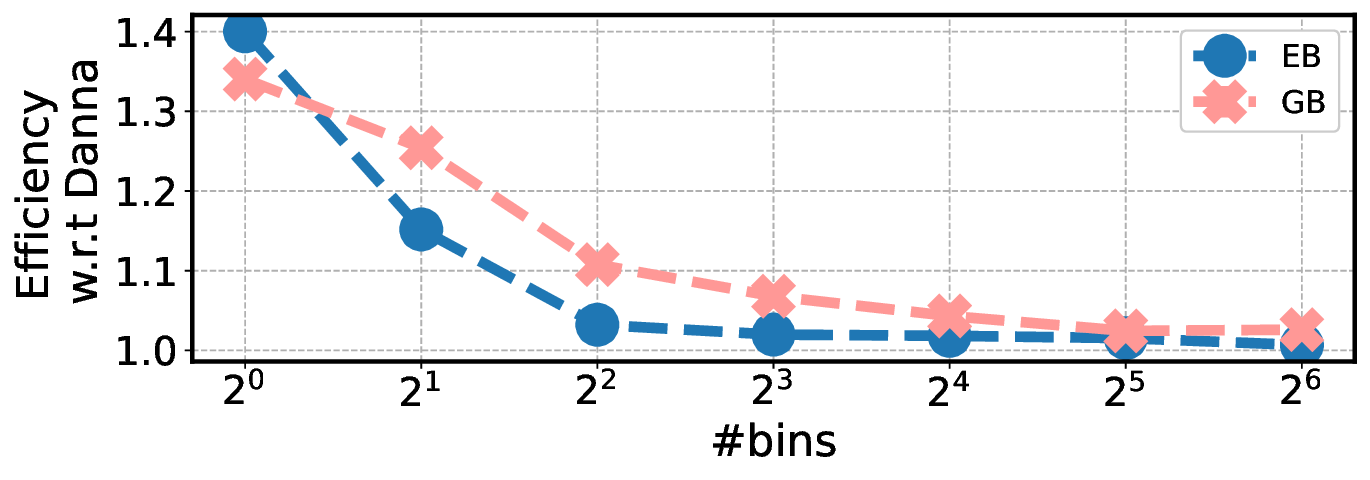} \label{fig:ablation:efficiency}} \vspace{-1mm} \\
	\caption{\textbf{Convergence and sensitivity analysis.} (a) Adaptive Waterfiller empirically converges within 5~--~10 iterations. (b, c) The number of bins controls the trade-off between fairness and efficiency in {\sf EB} and {\sf GB} (fewer bins lead to higher efficiency and lower fairness). Results are on the Cogentco topology and Gravity traffic distribution (scale factor = 64). (see~\figref{fig:bins:poisson} for Poisson)}
\end{figure}

\subsection{Cluster Scheduling}
\label{subsec:CS}

\noindent {\bf Experiment Setup.} We generate job requests from Gavel's job generator: we consider $3$ types of GPUs (V100, P100, K80) and uniformly sample jobs from the $26$ different job types available in Gavel (see \secref{sec::CS_eval_extended}).
Jobs are heterogeneous: they require a different number of workers (which we derive from the Microsoft public trace~\cite{microsoft-philly}) and have different priorities (which we sample uniformly from the set $\{1, 2, 4, 8\}$). 

\parab{Comparison to benchmarks.} We report results on over 40 different scenarios, which capture different number of available GPUs and competing jobs (see~\secref{sec::CS_eval_extended} for more details). Our results match our observations from WAN-TE; \sysname Pareto-dominates both Gavel and Gavel with waterfilling. We present these results in~\figref{fig:cluster-scheduling:all-samples:speedup-fairness} in~\secref{sec::CS_eval_extended} for space. 

We provide further insight into {\sysname}'s performance through an example scenario where 8192 jobs compete for resources (\figref{fig:cluster-scheduling:example}).
Adaptive Waterfiller outperforms standard Gavel in fairness, efficiency, and speed. For CS, {\sf GB} is slower than Gavel but fairer (more than $10\%$) and more efficient (more than $30\%$). We can augment Gavel with waterfilling~\cite{Gavel-Deepak} to improve it, but with a substantial slowdown.
In contrast, {\sf EB} provides comparable fairness and efficiency as Gavel with waterfilling and is $\sim$ 2 orders of magnitude faster.

\subsection{Convergence and Sensitivity Analysis}
\label{s:conv-sens-analys}

\noindent {\bf Convergence.} We empirically evaluate the convergence of the Adaptive Waterfiller. In~\secref{sec:app_proof}, we proved the algorithm in~\secref{s:genwaterfilling} only converges to and stops if it finds a bandwidth-bottlenecked allocation but may not converge if it does not find one.
We empirically find that Adaptive Waterfiller always converges. \figref{fig:ablation:iter} shows how its weights and fairness properties change with the number of iterations: the weights stabilize after 5 iterations. 

\parab{Impact of number of bins.} \figref{fig:ablation:fairness} and~\ref{fig:ablation:efficiency} show fairness and efficiency of binners ({\sf GB}/{\sf EB}) for different number of bins. \sysname uses this parameter to tune the trade-off between efficiency, fairness, and run-time. Using more bins increases fairness because the number of demands within each bin decreases but at the cost of higher run-time (more variables in the optimization). {\sf EB} is fairer than {\sf GB} for up to 16 bins because {\sf GB} suffers from bin-imbalance. However, {\sf GB} does not incur bin-imbalance for $\geq 32$ bins and both methods have roughly the same fairness. The slightly lower fairness of {\sf EB} is due to Adaptive Waterfiller making small mistakes when estimating the order of rates and influencing {\sf EB}'s binning.

\subsection{Other Experiments}
\label{s:exp:other}
\noindent {\bf Impact of number of paths.} We explore how sensitive our solutions are to the number of paths by varying this parameter and comparing our fairest methods (\ie Adaptive Waterfilling and {\sf EB}) to SWAN (\figref{fig:ablation:paths}). Increasing the number of paths improves the benefit of \sysname in both speedup and fairness. With more paths, each optimization of SWAN becomes more expensive, while Adaptive Waterfiller as well as {\sf EB} can exploit path diversity better to achieve higher fairness.

\begin{figure}[t]
	\centering
	\subfigure{\includegraphics[width=0.8\linewidth]{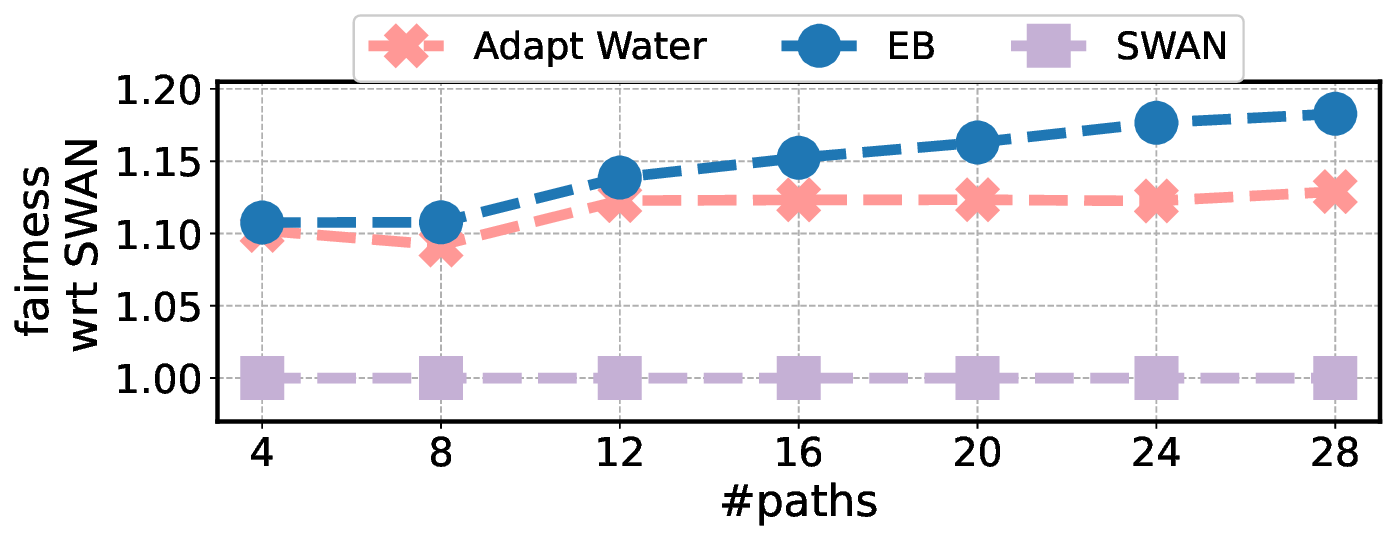}}
	\vspace{-5mm} \\
	\subfigure{\includegraphics[width=0.8\linewidth]{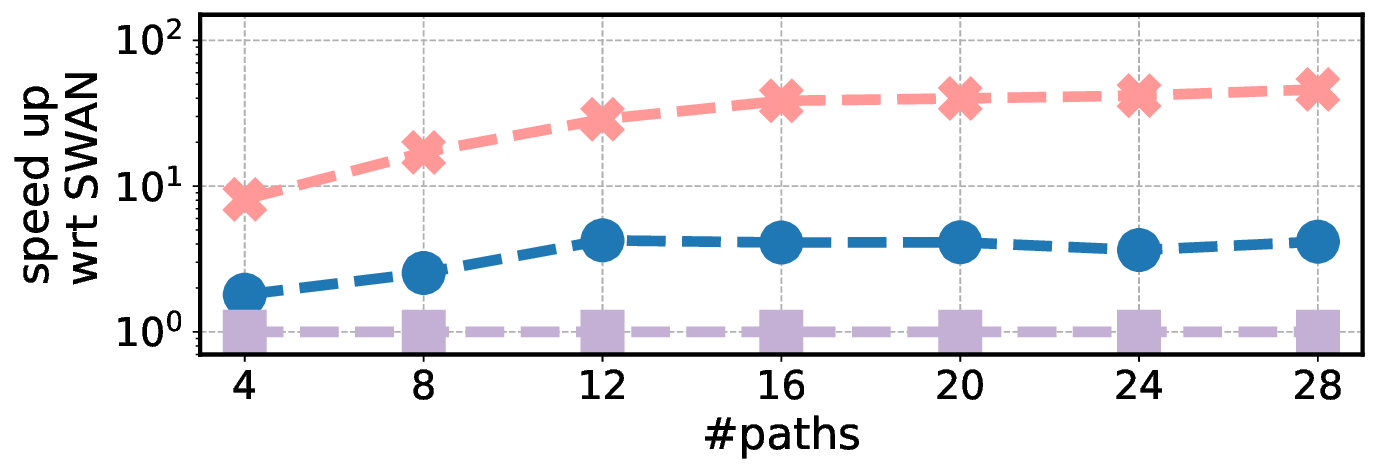}}
	\caption{\textbf{Increasing the number of paths improves the fairness and speedup of \sysname compared to SWAN.} Results are on the Cogentco topology and Gravity traffic distribution (scale factor = 64). (see~\figref{fig:ablation:paths:poisson} for Poisson)}
	\label{fig:ablation:paths}
	\centering
	\subfigure{\includegraphics[width=0.8\linewidth]{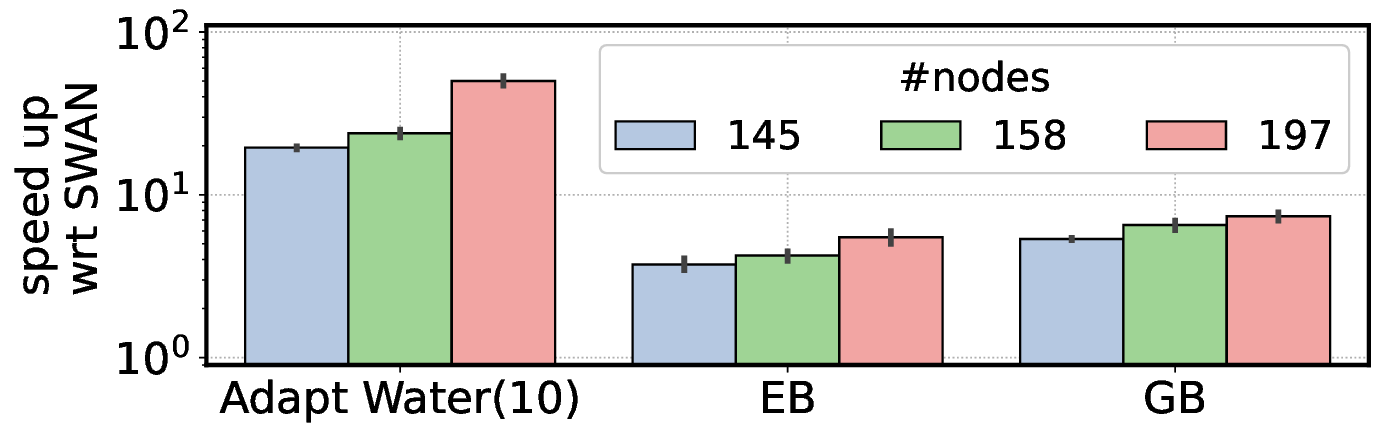}}
	\caption{\textbf{Impact of topology size.} \sysname's speed up relative to SWAN improves with the size of the topology.}
	\label{fig:ablation:toposize}
\end{figure}

\parab{Impact of topology size.} 
The benefit of \sysname's allocators increases with the topology size (\figref{fig:ablation:toposize}): SWAN needs to solve more optimizations for larger topologies while \sysname solves a fixed number of optimizations (=1 for {\sf EB}/{\sf GB} and =0 for adaptive waterfilling).

\parab{Comparison to NCFlow and POP.} NCFlow~\cite{Abuzaid-ncflow} and POP~\cite{Narayanan-POP} decompose the resource allocation problem to scale but do not directly address max-min fairness~\cite{TESurvery}. NCFlow only maximizes the total flow, and the authors mention in the paper that it is hard to extend it to max-min fairness objective~\cite{Abuzaid-ncflow}. Similarly, POP maximizes total flow and maximum concurrent flow (\ie the smallest fractional allocation) but does not provide any results on max-min fairness. To understand how POP compares to \sysname, we adapt both SWAN and \sysname to use it. We randomly divide demands (with client splitting as needed per POP's guidelines) among different partitions and run SWAN or \sysname in parallel on each partition (\figref{fig:decomp:pop}, extended evaluation in~\secref{sec:POP-eval-extended}).
	
We use {\sf GB} to ensure a fair comparison to SWAN: it has the same theoretical guarantees, is more than $10\times$ faster, and maintains the same level of fairness. When we apply POP to SWAN, we lose the worst-case guarantee~\cite{MetaOpt} and have to sacrifice over $10\%$ in fairness to achieve the same speed as \sysname. We also observe that applying POP to \sysname results in the same fairness as SWAN for the same number of partitions (but is also substantially faster).

\begin{figure}[t]
	\centering
	\subfigure{\includegraphics[width=0.77\linewidth]{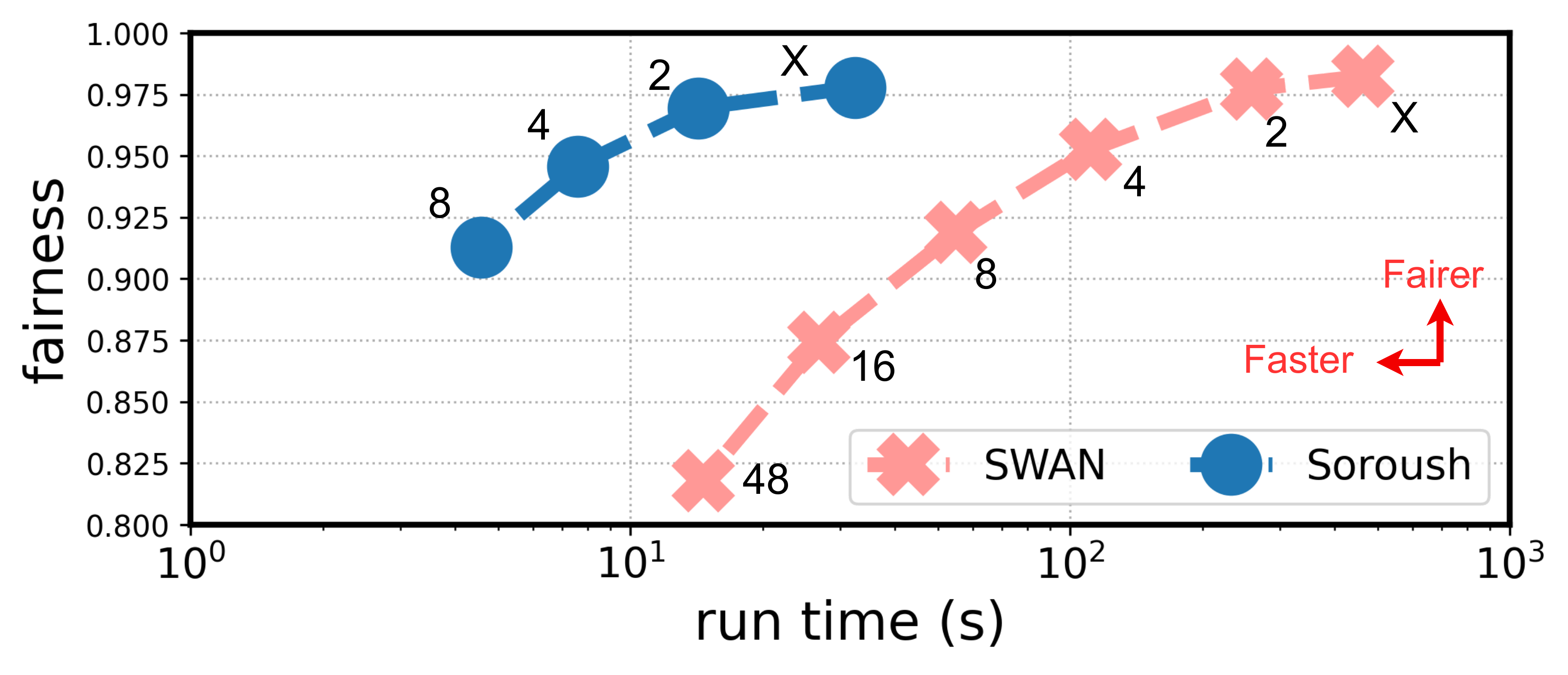}}
	\vspace{-2mm} \\
	\caption{\textbf{Impact of POP~\cite{Narayanan-POP}.} The results are on 3 randomly generated traffic, following Poisson distribution with a scale factor 64, on the Cogentco topology. Consistent with POP, we use client splitting (ratio$=0.75$) for this traffic distribution. ["X" indicates that POP is not used, and "Numbers" = number of POP partitions.]}
	\label{fig:decomp:pop}
\end{figure}

%% file: cr/discussion.tex
\section{Discussion}
\label{sec:discussion}

\sysname allows operators to adjust the trade-off between fairness, speed, and efficiency. We focus on multi-path allocations but our solutions apply to single-path settings too~\cite{Janus-Alipourfard,swarm}. Under this setting, our experiments show the Approximate Waterfiller is an order of magnitude faster than the fastest single-path allocator with only a slight decrease in efficiency. We defer the following to future work:

\parab{Other fairness metrics.} \sysname does not apply to other, less commonly used, fairness metrics~\cite{jain1999throughput, bonald2006queueing, bonald2004performance}.

\parab{Other problem domainss.} \sysname applies to any graph-based resource allocation problem which seeks to achieve max-min fairness. We demonstrate significant benefit of \sysname using examples from CS and WAN-TE. To use \sysname in other domains~\cite{Facebook-WAN-Risk-Xia,Janus-Alipourfard,s-PERC-Lavanya, schad2012max, mondal2022min, li2019max, gogu2014max,swarm}, users need to model the additional constraints in our graph model. We aim to provide tools to simplify this in future work.

\parab{Distributed extension.} \sysname applies to centralized resource allocation problems. Our future work aims to extend it to distributed settings~\cite{s-PERC-Lavanya,ros2001theory,cabinae,poseidon}.

%% file: cr/relatedworks.tex
\section{Related Work}
\label{s:related-work}

\noindent{\bf TE and CS resource allocation.} Prior approaches to both TE and CS aim to produce fast and efficient allocations~\cite{SWAN,Gavel-Deepak,B4,danna-practical-max-min,Pioro-Efficient-Max-Min-Fair,s-PERC-Lavanya,ghodsi2013choosy,danna-upward,Kandula-texcp,Gandiva-fair-Chaudhary,Allox-Tan}. In~\secref{sec:eval}, we show \sysname outperforms the sate-of-the-art in multi-resource max-min fair allocation (SWAN, Danna, B4, waterfilling, and Gavel).

Prior work employs ML in TE~\cite{LearningToRoute,DOTE,Teal} to optimize objectives that are already solved using a single LP (\eg max flow). These objectives are either convex or quasi-convex~\cite{DOTE}. However, the exact from of max-min fairness is sequential and we are unaware of any work that considers end-to-end training on a sequence of LPs. In fact, it may be more tractable for ML methods to learn our {\sf GB} which is a single LP. Applying ML to further speed up \sysname is an interesting future direction.

Recent work~\cite{Abuzaid-ncflow,Narayanan-POP} uses decomposition techniques to scale resource allocation problems. However, they focus on simpler objectives such as max flow or max concurrent flow that require single LPs and do not explicitly support max-min fairness. Extending NCFlow to max-min fairness is non-trivial as the authors mentioned~\cite{Abuzaid-ncflow}. We have extended POP to support max-min fairness and empirically compare it with \sysname. Our results show POP's performance depends on the traffic distribution, whereas \sysname works consistently well. POP also does not have worst-case fairness guarantees~\cite{MetaOpt}, whereas \sysname has allocators that do ({\sf GB}).


\parab{Algorithms for computing max-min fair rates.} Prior work has expanded our understanding of max-min fair resource allocation~\cite{nguyen2017approximate, radunovic2007unified}. These are largely theoretical and do not provide a practical and fast solution. Bandit-based solutions~\cite{bistritz2020my} lack worst-case guarantees and do not allow users to control the trade-off between fairness, efficiency and speed.

%% file: cr/conclusions.tex
\section{Conclusion}
\label{s:conclusions}

\sysname enables fast max-min fair resource allocation for graph-based problems such as traffic engineering and cluster scheduling. It provides a suite of allocators that Pareto-dominate state-of-the-art in both of these domains. Some of the allocators in \sysname have theoretical guarantees, and all of them have parameters for users to control the trade-offs. We have deployed \sysname in Azure's WAN traffic engineering pipeline. Future work can explore other applications and other notions of fairness.

%% file: cr/appendix.tex
\appendix

\renewcommand\thefigure{A.\arabic{figure}}
\renewcommand{\thealgocf}{A.\arabic{algocf}}
\renewcommand\thetable{A.\arabic{table}}
\setcounter{algocf}{0}
\setcounter{figure}{0}
\setcounter{table}{0}

\begin{figure}
	\centering
	\includegraphics[width=1.0\linewidth]{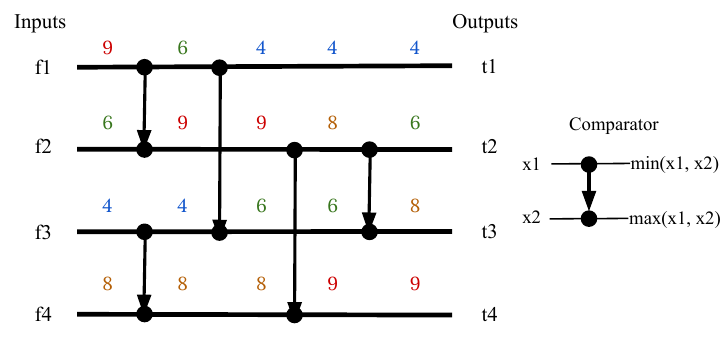}
	\caption{\textbf{Sorting Network Example.} }
	\label{fig:sorting-network-example}
	\vspace{2mm}
\end{figure}

\section{Max-min fair allocation optimization}
\label{sec:general-form}
\input{texfiles/appendix_model.tex}

\section{Closed-form max-min fair objective}
\label{sec:closed-form}

We present two closed form representations of the max-min fair objective: one exact, and one that converges to the max-min fair objective in the limit. The exact form is the following:

\begin{align}
\label{sec:exact_objective}
\text{fair}(\mathbf{f}) = \text{arg}\max_{\mathbf{f}} \bigcup\limits_{\{\set{F}_A \mid \set{F}_A \subseteq \mathbf{f}\}}\min(\{f_k \mid f_k \in \set{F}_A\})
\end{align}

Intuitively, this is a collection of maximization problems, where each maximizes the smallest flow in a given subset of $\mathbf{f}$ (a total of $2^{\mid\mathbf{f}\mid}$ maximizations). We next prove that this objective, in the instance that $\mathbf{f}$ are bounded, results in max-min fair allocations.

\begin{proof}
	
	Without loss of generality, we assume if $i < j$ then $f_i \le f_j$ for all $f_i, f_j \in \mathbf{f}$. 
	
	Suppose the theorem is not true: there exists an allocation $\mathbf{f}^*$ which is optimum as measured by the objective in~\ref{sec:exact_objective} but is not max-min fair. Three scenarios might have caused this;
	
	\textbf{Case 1.} A flow $i$ exists with unbounded $f_i^{*}$, which can not be true as we assume all the flows are bounded.
	
	\textbf{Case 2.} A flow $i$ exists that we can improve its rate without hurting other flows with $\leq$ rate. One of the constraints in~\eqnref{sec:exact_objective} is to maximize $f_{i}$ as a result such $i$ can not exist.
	
	\textbf{Case 3.} Two flows $i$ and $j$ exist ($i < j$) with optimal max-min fair rates of $\hat{f}_i$ and $\hat{f}_j$ such that $\hat{f}_j < f^{*}_j$ and $f^{*}_i < \hat{f}_i$. This means that in the solution from \eqnref{sec:exact_objective}, flow $j$ is receiving more than its share and is hurting flow $i$. This also can not happen since it violates one of the constraints in~\eqnref{sec:exact_objective} that maximizes the minimum of $i$ and $j$. (Note that this holds even if $\hat{f}_i = \hat{f}_j$ since maximizing the minimum of these two ensures they get equal rates.)
	
	As a result, each flow is guaranteed to be bounded, achieve its maximum possible rate, and can not hurt any other flow with less than or equal rate. This is the definition of max-min fairness ($\mathbf{f^{*}}$ is max-min fair).
\end{proof}

An alternate closed form representation of max-min fair is the following:

\begin{align}
\label{eq:limit_proof}
\text{fair}(\mathbf{f}) =\arg\max_{\mathbf{f}} \sum_{i} \epsilon^{\sum_{j \neq i} \mathbb{I}(f_i \le f_j)} f_i
\end{align}

We can prove this converges to the max-min fair rate allocations as $\epsilon \rightarrow 0$ similar to the proof of~\theoremref{th:epsilonOneShot}.

\section{SWAN as a sequence of LPs}
\label{sec:swan_form}
\input{texfiles/appendix_swan_formulation}

\section{Proofs of results for AdaptiveWaterfiller}
\label{sec:app_proof}
We present the proofs of the various results mentioned in~\secref{s:genwaterfilling} for Adaptive Waterfiller.

\subsection{Proof of Theorem~\ref{T:bandwith_bottleneck}}
\label{sec:proof_bottleneck}

If we denote by $\mathbf{f}(\mathbf{\theta})$, the solution of solving the weighted waterfilling sub-flow problem with weights $\mathbf{\theta} = \{\theta^p_{k}\}$, then convergence implies that
\begin{equation}\label{e:weights_convergence}
	\theta^p_{k} = \frac{f^p_{k}(\mathbf{\theta})}{f_{k}(\mathbf{\theta})},
\end{equation}
so that $\theta^p_{k}(t+1) = \theta^p_{k}(t)$ for all $p,k$.
From the definition of single-path weighted waterfilling, it must be that if $f^p_{k}$ is bottlenecked at link $l$, then $\frac{f^p_{k}}{\theta^{p}_{k}} \geq \frac{f^{\hat{p}}_{j}}{\theta^{\hat{p}}_{j}}$ for all non-zero $f^{\hat{p}}_{j}$ going through that link.
Using \eqnref{e:weights_convergence} to replace the weights in this inequality, it immediately follows that $f_k \geq f_j$.
Since this must hold for every $j$ such that there exists a non-zero subflow $f^{\hat{p}}_{j}$ going through link $l$, it must be that $\mathbf{f}$ is bandwidth-bottlenecked (see definition before Theorem~\ref{T:bandwith_bottleneck}).

\subsection{Other results}
\label{sec:additional_results_waterfilling}

In the discussion after Theorem~\ref{T:bandwith_bottleneck}, two results are stated without proof: the max-min fair rate allocation is bandwidth-bottlenecked and the adaptive waterfiller converges when it finds a bandwidth-bottlenecked rate allocation.
Here, we provide their proofs in the form of the two following lemmas:

\begin{lemma}
	If $\mathbf{f}$ is a max-min fair rate allocation then it must be bandwidth-bottlenecked.
\end{lemma}
\begin{proof}
	Suppose that this is not true and a max-min rate allocation is \emph{not} bandwidth-bottlenecked.
	This must mean that for some subflow $f^{p}_{k}$ bottlenecked on link $l$, there is another non-zero subflow $f^{\hat{p}}_{j}$ going through that link and $f_j > f_k$.
	This implies that we can increase the subflow $f^{p}_{k}$ at the expense of $f^{\hat{p}}_{j}$. Ultimately, this increases the allocation of $f_k$ without reducing the allocation of any other equal or smaller allocation (only reducing the allocation of $f_j$, which was larger to start with).
	We arrived at a contradiction since this violates the definition of max-min fair allocation.
\end{proof}

\begin{lemma}
	Every bandwidth-bottlenecked rate allocation $\mathbf{f}$ is a fixed point of the adaptive waterfiller algorithm.
\end{lemma}
\begin{proof}
	Assume that $\mathbf{f}$ is bandwidth-bottlenecked and we use these flows (and subflows) to construct weights $\theta^{p}_{k} = {f^p_{k}}/{f_k}$. Let us denote by $\tilde{\mathbf{f}}$ the solution of solving the weighted waterfilling with those weights. We want to show that $\mathbf{f} = \tilde{\mathbf{f}}$.
	Notice that the following must hold for a subflow $f^{p}_{k}$ bottlenecked at link $l$:
	\begin{equation}
		\frac{f^{p}_{k}}{\theta^{p}_{k}} = \frac{f^{p}_{k}}{f^{p}_{k}} f_k = f_k \geq f_j = f_j \frac{f^{\hat{p}}_{j}}{f^{\hat{p}}_{j}} = \frac{f^{\hat{p}}_{j}}{\theta^{\hat{p}}_{j}},
	\end{equation}
where the inequality follows from the definition of bandwidth bottleneck (prior to Theorem~\ref{T:bandwith_bottleneck}) and the equality after that one assumes that $f^{\hat{p}}_{j}$ is a non-zero subflow also going through link $l$.
Hence, we have established that for every $f^{p}_{k}$ bottlenecked at link $l$, it must hold that $\frac{f^{p}_{k}}{\theta^{p}_{k}} \geq \frac{f^{\hat{p}}_{j}}{\theta^{\hat{p}}_{j}}$ for all non-zero flows $f^{\hat{p}}_{j}$ going through that link. This implies that $\mathbf{f}$ is a solution to the weighted waterfilling problem. However, we denoted by $\tilde{\mathbf{f}}$ the solution to this problem. From uniqueness of the weighted waterfilling solution, it must be that $\mathbf{f} = \tilde{\mathbf{f}}$.
\end{proof}

\section{Equi-depth binner formulation}
\label{ss:eb_formulation}
\input{texfiles/appendix_eb_formulation}

\section{Expected Run-time Benefit of GB and EB}
\label{ss:expected-run-time}
\input{texfiles/appendix_expected_speedup.tex}

\section{Extended Evaluation}

In this section, we provide both additional experiment details as well as an extended evaluation of \sysname.

\subsection{Tuning benchmarks for performance}
\label{sec:benchmark-tuning}

\begin{itemize}[noitemsep,leftmargin=*]
\item We warm start SWAN's and Danna's optimizations for iterations $> 1$ to reduce the run-times. We further tune Gurobi's solver parameters using $5\%$ of the traffic demands to achieve the best run-time.

\item Our Danna's implementation is that of Figure 2 in~\cite{danna-practical-max-min} (\ie binary and linear search): we found this algorithm is faster than the other proposed by the same work (\ie binary then linear search in Figure 4) as it can find and eliminate more demand-constrained flows. 

\item Our modified K-waterfilling algorithm uses K=1 which is the fastest and most parallelizable version of the K-waterfilling~\cite{s-PERC-Lavanya}.

\item For cluster scheduling (CS), we changed {\sysname}'s implementation to use CVXPY~\cite{cvxpy} to match Gavel's implementation and ensure fair run-time comparisons.

\end{itemize}



\subsection{Evaluation on CS}
\label{sec::CS_eval_extended}
\input{texfiles/appendix_cs_experiments}

\input{texfiles/appendix_extended_ablation}

\begin{figure}
	\centering
	\includegraphics[width=0.8\linewidth]{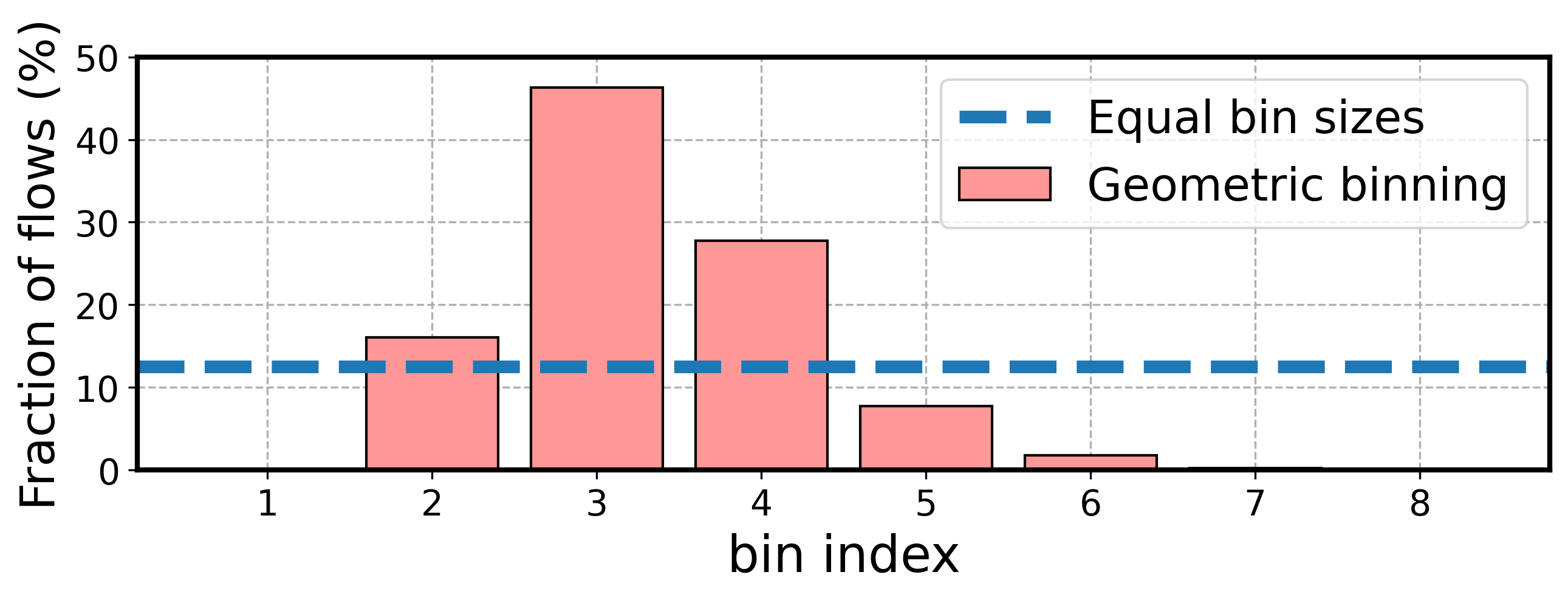}
	\caption{\textbf{Example of imbalanced bins in {\sf GB} for the TE usecase.} }
	\label{fig:imbalanced-bins:app}
\end{figure}

\subsection{Evaluation of POP}
\label{sec:POP-eval-extended}
\input{texfiles/appendix_pop_experiments}

\cut{
\input{texfiles/waterfillingpseudocodes}

\renewcommand{\arraystretch}{1.2}
\begin{figure}[t]
	\centering
	\label{fig:example:water}
	\subfigure[Global vs. Sub-flow level max-min fair]{\includegraphics[width=1.0\linewidth]{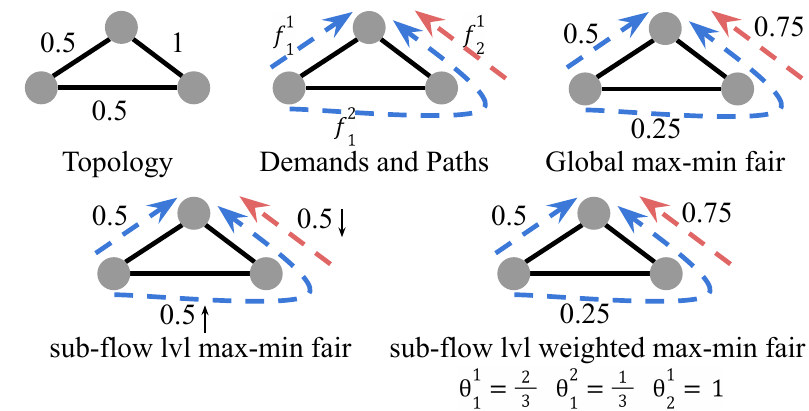}}
	
	\subfigure[Illustrating Adaptive Waterfilling]{
			\begin{tabular}{l|cccccc}

				& \multicolumn{6}{l}{\bf \# iteration $t\longrightarrow$}\\\hline

				$\theta^1_1$ 	& $\frac{1}{2}$ & $\frac{3}{5}$ 	& $\frac{7}{11}$ 	& $\frac{15}{23}$ 	& $\frac{31}{47}$ 	& $\frac{63}{95}$ \\

				$\theta^2_1$ 	& $\frac{1}{2}$ & $\frac{2}{5}$ 	& $\frac{4}{11}$ 	& $\frac{8}{23}$	& $\frac{16}{47}$	& $\frac{32}{95}$\\

				$\theta^1_2$ 	& $1$ 			& $1$				& $1$ 			 	& $1$				& $1$				& $1$\\

				$f^1_1$ 		& $\frac{1}{2}$ & $\frac{1}{2}$		& $\frac{1}{2}$ 	& $\frac{1}{2}$		& $\frac{1}{2}$		& $\frac{1}{2}$\\

				$f^2_1$ 		& $\frac{1}{3}$ & $\frac{2}{7}$		& $\frac{4}{15}$ 	& $\frac{8}{31}$	& $\frac{16}{63}$	& $\frac{32}{127}$\\

				$f^1_2$ 		& $\frac{2}{3}$ & $\frac{5}{7}$		& $\frac{11}{15}$ 	& $\frac{23}{31}$	& $\frac{47}{63}$	& $\frac{95}{127}$\\

			\end{tabular}
	
	}
	\caption{Example.}
\end{figure}
\renewcommand{\arraystretch}{1}
}

%% file: texfiles/appendix_model.tex
\begin{table*}[t!]
	\centering
	{\footnotesize
		\renewcommand{\arraystretch}{1.5}
		\begin{tabular}{l p{2.5in} p{2.0in} p{1.3in}}
			{\bf Term} & {\bf Meaning} & {\bf CS} & {\bf WAN-TE}\\
			\hline
			\hline
			
			$\mathbf{w}$ & inverse of coefficient vector in the objective of fair(.) where $w_k$ indicates the weight for the $k$-th demand and encodes the desired proportional max-min fair allocations. & priority of job $k$ (In~\cite{Gavel-Deepak} = user specified priority $\times$ effective average throughput $/$ number of workers) & priority of different services (\eg search and ads)\\
			\hline
			$q^p_{k}$ & the utility obtained by demand $k$ when assigned 1 unit on path $p$. & progress rate of the $k$-th job when assigned 1 unit on server $p$ & $ = 1$\\
			\hline
			$r^e_{k}$ & capacity consumed on resource $e$ (\ie link or GPU) when allocating 1 unit to demand $k$  & capacity consumed by job $k$ from resource $e$ (CPU, GPU or memory) & $ = 1$\\
			\hline
			$c_e$ & {capacity of resource $e \in \mathcal{E}$} & capacity of CPU, GPU or any other resources on a server & capacity of link $e$ \\
			\hline
			$d_k$ & {the resource requested by the $k$-th demand} & job $k$'s requested duration of time {($= 1$ in \cite{Gavel-Deepak})} & flow $k$'s requested rate \\
			\hline
			$f_{k}, f^p_{k}$ & $f_k: $ demand $k$'s total utility \newline $f^p_{k}: $ demand $k$'s obtained allocation from path $p$ & $f_k:$ job $k$'s total progress rate\newline $f^p_{k}$: fraction of time server $p$ is assigned to job $k$ & $f_{k}$: flow $k$'s total rate \newline $f^p_{k}$: flow $k$'s rate on path $p$ \\
			\hline\hline
		\end{tabular}
	}
	\caption{Additional notation for the general multi-resource max-min fair formulation in \secref{sec:general-form}. \label{t:notation-general}}
\end{table*}

\sysname offers a range of general algorithms that can solve any max-min fair resource allocation problem expressed using the model described in~\secref{s:unified-formulation}. In this section, we present the formulation behind this model.
\tabref{t:notation-general} describes all the notations, their meanings, and their mappings to WAN Traffic Engineering (WAN-TE) and Cluster Scheduling (CS). 

\parab{Feasible Allocation.} Given a set of demands and a set of paths over a group of resources, an allocation is feasible if it satisfies demand and capacity constraints. We define the feasible allocation as:
{\small
	\begin{alignat}{4}
		\multispan{4}\mbox{$\displaystyle \label{eq:feasible-alloc} {\sf FeasibleAlloc}(\mathcal{E}, \mathcal{D}, \mathcal{P}, \mathcal{Q}, \mathcal{R}) \triangleq \big\{\mathbf{f} \mid$}&\\
		&f_k = \sum_{p \in \mathcal{P}_k} q^p_k f_k^p, & \forall k \in \mathcal{D} & \text{(allocation for demand $k$)}\nonumber\\
		& \sum_{p \in \mathcal{P}_k}{f^p_k} \leq d_k, & \forall k \in \mathcal{D} & \text{(allocation below volume)}\nonumber\\
		& \sum_{k,p \mid p \in \mathcal{P}_k, e \in p} r^e_{k}f^p_{k} \le C_{e} & \forall e \in \mathcal{E}  & \text{(allocation below capacity)} \nonumber\\		
		&f_k^p \geq 0 & \forall p \in \mathcal{P}_k, k \in \mathcal{D} & \text{(non-negative allocation)}\big\}\nonumber
	\end{alignat}
}

\parab{Max-Min Fairness.} Among all the feasible allocations, the optimal max-min fair solution seeks:
{\small
	\begin{align}
		\mbox{\sf OptFair}(\mathcal{E}, \mathcal{D}, \mathcal{P}, \mathcal{Q}, \mathcal{R}) \quad \triangleq \quad & \arg \max_{\mathbf{f}} \quad \text{fair} (\mathbf{f} / \mathbf{w}) \\
		\mbox{s.t.} \quad & \mathbf{f} \in {\sf FeasibleAlloc}(\mathcal{E}, \mathcal{D}, \mathcal{P}, \mathcal{Q}, \mathcal{R}) \nonumber
	\end{align}
}
where the function $\text{fair}(\mathbf{x})$ encodes the max-min fairness objective. To our knowledge, prior works do not present a closed form of this function.  In~\secref{sec:closed-form}, we introduce two potential candidates (one exact and one that converges in the limit).


%% file: texfiles/appendix_swan_formulation.tex
The original formulation of the $b^{\mbox{\tiny{th}}}$ iteration of SWAN~\cite{SWAN} is the following:
{\small
		\begin{align}
			\label{eq:seq-swan}
			& \mbox{\sf SWANMaxMin}_{\textcolor{maroon}{b}}(\mathcal{E}, \mathcal{D}, \mathcal{P})~\triangleq && \arg\max \sum_{k \in \mathcal{D}} \textcolor{maroon}{f_{kb}} \\
			& \hspace{5mm} \mbox{s.t.} \quad && \hspace{-25mm}\textcolor{maroon}{f_{kb} \leq U\alpha^{b-1},} & \textcolor{maroon}{\forall k \in \mathcal{D}} \nonumber \\
			&&& \hspace{-25mm} \textcolor{maroon}{f_{kb} ~ \left\{\begin{array}{cl}
					\rowcolor{whiteRowColor} = f_{k(b-1)}  & \mbox{if}~f_{k(b-1)} < U\alpha^{b-2} \\
					\rowcolor{whiteRowColor} \geq f_{k(b-1)} & \mbox{otherwise} \\
				\end{array}\right. ,} & \textcolor{maroon}{\forall b > 1}\nonumber \\
			&&& \hspace{-25mm} (\mathbf{f}_{1}, \dots, \mathbf{f}_{b-1}) \in {\sf SWANMaxMin}_{b-1}, \nonumber \\
			&&& \hspace{-25mm} \mathbf{f} \in {\sf FeasibleAlloc}(\mathcal{E}, \mathcal{D}, \mathcal{P}).\nonumber
		\end{align}
}
where $f_{kb}$ is the total allocated rate to demand $k$ up to iteration $b$.

%% file: texfiles/appendix_eb_formulation.tex
In this section, we present two variants of Equi-depth binner~---~one where boundaries are elastic but the demands get allocation from exactly one bin, and the other where the bin boundaries are fixed but the demands are allowed to get allocation from multiple bins.

\parab{Equi-depth binner with elastic bin boundaries.} In this variant, we use the output of {\sf AdaptiveWaterfiller} to sort demands by their estimated max-min fair rates. We then divide the demands from smallest to largest into equal-sized bins ($D_b$), each assigned to one specific bin. The order of bins is maintained using bin boundaries $\ell_{b}$, which are determined by the optimization. During the allocation process, we prioritize bins with smaller demands, following a similar linearization technique described in~\secref{s:oneshotopts}:

\vspace{-0.1in}
{\small
	\begin{align}
		\label{eq:equibinning}
		& \mbox{\sf ElasticBoundaryEquiBinning}(\mathcal{E}, \mathcal{D}, \mathcal{P})~\triangleq\\
		&&&	\hspace{-45mm}\arg\max_{\mathbf{f},\textcolor{maroon}{\mathbf{\ell}}} \textcolor{maroon}{\sum_{\mbox{bins}~b}~\sum_{k \in \mathcal{D}_b}\epsilon^{b - 1} f_{k}} \nonumber \\
		&\hspace{5mm} \mbox{s.t.} \quad && \hspace{-35mm} \textcolor{maroon}{f_{k} < {\ell}_b + s_b,} &  \textcolor{maroon}{\forall b < N_\beta,\forall k \in \mathcal{D}_b}\nonumber\\
		&&& \hspace{-35mm} \textcolor{maroon}{f_{k} \ge {\ell}_{b-1},} & \textcolor{maroon}{\forall b > 1, \forall k \in \mathcal{D}_b}\nonumber \\
		&&&\hspace{-35mm} \textcolor{maroon}{{\ell}_b \ge 0,} & \textcolor{maroon}{\forall b} \nonumber\\
		&&&\hspace{-35mm} \mathbf{f} \in {\sf FeasibleAlloc}(\mathcal{E}, \mathcal{D}, \mathcal{P}).\nonumber 
	\end{align}
}
where $N_\beta$ is the number of bins, $D_b$ denotes the set of demands in bin $b$, $\ell_b$ shows the boundary of bin $b$ (determined by the optimization), and $s_b$ is the slack in quantization boundary of bin $b$ (input to the optimization).

\parab{Equi-depth binner with multi-bin allocations.} In this variant, we use the output of {\sf AdaptiveWaterfiller} to compute the bin boundaries $\ell_b$ that result in roughly the same number of demands per bin. Then, we reuse the {\sf Geometric Binner}'s formulation from~\eqnref{eq:new_geobinning} but with the estimated bin boundaries instead of geometrically increasing sizes:

\vspace{-0.1in}
{\small
	\begin{align}
		\label{eq:equibinning_v2}
		& \mbox{\sf MultiBinEquiBinning}(\mathcal{E}, \mathcal{D}, \mathcal{P}, \textcolor{maroon}{\{\ell_b\}})~\triangleq & \\ 
		&&& \hspace{-40mm} \arg\max_{\mathbf{f}} \sum_{k \in \mathcal{D}}~\textcolor{maroon}{\sum_{\mbox{bins} ~ b} \epsilon^{b-1} f_{kb}} \nonumber \\
		& \hspace{5mm} \mbox{s.t.} \quad && \hspace{-35mm} \textcolor{maroon}{f_{k} = \sum_{\mbox{bins} ~ b} f_{kb},} & \textcolor{maroon}{\forall k \in \mathcal{D}} \nonumber \\
		&&& \hspace{-35mm} \textcolor{maroon}{f_{k1} \leq \ell_1,} & \textcolor{maroon}{\forall k \in \mathcal{D}} \nonumber \\
		&&& \hspace{-35mm} \textcolor{maroon}{f_{kb} \leq \ell_{b} - \ell_{b-1},} & \textcolor{maroon}{\forall b > 1,\forall k \in \mathcal{D}}\nonumber \\
		&&& \hspace{-35mm} \mathbf{f} \in {\sf FeasibleAlloc}(\mathcal{E}, \mathcal{D}, \mathcal{P}).\nonumber 
	\end{align}
}

%% file: texfiles/appendix_expected_speedup.tex
Solving a linear program (with \#constraints = $\Omega$(\#variables) -- holds for resource allocation problems such as TE and CS) has worst-case time complexity of $O(\nu^{a})$ where $a \approx 2.373$~\cite{solvingLP} and $\nu$ is the number of variables in the optimization.

One can argue that simply solving a single optimization (as in the case of {\sf EB} and {\sf GB}) does not guarantee lower run-times compared to solving multiple optimizations (\eg SWAN). Solving multiple optimizations adds a multiplicative term to the time complexity. However, a naive single-shot optimization may use too many additional variables and ends up being slower. In this part, we theoretically analyze the expected run-time benefit of {\sf GB} and {\sf EB}. We show that the speed up of \sysname's optimization-based methods is due to their carefully designed approaches that only require a small number of additional variables compared to each optimization in the multi-shot variant.

\parab{SWAN} uses 1 variable per demand per path to demonstrate the allocation from an specific path ($\nu = PK$ where $P$ is the number of paths and $K$ is the number of demands). Therefore, if SWAN needs $N^{S}_\beta$ iterations, its worst-case complexity is $O(N^S_\beta P^{a} K^a)$.

\parab{GB} needs 1 extra variable per demand per bin to measure the allocation from each bin ($\nu = (N^G_\beta + P)K$ -- note that the number of bins in {\sf GB} ($N^G_\beta$) is the same as the number of iterations in SWAN ($N^S_\beta$)). This leads to a worst-case complexity of $O((N^G_\beta+ P)^a K^a)$. Compared to SWAN, the run-time saving of {\sf GB} is proportional to $N_\beta\big[1 + \frac{N_\beta}{P}\big]^{-a}$. For $P=16$ paths and $N_\beta = 8$ bins, we expect {\sf GB} to be $\sim 3.06 \times$ faster.

Our analysis of {\sf GB} is only valid when the number of bins is small. When there are many bins, {\sf GB}'s allocation may have too many zero variables. For example, if we have 128 bins, and a demand only needs allocation from the first bin, the remaining 127 bin variables will be zero. Existing solvers such as Gurobi~\cite{gurobi} exploit this sparsity to improve their runtime.

\parab{EB}\footnote{We analyze the variant with elastic bin boundaries. The other multi-bin variant has the same complexity as {\sf GB} since it uses the same optimization.} only uses 1 extra variable per bin to show the bin boundaries ($\nu = N^E_\beta + PK $). Therefore, its worst-case complexity is $O((N^E_\beta + PK)^a)$. Compared to SWAN, {\sf EB} has a run-time saving proportional to ${N^{S}_\beta}\big[1 + \frac{N^{E}_\beta}{PK}\big]^{-a}$. Since the number of demands is usually substantially larger than the number of bins $N^E_\beta$, we can approximate the run-time saving by ${N^{S}_\beta}$. For $N^S_\beta = 8$, we expect {\sf EB} to be $\sim 8 \times$ faster.

Although theoretical analysis can help us understand the speedup of \sysname, solvers can typically perform better than their worst-case and can take advantage of the structure of the optimization (\eg sparsity). In \secref{sec:eval}, we empirically show that the speedup of {\sf GB} compared to SWAN is larger than what our theoretical analysis predicted. In fact, it is even faster than {\sf EB}. Similarly, we found that the speed up of {\sf EB} compared to SWAN is only slightly less than the theoretical analysis.

%% file: texfiles/appendix_cs_experiments.tex

\begin{table}
	\footnotesize
	\centering
	\begin{tabular}{llcc}
		\hline
		\rowcolor{rowColorOne}
		Model & Task & App/Dataset & \#Batch sizes \\
		\hline
		\hline
		\hline
		\rowcolor{rowColorTwo}
		ResNet-18 & Image  & CIFAR-10~\cite{CIFAR10Dataset} & 16, 32, 64, \\
		\rowcolor{rowColorTwo}
		\cite{ResNet,Pytorch-CIFAR10} & Classification &  & 128, 256\\ 
		
		\rowcolor{rowColorOne}
		ResNet-50 & Image & ImageNet~\cite{ImageNet} & 16, 32, \\
		\rowcolor{rowColorOne}
		\cite{ResNet,Pytorch-ImageNet} & Classification & & 64, 128\\
		
		\rowcolor{rowColorTwo}
		CycleGAN & Image-to-Image & monet2photo~\cite{CycleGAN} & 1  \\
		\rowcolor{rowColorTwo}
		\cite{CycleGAN, Pytorch-CycleGAN} & Translation & & \\ 
		
		\rowcolor{rowColorOne}
		LSTM & Language & Wikitext2~\cite{Wikitext2} & 5, 10, 20, \\
		\rowcolor{rowColorOne}
		\cite{LSTM} & Modeling & & 40, 80 \\

		\rowcolor{rowColorTwo}
		Transformer & Language & Multi30k~\cite{Multi30k} & 16, 32, 64, \\
		\rowcolor{rowColorTwo}
		\cite{Attention-Transformer, Pytorch-Transformer} & Translation & (de-en) & 128, 256 \\
				
		\rowcolor{rowColorOne}
		A3C~\cite{RL-A3C,Pytorch-RL-A3C} & Deep RL & Pong & 4 \\ 
		
		\rowcolor{rowColorTwo}
		 &  & & 512, 1024, \\
		\rowcolor{rowColorTwo}
		Autoencoder & Recommendation & ML-20M~\cite{MV-20M} & 2048, 4096, \\ 
		\rowcolor{rowColorTwo}
		\cite{Autoencoder} & & & 8192 \\
		
		\hline                           
	\end{tabular}
	\caption{Type of jobs used for the evaluation of \sysname. We use Gavel's job generator~\cite{Gavel-Deepak}.}
	\label{tab:jobs:summary}
\end{table}

\parab{Experiment details.} We consider 3 different types of GPUs (V100, P100, K80) and base the number of GPUs on the number of jobs. We randomly select jobs from a set of 26 different job types available in Gavel (see~\tabref{tab:jobs:summary}). Each job has a specific priority and requires a certain number of workers. We determine the number of workers by randomly sampling from the distribution obtained from the Microsoft public trace~\cite{microsoft-philly} -- $70\%$ of jobs need a single worker, 25\% need between 2 and 4, and the remaining 5\% need 8. We also sample job priorities uniformly from $\{1, 2, 4, 8\}$. To compute the weights in the weighted max-min fair objective, we use the job throughput estimations from Gavel.

\parab{Results.}
In~\figref{fig:cluster-scheduling:all-samples:speedup-fairness}, we compare \sysname to two variants of Gavel (with and without waterfilling) in 40 different scenarios. The number of competing jobs in each scenario is selected from the set $\{1024, 2048, 4096, 8192\}$ -- 10 scenarios from each. Following Gavel, we set the number of each type of GPU to one-fourth of the total number of jobs and generate each job using the methodology explained above.

The results are in line with our observations from WAN-TE; (a) Gavel with waterfilling is an optimal max-min fair algorithm, but it is more than $20\times$ slower than other methods, (b) Gavel is $\sim 100\times$ faster than the Gavel w waterfilling but at the cost of $\sim 40\%$ drop in fairness and $\sim 15\%$ drop in efficiency, and (c) {\sysname}'s algorithms empirically Pareto-dominate both Gavel and Gavel w waterfilling.

We also observe that {\sf GB} is slower than the rest of the methods except (Gavel with waterfilling), but it provides worst-case per-demand fairness guarantees.

\begin{figure}
	\centering
	\subfigure[Fairness vs Speed up.]{\includegraphics[width=1.0\linewidth]{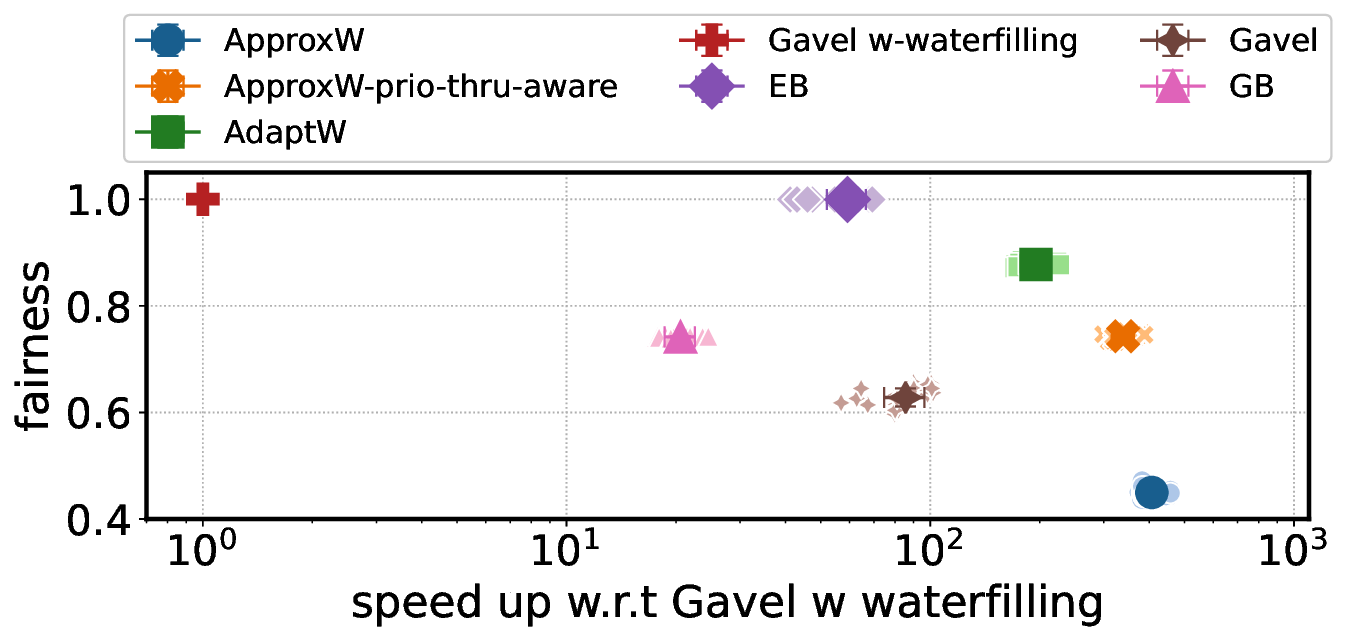}}
	\subfigure[Efficiency (total effective rate)]{\includegraphics[width=1.0\linewidth]{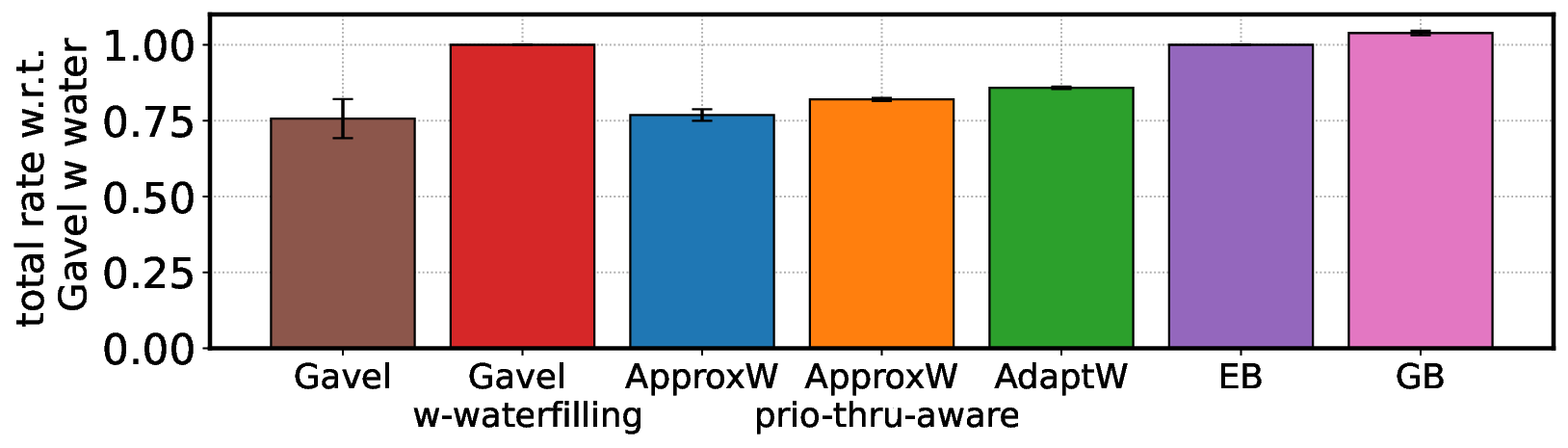}}
	\caption{\textbf{\sysname empirically Pareto-dominates Gavel}. These results are on 40 different scenarios with varying number of jobs and GPUs.}
	\label{fig:cluster-scheduling:all-samples:speedup-fairness}
\end{figure}

%% file: texfiles/appendix_extended_ablation.tex
\begin{figure}[t]
	\centering
	\subfigure[Impact of the number of bins on the fairness]{\includegraphics[width=0.8\linewidth]{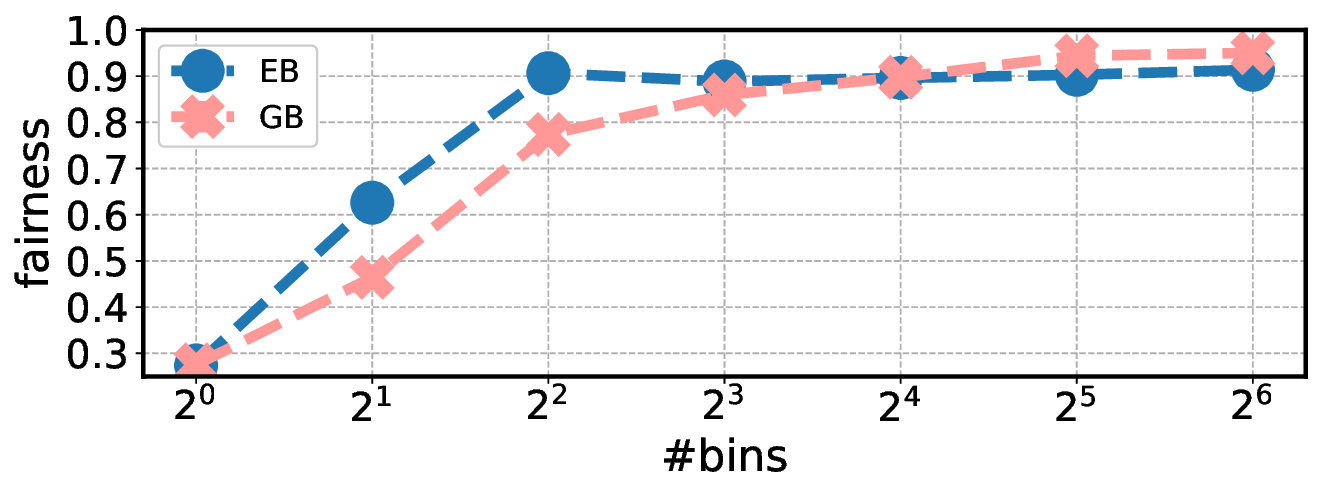} \label{fig:bins:poisson:fairness}} \vspace{-2mm} \\
	\subfigure[Impact of the number of bins on the efficiency]{\includegraphics[width=0.8\linewidth]{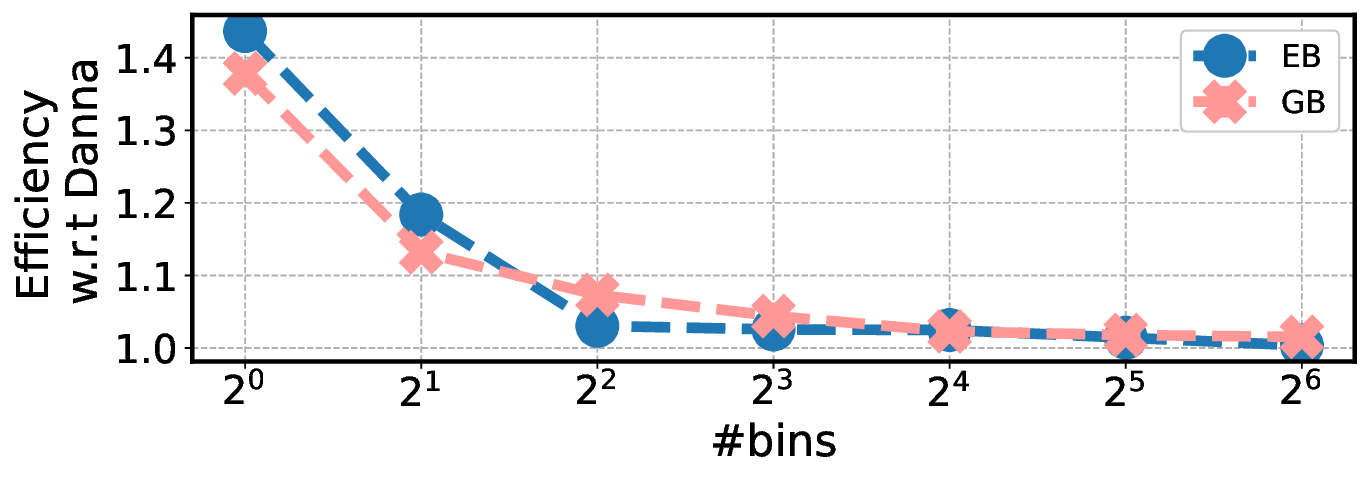} \label{fig:bins:poisson:efficiency}}
	\caption{\textbf{Impact of number of bins on fairness and efficiency of {\sf GB} and {\sf EB}}. These results are on the Cogentco topology and Poisson traffic distribution (scale factor = 64). \label{fig:bins:poisson}}
\end{figure}

\begin{figure}[t]
	\centering
	\subfigure{\includegraphics[width=0.8\linewidth]{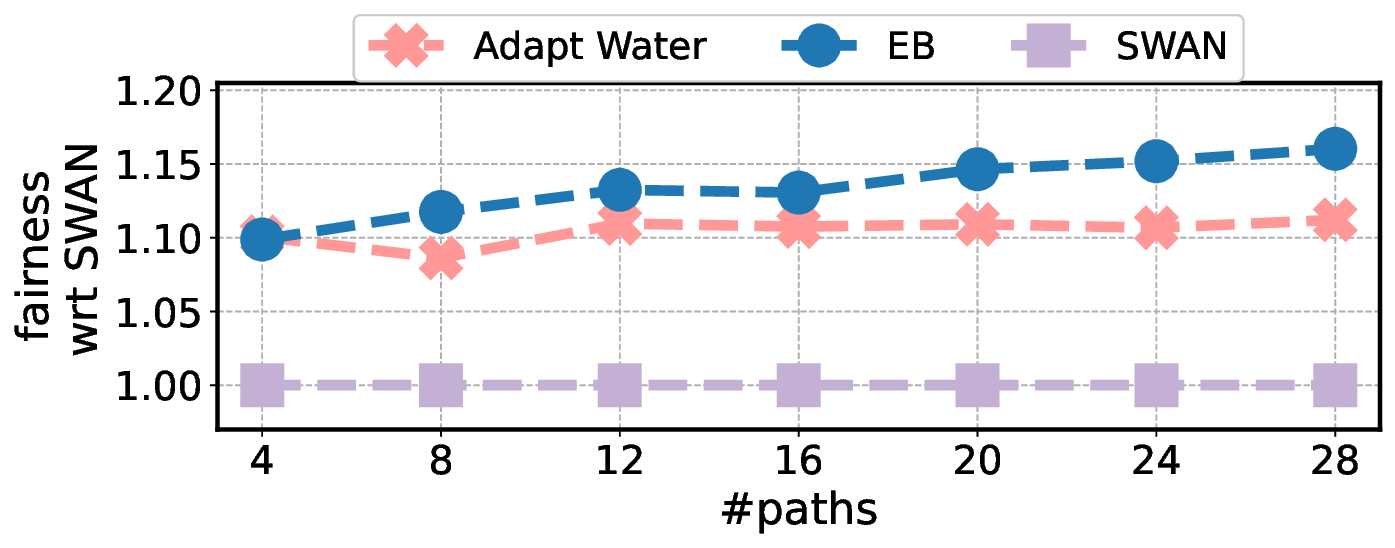}}
	\vspace{-5mm} \\
	\subfigure{\includegraphics[width=0.8\linewidth]{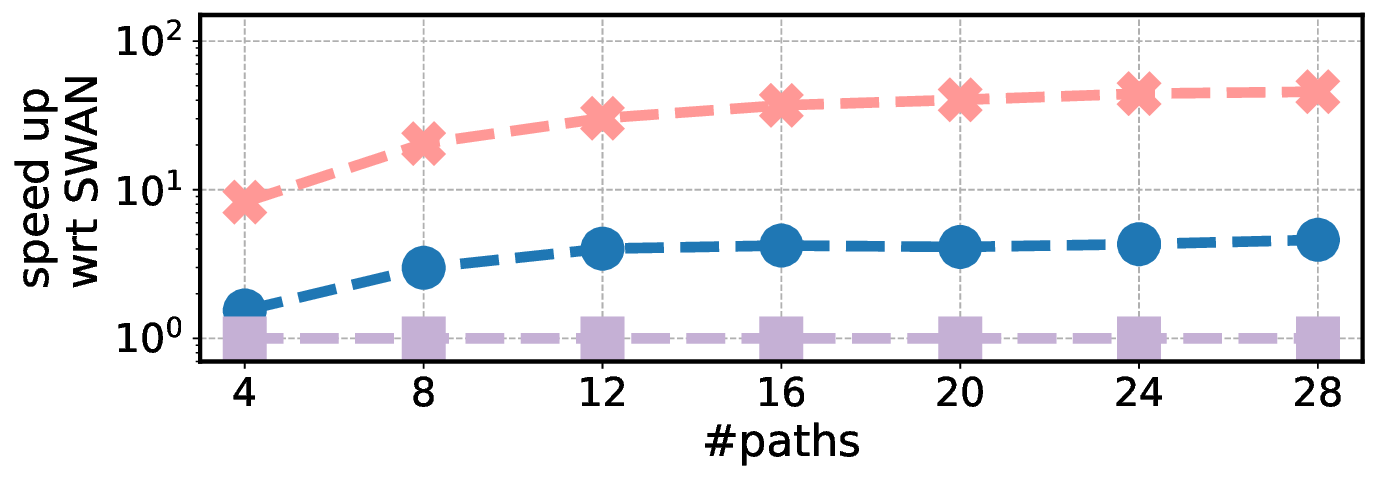}}
	\caption{\textbf{Impact of number of paths in TE.} These results are on the Cogentco topology and Poisson Traffic distribution (scale factor = 64).}
	\label{fig:ablation:paths:poisson}
\end{figure}

%% file: texfiles/appendix_pop_experiments.tex
POP~\cite{Narayanan-POP} is a decomposition technique used to scale granular resource allocations. It involves dividing demands uniformly at random into partitions, assigning an equal share of each resource to each partition, and then, solving the resource allocation for each partition in parallel. This procedure is called \textit{resource splitting}.

For large demands, POP incorporates an additional method called \textit{client splitting}, where demands are divided among multiple partitions.~\cite{Narayanan-POP} recommends using client splitting for Poission traffic distribution, as this can improve resource utilization. However, it is unnecessary for the other distributions such as Gravity.

POP focuses on objectives such as maximizing utilization or maximizing the minimum allocation, a different objective than max-min fairness. To assess the impact of POP on max-min fairness, we adapt POP to work with both \sysname and SWAN; We use the same procedure for partitioning the problem (resource splitting and client splitting as needed). We then allocate resources in each partition using a max-min fair solver such as SWAN or \sysname.

\parab{Theoretical Guarantee.} POP results in losing all the theoretical guarantees ($\alpha$-approximation for \sysname and SWAN). In fact,~\cite{MetaOpt} shows a substantial worst-case optimality gap for POP. However, \sysname is faster than SWAN because of its single optimization reformulation, while maintaining the same theoretical guarantees.

\parab{Empirical Evaluation.}
In~\figref{fig:decomp:pop:app}, we evaluate the performance of POP when applied to \sysname (specifically, GB) and SWAN for two different topologies, two load factors and two traffic distributions.

We find that the performance of POP depends on the traffic distribution whereas \sysname maintains the same level of fairness in all cases while being up to $15\times$ faster than SWAN. For distributions with granular demands such as Gravity, POP speeds up both SWAN and \sysname with only a minor drop in fairness. However, using POP causes significant fairness degradation -- more than 10\% to match the run-time of \sysname -- for traffic distributions such as Poisson that are not granular and require client splitting to avoid resource under-utilization. This unfairness is a result of per-partition max-min fairness in POP, which differs from global max-min fairness in \sysname or SWAN.

We also observe that applying POP to \sysname results in the same fairness level as SWAN for the same number of partitions. In each partition, \sysname is guaranteed to produce an allocation similar to SWAN (see~\secref{s:oneshotopts}), and therefore, the aggregated allocation is guaranteed to be the same. Since \sysname is faster in each partition, the overall run-time of \sysname + POP is substantially lower.

\begin{figure*}[t]
	\centering
	\subfigure[Topo: Cogentco, Traffic: Poisson, Scale factor=16, w client splitting]{\includegraphics[width=0.45\linewidth]{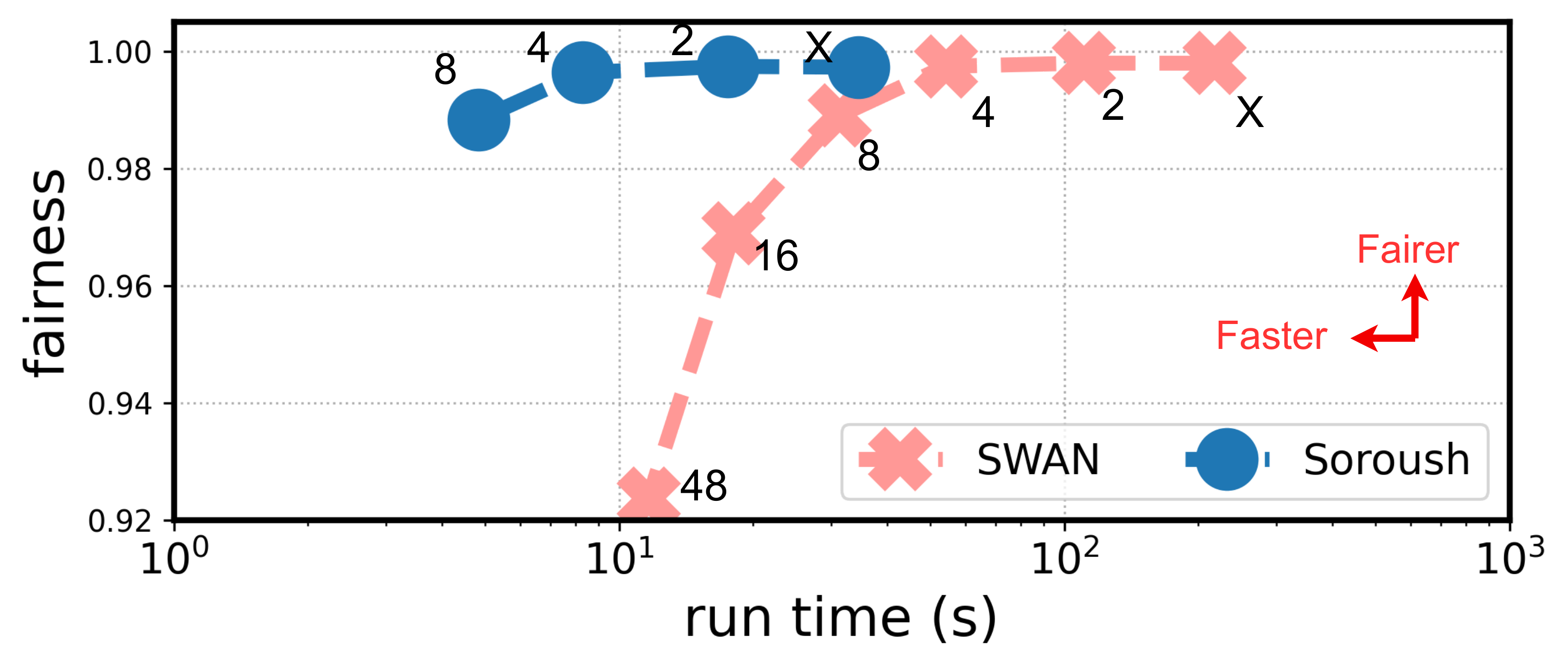}} \hfill
	\subfigure[Topo: Cogentco, Traffic: Poisson, Scale factor=64, w client splitting]{\includegraphics[width=0.45\linewidth]{figs/traffic_engineering/pop/POP-Cogentco-poisson-64.pdf}}
	\vspace{-1mm} \\
	\subfigure[Topo: Cogentco, Traffic: Gravity, Scale factor=16, no client splitting]{\includegraphics[width=0.45\linewidth]{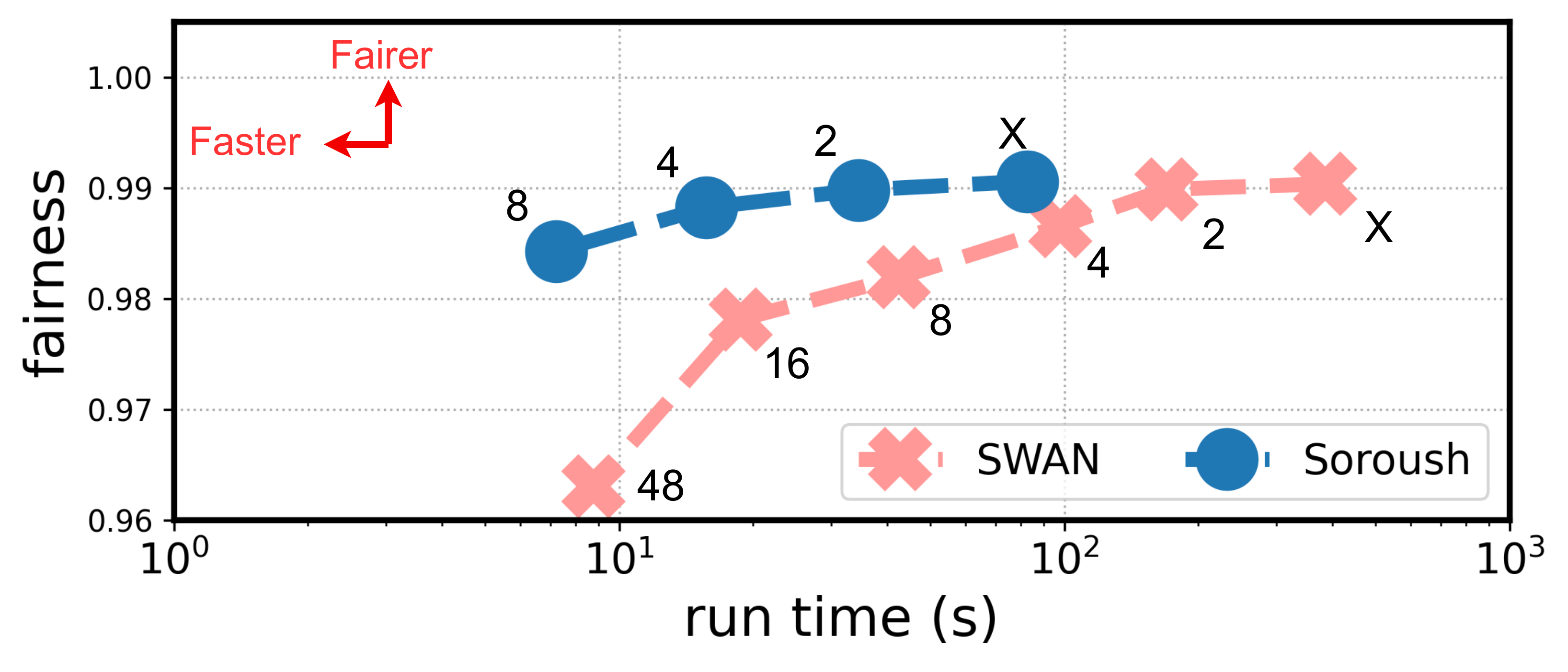}} \hfill
	\subfigure[Topo: Cogentco, Traffic: Gravity, Scale factor=64, no client splitting]{\includegraphics[width=0.45\linewidth]{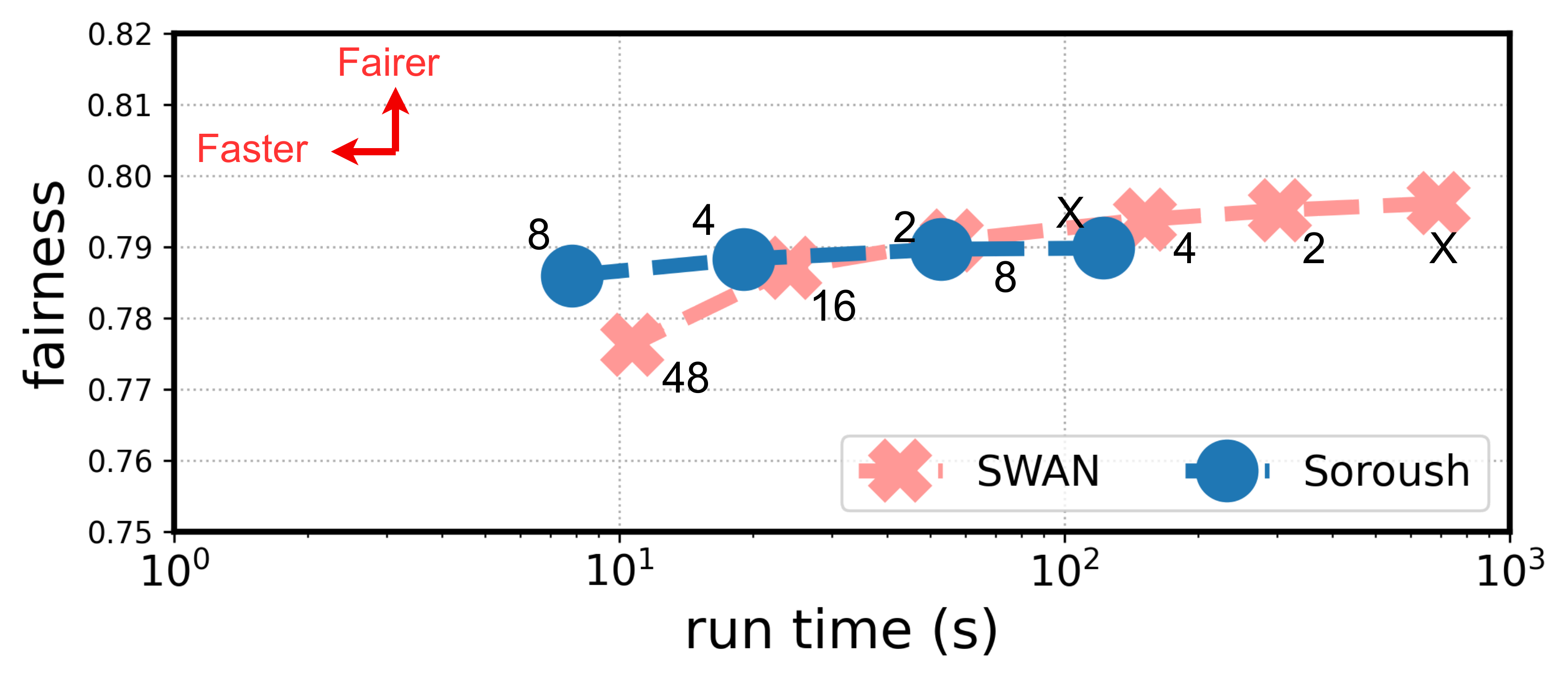}}
	\vspace{-1mm} \\
	\subfigure[Topo: GtsCe, Traffic: Poisson, Scale factor=16, with client splitting]{\includegraphics[width=0.45\linewidth]{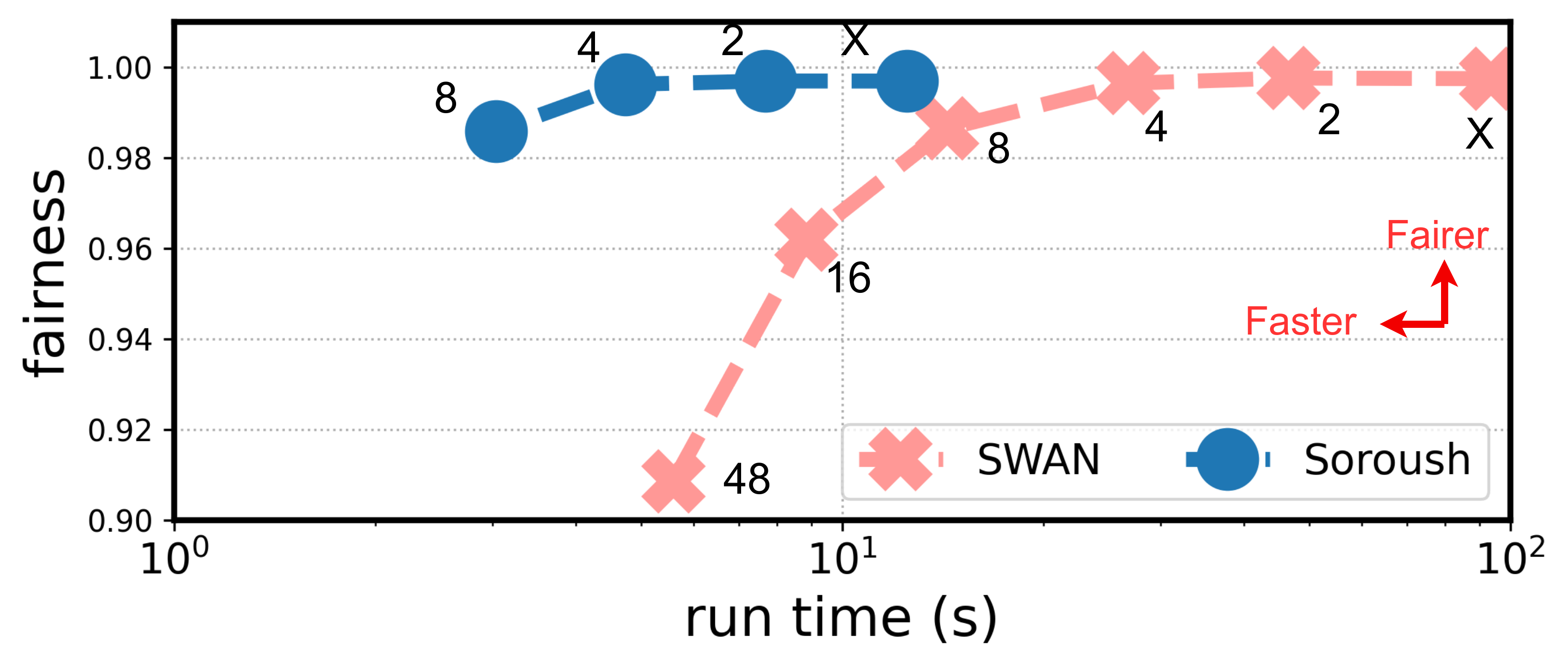}} \hfill
	\subfigure[Topo: GtsCe, Traffic: Poisson, Scale factor=64, with client splitting]{\includegraphics[width=0.45\linewidth]{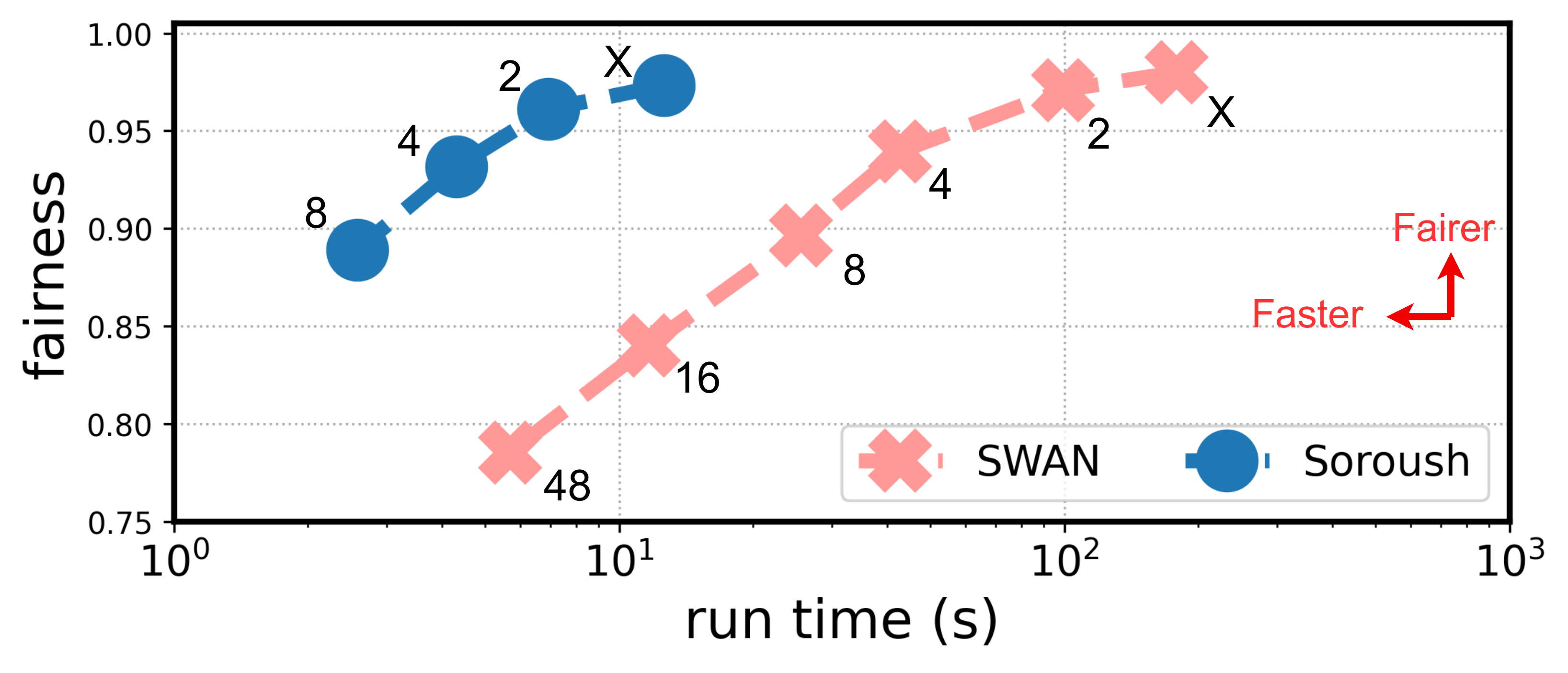}}
	\vspace{-1mm} \\
	\subfigure[Topo: GtsCe, Traffic: Gravity, Scale factor=16, no client splitting]{\includegraphics[width=0.45\linewidth]{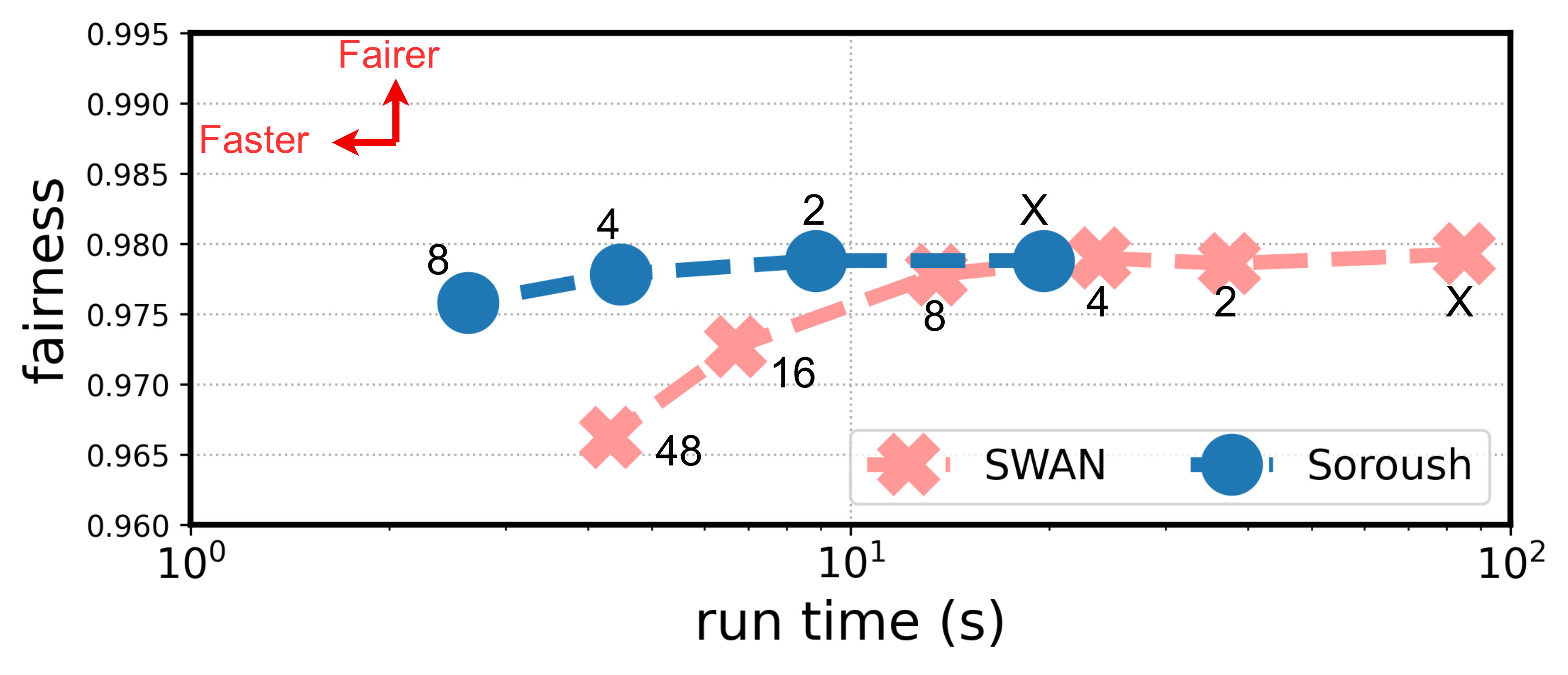}} \hfill
	\subfigure[Topo: GtsCe, Traffic: Gravity, Scale factor=64, no client splitting]{\includegraphics[width=0.45\linewidth]{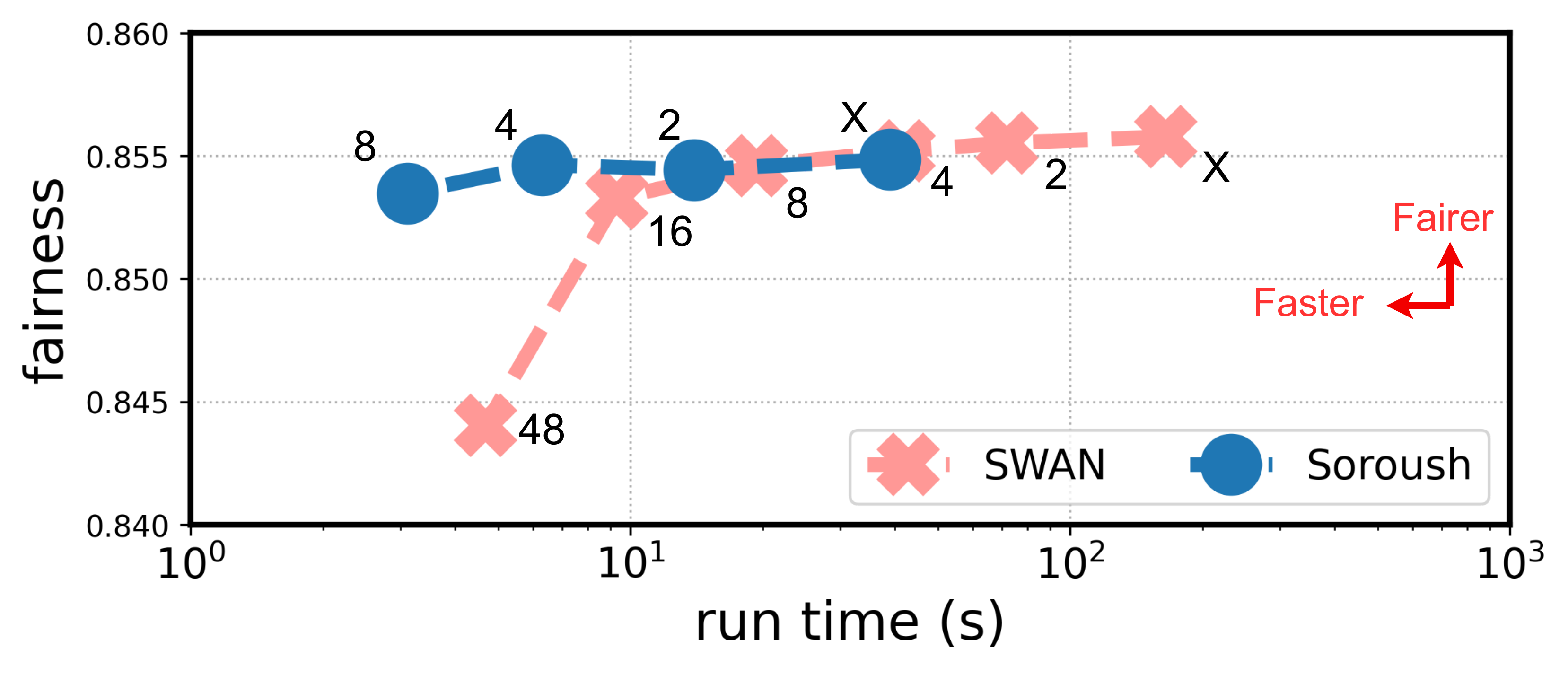}}
	\caption{\textbf{Impact of POP~\cite{Narayanan-POP}.} POP is not designed for max-min fair allocation and can cause drop in fairness depending on the traffic distribution (both on \sysname and SWAN). In contrast, \sysname achieves lower runtime compared to SWAN while maintaining the same level of fairness and theoretical guarantees. The figure reports the average over 3 randomly generated demands. "X" indicates that POP is not used, and the "Numbers" represent the number of POP partitions.}
	\label{fig:decomp:pop:app}
\end{figure*}

%% file: texfiles/waterfillingpseudocodes.tex
\RestyleAlgo{ruled}
\begin{algorithm}[t]
	\DontPrintSemicolon
	\LinesNumbered
	\caption{The waterfilling algorithm to compute sub-flow level max-min fairness.}
	\label{Alg:waterfilling}
	\KwIn{$\routingmatrix_{M \times K}$: sub-flow lvl routing matrix where $\routingmatrix[m,k] = 1$ if sub-flow $k$ uses link $m$, otherwise $=0$.}
	\KwIn{$\setactiveflows_{K \times 1}$: set of (active) sub-flows}
	\KwIn{$\capacity_{M \times 1}$: link capacities.}
	\KwOut{$\rate_{K \times 1}$: max-min rate of each sub-flow.}
	
	$\rate_{K \times 1} \gets 0$ \algcomment{initial rate for each sub-flow.} \;
	\While{$|\setactiveflows| > 0$}{
		$N \gets \routingmatrix \mathbbm{1}_{|\setactiveflows| \times 1}$ \algcomment{number of sub-flows per link.}\;
		$\fairshare \gets \frac{\capacity}{N}$ \algcomment{fairshare of each link.}\;
		$l \gets \mathop{\mathrm{arg\,min}_e} \fairshare_e$ \algcomment{link with minimum fair share}\;
		$\setflows_l \gets \{ f: \routingmatrix[l, f] = 1 \}$ \algcomment{sub-flows on link $l$}\;
		\For{$f \in \setflows_l$}{ 
			$\rate_f \gets \fairshare$ \algcomment{Fix rate of all sub-flows in $\setflows_l$}\;
		}
		$\capacity \gets \capacity - \routingmatrix \rate[\setactiveflows]$ \algcomment{update link capacities}\;
		$\capacity \gets \capacity[\setminus l]$ \algcomment{remove link}\;
		$\routingmatrix \gets \routingmatrix[\setminus l, \setminus \setflows_l]$ \algcomment{remove link and all sub-flows}\;
		$\setactiveflows \gets \setactiveflows \setminus \setflows_l$	\algcomment{update set of active sub-flows}\;
	}
	\Return{$\rate$}
\end{algorithm}

\RestyleAlgo{ruled}
\begin{algorithm}[t]
	\DontPrintSemicolon
	\LinesNumbered
	\caption{The waterfilling algorithm to compute weighted sub-flow level max-min fairness.}
	\label{Alg:weighted-waterfilling}
	\KwIn{$\weightedroutingmatrix_{M \times K}$: sub-flow weights where $\weightedroutingmatrix[m,k]$ shows the weight of sub-flow $k$ on link $m$.}
	\KwIn{$\routingmatrix_{M \times K}$: routing matrix where where $\routingmatrix[m,k] = 1$ if sub-flow $k$ uses link $m$, otherwise $=0$}
	\KwIn{$\setactiveflows$: set of (active) sub-flows}
	\KwIn{$\capacity_{M \times 1}$: link capacities.}
	\KwOut{$\rate_{K \times 1}$: max-min rate of each sub-flow.}
	
	$\rate_{K \times 1} \gets 0$ \algcomment{initial rate for each sub-flow.} \;
	\While{$|\setactiveflows| > 0$}{
		$N \gets \weightedroutingmatrix \mathbbm{1}_{|\setactiveflows| \times 1}$ \algcomment{weight of sub-flows per link.}\;
		$\fairshare \gets \frac{\capacity}{N}$ \algcomment{fairshare of each link.}\;
		$l \gets \mathop{\mathrm{arg\,min}_e} \fairshare_e$ \algcomment{link with minimum fair share}\;
		$\setflows_l \gets \{ f: \routingmatrix[l, f] = 1 \}$ \algcomment{sub-flows on link $l$}\;
		\For{$f \in \setflows_l$}{ 
			$\rate_f \gets \fairshare \weightedroutingmatrix[l, f]$ \algcomment{Fix rate of sub-flows}\;
		}
		$\capacity \gets \capacity - \routingmatrix \rate[\setactiveflows]$ \algcomment{update link capacities}\;
		$\capacity \gets \capacity[\setminus l]$ \algcomment{remove link l}\;
		$\weightedroutingmatrix \gets \weightedroutingmatrix[\setminus l, \setminus \setflows_l]$ \algcomment{remove link and sub-flows}\;
		$\setactiveflows \gets \setactiveflows \setminus \setflows_l$	\algcomment{update set of active flows}\;
	}
\Return{$\rate$}
\end{algorithm}

\RestyleAlgo{ruled}
\begin{algorithm}[t]
	\DontPrintSemicolon
	\LinesNumbered
	\caption{Sub-flow level Approx Max-Min Fair.}
	\label{Alg:approx-max-min-fair}
	\KwIn{$\weightedroutingmatrix_{M \times K}$: sub-flow weights where $\weightedroutingmatrix[m,k]$ shows the weight of sub-flow $k$ on link $m$.}
	\KwIn{$\capacity_{M \times 1}$: link capacities.}
	\KwOut{$\rate_{K \times 1}$: max-min rate of each sub-flow.}
	
	${\rate} \gets \infty_{K \times 1}$ \algcomment{initial max-min rate set to $\infty$}\;
	$N \gets \weightedroutingmatrix \mathbbm{1}_{K \times 1}$ \algcomment{weight of sub-flows per link}\;
	$\set{L} \gets \mathop{\mathrm{arg\,sort}}\frac{\capacity}{N}$ \algcomment{sort links in ascending order}\;
	\For{$l \in \set{L}$}{
		$\setflows_l \gets \{ f: \weightedroutingmatrix[l, f] > 0 \}$ \algcomment{sub-flows on link $l$}\;
		$\hat{\rate} \gets 0_{|\setflows_l|}$ \algcomment{initial fair share for each sub-flow}\;
		\While{$\setflows_l \neq \emptyset$}{
			\tcc{compute link and sub-flow fair shares.}
			$\fairshare \gets \frac{\capacity_l}{\sum \weightedroutingmatrix[l, \setflows_l]}$ \;
			$\hat{\rate}_{\setflows_l} = \fairshare \weightedroutingmatrix[l, \setflows_l]$ \;
			\tcc{compuate set of flows requesting less than fair share}
			$\set{U} \gets \{i:i \in \setflows_l \And \rate_i < \hat{\rate}_i\}$ \;
			\uIf{$\set{U} = \emptyset$}{
				\tcc{if all the sub-flows are requesting more than fair share, set their rate and break.}
				$\rate_{\setflows_l} \gets \hat{\rate}_{\setflows_l}$ \;
				$break$ \; 
			}\Else{
				\tcc{otherwise, remove the sub-flows requesting less and recompute the fair share.}
				$\capacity_l \gets \capacity_l - \sum \rate_{\set{U}}$ \;
				$\setflows_l \gets \setflows_l \setminus \set{U}$\;
			}
		}
	}
	\Return{$\rate$}
\end{algorithm}